\RequirePackage{fix-cm}
\documentclass[smallextended]{svjour3}       
\smartqed  
\usepackage{graphicx}
\usepackage[utf8]{inputenc}
\usepackage{hyperref}
\hypersetup{
    colorlinks=true,
    linkcolor=magenta,
    filecolor=blue,      
    urlcolor=blue,
    citecolor=blue,
    pdftitle={Sharelatex Example},
    }
\usepackage{url}
\usepackage[sort&compress,numbers]{natbib}
\usepackage{amsmath,amssymb}
\usepackage{graphicx}
\usepackage[skins]{tcolorbox}
\newtheorem{mylemma}{Lemma}
\newtheorem{mytheorem}{Theorem}

\newtheorem{mydef}{Definition}
\usepackage{authblk}
\usepackage[procnumbered,linesnumbered,ruled,vlined]{algorithm2e}
\begin{document}

\title{An Efficient Updation Approach for Enumerating Maximal $(\Delta, \gamma)$\mbox{-}Cliques of a Temporal Network \thanks{Dr. Suman Banerjee was supported by the Institute Post Doctoral Fellowship Grant sponsored by Indian Institute of Technology, Gandhinagar. Both the authors have contributed equally in this work.}
}


\author{Suman Banerjee        \and
        Bithika Pal 
}


\institute{Suman Banerjee \at
              Department of Computer Science and Engineering, \\
              Indian Institute of Technology, Gandginagar, India.\\
              \email{suman@iitkgp.ac.in}           
           \and
           Bithika Pal \at
              Department of Computer Science and Engineering,\\ 
              Indian Institute of Technology, Kharagpur, India.\\
              \email{bithikapal@iitkgp.ac.in}  
}

\date{Received: date / Accepted: date}

\maketitle

\begin{abstract}
Given a temporal network $\mathcal{G}(\mathcal{V}, \mathcal{E}, \mathcal{T})$, $(\mathcal{X},[t_a,t_b])$ (where $\mathcal{X} \subseteq \mathcal{V}(\mathcal{G})$ and $[t_a,t_b] \subseteq \mathcal{T}$) is said to be a $(\Delta, \gamma)$\mbox{-}clique of $\mathcal{G}$, if for every pair of vertices in $\mathcal{X}$, there must exist at least $\gamma$ links in each $\Delta$ duration within the time interval $[t_a,t_b]$. Enumerating such maximal cliques is an important problem in temporal network analysis, as it reveals contact pattern among the nodes of $\mathcal{G}$. In this paper, we study the maximal $(\Delta, \gamma)$\mbox{-}clique enumeration problem in online setting; i.e.; the entire link set of the network is not known in advance, and the links are coming as a batch in an iterative manner. Suppose, the link set till time stamp $T_{1}$ (i.e., $\mathcal{E}^{T_{1}}$), and its corresponding $(\Delta, \gamma)$-clique set are known. In the next batch (till time $T_{2}$), a new set of links (denoted as $\mathcal{E}^{(T_1,T_2]}$) is arrived. Now, the goal is to update the existing  $(\Delta, \gamma)$-cliques to obtain the maximal $(\Delta, \gamma)$-cliques till time stamp $T_{2}$. We formally call this problem as the Maximal \emph{$(\Delta, \gamma)$-Clique Updation Problem} for enumerating maximal $(\Delta, \gamma)$-cliques. For this, we propose an efficient updation approach that can be used to enumerate maximal $(\Delta, \gamma)$-cliques of a temporal network in online setting. We show that the proposed methodology is correct and it has been analyzed for its time and space requirement. An extensive set of experiments have been carried out with four benchmark temporal network datasets. The obtained results show that the proposed methodology is efficient both in terms of  time and space to enumerate maximal $(\Delta, \gamma)$-cliques in online setting. Particularly, compared to it's off-line counterpart, the improvement caused by our proposed approach is in the order of hours and GB for computational time and space, respectively, in large dataset.
\keywords{Temporal Network \and $(\Delta, \gamma)$\mbox{-}Clique \and Temporal Link \and Updating Algorithm \and Online Setting}
\end{abstract}

\section{Introduction} \label{intro}
A pairwise relation among a group of agents is represented by a `network' (also known as `graph'), where the set of agents forms the \emph{vertex set}, and the links among them form the \emph{edge set} of the network \cite{barabasi2016network}. Examples include `social network' (interconnected structure among a group human) \cite{musial2013creation}, `information network' (interconnected structure among a group of data centers) \cite{sun2013mining} and many more. Analysis of such networks for different topological structures brings out many important characteristics regarding the contact pattern. For example, a cohesive group of users in a social network can be interpreted as close friends. Most of the real\mbox{-}world networks are time varying in nature, which means the structure of the network are changing over time. This kind of networks are effectively modeled as \textit{temporal network} (also known as the \emph{time varying graph} or \emph{link stream}) \cite{kostakos2009temporal}. 
\par To analyze a static network, there are several topological structures have been defined in the literature, such as \textit{clique} \cite{akkoyunlu1973enumeration} \cite{eppstein2013listing}, \textit{pseudo clique} \cite{zhai2016fast}, \emph{k-plex} \cite{berlowitz2015efficient}, $k$\mbox{-}\emph{club} \cite{almeida2014two}, $k$\mbox{-}\emph{cores} \cite{khaouid2015k} and many more \cite{akiba2013linear}. Among them the widely studied topological structure is clique. Plenty of solution methodologies have been proposed in the literature to enumerate such structures present in the network \cite{xu2013topological}. Initially, E. A. Akkoyunlu \cite{akkoyunlu1973enumeration} proposed an enumeration technique for maximal cliques. After that, a recursive technique has been proposed by Bron and Kerbosch \cite{bron1973algorithm}. Since then, a significant effort has been put to develop practical algorithms for enumerating maximal cliques in different scenarios such as for sparse graph \cite{eppstein2011listing,eppstein2013listing}, in large networks \cite{cheng2010finding,cheng2011finding,rossi2014fast}, in uncertain graphs \cite{mukherjee2015mining,mukherjee2017enumeration,zou2010finding}, using map reduce framework \cite{hou2016efficient,xiang2013scalable}, in parallel computing framework \cite{chen2016parallelizing,rossi2015parallel,schmidt2009scalable}, with limited memory resources \cite{cheng2012fast} and so on. As the real world networks are time varying, none of the mentioned techniques can be applied for analyzing such networks.
\par To analyze a temporal network for cohesive structures, some generalization of `clique' is required. In this direction, the first contribution came from Viard et al. \cite{viard2015revealing,viard2016computing}, who introduced the notion of $\Delta$\mbox{-}Clique. For a given temporal network, a $\Delta$\mbox{-}Clique is a vertex subset and time interval pair, where for every $\Delta$ duration of the interval, there exist at least one link between every pair of vertices in the subset. Recently, Banerjee and Pal \cite{banerjee2019enumeration} extended the notion of $\Delta$\mbox{-}Clique to $(\Delta, \gamma)$\mbox{-}Clique by incorporating frequency along with the duration, and proposed an enumeration algorithm for maximal $(\Delta, \gamma)$\mbox{-}Cliques present in a temporal network. In all these studies, it is implicitly assumed that the whole temporal links are available before the enumeration algorithm starts execution. However in reality, the scenario may be little different which has been explained with a real\mbox{-}world case study.
\par In $2011$ and $2012$, a human contact dataset was collected among a group of high school children in France in a proximity sensing platform based on wearable sensors for the duration of almost $8$ days \cite{fournet2014contact}. The goal was to study and analyze the following: how the mixing pattern happens by re-partition the students into groups? how the gender differences have an impact on contact patterns? how the evaluation of the contact pattern happens in two widely distinct timescales?, and finally how the contact pattern happens in different duration? To address the last question (the broader one), the notion of enumerating maximal $(\Delta, \gamma)$\mbox{-}Cliques can be applied effectively, as it returns the vertex subset and a time interval, where in each $\Delta$ duration of the interval, there must be at least $\gamma$ edges between every pair of vertices in that set. Suppose, once a contact is happening between two students, that information is getting stored in a centralized server. To apply the maximal $(\Delta, \gamma)$\mbox{-}Clique enumeration methodology proposed in \cite{banerjee2019enumeration} requires the completion of the data collection to obtain the entire set of links. However, in many data collection scenarios, the required time is much more. As an example, for the \textit{`college message'} dataset \cite{panzarasa2009patterns}, the duration for collecting the data was $193$ days. In this scenario, instead of waiting for the end of data collection process, it is more practical to start analyzing the contact pattern once a fraction of the dataset is available. Suppose, the first $24$ hours contact details are available and the $(\Delta, \gamma)$\mbox{-}Clique enumeration algorithm from \cite{banerjee2019enumeration} are used to obtain the partial results. Once the next $24$ hours contact history are available, instead of running from the scratch, can we update the previous maximal clique set to obtain the maximal clique set after $48$ hours? This is the problem that we address in this paper.
\par In summary, we study the maximal $(\Delta, \gamma)$\mbox{-}clique enumeration problem in `Online Setting', where the entire link set of the network is not known in advance, and the links are coming as a batch in an iterative manner. The goal here is to update the $(\Delta, \gamma)$\mbox{-}cliques of the links till the previous batch with the links of the current batch to obtain the updated maximal $(\Delta, \gamma)$\mbox{-}cliques. Formally, we call the problem as the \textit{Maximal $(\Delta, \gamma)$\mbox{-}Clique Updation Problem}. Particularly, we make the following contributions in this paper:
\begin{itemize}
\item We introduce a noble \emph{Maximal $(\Delta, \gamma)$\mbox{-}Clique Updation Problem}, where the links are coming as a batch in an iterative manner.
\item For this problem, we propose an efficient methodology to update the existing clique by considering new links in a sequential manner and named it as `edge on clique'.
\item By drawing consecutive arguments, we show that the proposed methodology correctly updates the previous $(\Delta, \gamma)$\mbox{-}cliques.
\item The proposed methodology has been analyzed to understand its time and space requirements.
\item We implement the proposed methodology to perform an extensive set of experiments on four real\mbox{-}world temporal network datasets with two different data partitioning schemes, and show that the updation approach can be effectively used to enumerate maximal $(\Delta, \gamma)$\mbox{-}cliques. 
\item We also show that this methodology can be adopted in the offline setting (when all the links are available before the start of execution) to enumerate all the $(\Delta, \gamma)$\mbox{-}cliques present in a temporal networks by splitting the dataset into parts and then applying the proposed $(\Delta, \gamma)$\mbox{-}clique updation approach.
\end{itemize}   
Rest of the paper is organized as follows: Section \ref{Sec:RW} describes the relevant studies from literature.  Section \ref{Sec:PPD} contains the required preliminary definitions and defines the `Maximal $(\Delta, \gamma)$\mbox{-}Clique Updation Problem' formally. Section \ref{Sec:PA} discusses the proposed methodology. Section \ref{Sec:EE} contains the experimental evaluations of our proposed approach and finally, Section \ref{Sec:CFD} draws the conclusion of this study and gives future directions. 
\section{Related Work} \label{Sec:RW}
This study comes under the broad theme of time varying graph analysis and in particular structural pattern finding in temporal networks. We describe both of them in two consecutive subsections.  
\subsection{Time Varying Graph Analysis}
As most of the real\mbox{-}life networks such as \emph{wireless sensor networks}, \emph{social networks} are time varying in nature, past one decade has witnessed a significant effort in understanding and mining time varying graphs \cite{casteigts2012time}. Several kinds of problems have been studied, and hence, it is not possible here to survey all the results. Here, we present a few fundamental graph problems that have been studied in temporal setting with corresponding literature. First one is the `temporal connectivity problem'. One very fundamental problem studied in the context is finding the shortest paths in temporal graphs \cite{wu2014path}. Another very important problem in the context of temporal graph analysis is the `reachability' and there exist several studies. Basu et al. \cite{basu2014sample} studied the reachability estimation problem in temporal graphs. Wildemann et al. \cite{wildemann2015time} studied the traversal and reachability problem on temporal graphs and derive three classes of temporal traversals from a set of realistic use cases.  Whitbeck et al. \cite{whitbeck2012temporal} introduced the concept of $(\tau,\delta)$ reachability graph for a given time\mbox{-}varying graph and they studied the mathematical properties of this graph and also provided several algorithms for computing such a graph. Following this study there are several works in this direction \cite{casteigts2015efficiently, xu2015network, wu2016reachability}. Huang et al. \cite{huang2015minimum} studied the minimum spanning tree problem on temporal graphs. Another well studied problem on temporal graph analysis is the community detection  \cite{bazzi2016community, he2015fast, rossetti2018community}. Also, there are several theoretical problems studied in the context of temporal graphs such as finding small `separator in temporal graphs' \cite{zschoche2020complexity, fluschnik2020temporal}, travelling salesman problem \cite{michail2016traveling}, Steiner network problem \cite{khodaverdian2016steiner} and many more.
\subsection{Structural Pattern Findings in Temporal Graphs}
For structural pattern mining of time varying graphs, there exists very few studies. As mentioned previously, the concept of `clique' in static graphs has been extended as $\Delta$\mbox{-}clique by Virad et al. \cite{viard2015revealing, viard2016computing} and used to analyze the time varying relationship among a group of school children. Later on Vired et al. \cite{viard2018enumerating} extended their study on temporal clique enumeration on link streams with duration. Recently, Banerjee and Pal \cite{banerjee2019enumeration, banerjee2020first} extended the notion on $\Delta$\mbox{-}clique to $(\Delta, \gamma)$\mbox{-}Clique and their proposed methodology has been applied to analyze three different temporal network datasets. Recently, Molter et al.  \cite{molter2019enumerating} extended the concept of `isolated clique' in the context of temporal networks and proposed \emph{fixed parameter tractable} algorithms to enumerate such maximal cliques. To the best of our knowledge, there does not exist any more literature on structural pattern analysis in the context of temporal graphs. However, there exist other studies for finding different structural patterns other than cliques such as plex \cite{bentert2019listing}, core decomposition \cite{wu2015core}, span cores \cite{galimberti2019span}.
\par In this paper, we study the problem of maximal $(\Delta, \gamma)$\mbox{-}clique enumeration when the entire links of the temporal network are not known at the beginning and links are available after a time gap. Formally, we name this problem as the Maximal $(\Delta, \gamma)$\mbox{-}Clique Updation Problem.

\section{Preliminaries and problem Definitions} \label{Sec:PPD}
Here, we present some preliminary concepts required to understand the main study presented in this paper. Given a set $\mathcal{X}$, $\binom{\mathcal{X}}{2}$ denotes the set of all $2$ element subsets of $\mathcal{X}$. First, we start by defining `temporal network' in Definition \ref{Def:1}.
\begin{mydef}[Temporal Network] \label{Def:1} \cite{holme2012temporal}
	A temporal network (also known as a time varying graph or link stream) is defined as a triplet $\mathcal{G}(\mathcal{V}, \mathcal{E}, \mathcal{T})$, where $V(\mathcal{G})$ and $\mathcal{E}(\mathcal{G})$ ($\mathcal{E}(\mathcal{G}) \subseteq \binom {\mathcal{V}(\mathcal{G})}2 \times \mathcal{T}$) are the vertex and edge set of the network. $\mathcal{T}$ is the time interval during which the network is observed. Throughout the paper, we use $|\mathcal{V}(\mathcal{G})|=n$ and $|\mathcal{E}(\mathcal{G})|=m$. 
\end{mydef}
Basically, a temporal networks is a collection of links of the form $(v_i, v_j, t)$, where $v_i, v_j \in \mathcal{V}(\mathcal{G})$, and $t$ is a timestamp in the time interval $\mathcal{T}$. It signifies, that there was a contact between $v_i$ and $v_j$ at time $t$. In our study, we assume the network is observed in discrete time steps and throughout the observation the vertex set remains fixed, however, the edge set is changing over time. We define $t_{min}=\underset{t}{argmin} \ (v_i,v_j,t) \in E(\mathcal{G})$ for all the vertex pairs, and similarly, $t_{max}=\underset{t}{argmax} \ (v_i,v_j,t) \in \mathcal{E}(\mathcal{G})$. The difference between $t^{max}$ and $t^{min}$ is called the \textit{lifetime of the network} and denoted as $t_{L}$, i.e., $t_{L}=t^{max}-t^{min}$. For any two vertices, $v_i, v_j \in \mathcal{V}(\mathcal{G})$, we say that there exist a static edge between $v_i$ and $v_j$ if $\exists t_{ij} \in [t^{min}, t^{max}]$, such that $(v_i,v_j, t_{ij}) \in \mathcal{E}(\mathcal{G})$. The frequency of a static edge is defined as the number of times it appeared throughout the lifetime of the network, i.e., $f_{(v_iv_j)}=|\{(v_iv_j,t_{ij}): t_{ij} \in [t^{min}, t^{max}] \}|$. We use the term `link' to denote a temporal link and `edge' to denote a static edge.  
\par Given a temporal network $\mathcal{G}(\mathcal{V}, \mathcal{E}, \mathcal{T})$, Virad et al. \cite{viard2015revealing,viard2016computing} introduced the notion of $\Delta$\mbox{-}Clique, which is a natural extension of a clique in a static graph and mentioned in Definition \ref{Def:2}.
\begin{mydef}[$\Delta$\mbox{-}Clique] \label{Def:2} \cite{viard2015revealing} \cite{viard2016computing}
	For a given time period $\Delta$, a $\Delta$\mbox{-}clique of the temporal network $\mathcal{G}$ is a vertex set, time interval pair, i.e., $(\mathcal{X}, [t_a,t_b])$ with $\mathcal{X} \subseteq \mathcal{V}(\mathcal{G})$, $\vert \mathcal{X} \vert \geq 2$ and $[t_a,t_b] \subseteq \mathcal{T}$, such that $\forall v_i,v_j \in \mathcal{X}$ and $t \in [t_a, max(t_b - \Delta, t_a)]$ there is an edge $(v_i, v_j, t_{ij}) \in \mathcal{E}(\mathcal{G})$ with $t_{ij} \in [t, min (t + \Delta, t_b)]$.
\end{mydef}
Recently, the notion of $\Delta$\mbox{-}Clique has been extended by Banerjee and Pal \cite{banerjee2019enumeration} to $(\Delta, \gamma)$\mbox{-}Clique mentioned in Definition \ref{Def:3}.
\begin{mydef}[$(\Delta, \gamma)$\mbox{-}Clique]\label{Def:3}
	For a given time period $\Delta$ and $\gamma \in \mathbb{Z}^{+}$, a $(\Delta, \gamma)$-Clique of the temporal network $\mathcal{G}$ is a vertex set, time interval pair, i.e., $(\mathcal{X}, [t_a,t_b])$ where $\mathcal{X} \subseteq \mathcal{V}(\mathcal{G})$, $\vert \mathcal{X} \vert \geq 2$, and  $[t_a,t_b] \subseteq \mathcal{T}$. Here $\forall v_i,v_j \in \mathcal{X}$ and $t \in [t_a, max(t_b - \Delta, t_a)]$, there must exist $\gamma$ or more number of edges, i.e., $(v_i, v_j, t_{ij}) \in \mathcal{E}(\mathcal{G})$ and $f_{(v_iv_j)} \geq \gamma$ with $t_{ij} \in [t, min (t + \Delta, t_b)]$. It is easy to observe, that a $(\Delta, \gamma)$\mbox{-}clique will be a $\Delta$\mbox{-}clique when $\gamma=1$.
\end{mydef}
\begin{mydef}[Maximal $(\Delta, \gamma)$\mbox{-}Clique] \label{Def:MDG} \label{Def:maximal}
	A $(\Delta, \gamma)$\mbox{-}clique $(\mathcal{X}, [t_a,t_b])$ of the temporal network $\mathcal{G}(\mathcal{V}, \mathcal{E}, \mathcal{T})$ will be maximal if neither of the following is true.
	\begin{itemize}
		\item $\exists v \in \mathcal{V}(\mathcal{G}) \setminus \mathcal{X}$ such that $(\mathcal{X} \cup \{v\}, [t_a,t_b])$ is a $(\Delta, \gamma)$\mbox{-}Clique.
		\item $(\mathcal{X}, [t_a - dt,t_b])$ is a $(\Delta, \gamma)$\mbox{-}clique. This condition is applied only if $t_a - dt \geq t$.
		\item $(\mathcal{X}, [t_a,t_b + dt])$ is a $(\Delta, \gamma)$\mbox{-}clique. This condition is applied only if $t_b + dt \leq t^{'}$.
	\end{itemize}
\end{mydef}

\begin{mydef}[Maximum $(\Delta, \gamma)$\mbox{-}Clique]
	Let $\mathcal{S}$ be the set of all maximal $(\Delta, \gamma)$-cliques of the temporal network $\mathcal{G}(\mathcal{V}, \mathcal{E}, \mathcal{T})$. Now, $(\mathcal{X}, [t_a,t_b]) \in \mathcal{S}$ will be
	\begin{itemize}
		\item temporally maximum if $\forall (\mathcal{Y}, [t_a^{'},t_b^{'}]) \in \mathcal{S} \setminus (\mathcal{X}, [t_a,t_b])$, $t_b-t_a \geq t_b^{'} - t_a^{'}$,
		\item cardinally maximum if $\forall (\mathcal{Y}, [t_a^{'},t_b^{'}]) \in \mathcal{S} \setminus (\mathcal{X}, [t_a,t_b])$, $\vert \mathcal{X} \vert \geq \vert \mathcal{Y} \vert$.
	\end{itemize}
\end{mydef}
Now as mentioned previously, all the links of the temporal network may not be available at the beginning of the execution of the $(\Delta, \gamma)$-clique enumeration algorithm. In this setting, to apply the existing algorithms for  $(\Delta, \gamma)$\mbox{-}clique enumeration, we need to wait till all the links are available. However, the better way to handle this problem is to adopt the updation approach. Assume that $T_{0}$ is the time stamp from which we are observing the network. Now, the first batch of links has appeared till time $T_1$. So, at time stamp $T_1$, without waiting for entire link set, it is desirable to execute existing $(\Delta, \gamma)$\mbox{-}clique enumeration to understand the contact pattern among the entities. Now, the next batch of links has appeared till time stamp $T_2$. At this point, it is always desirable to update the previously enumerated $(\Delta, \gamma)$\mbox{-}cliques to obtain the $(\Delta, \gamma)$\mbox{-}cliques till time stamp $T_2$. Now, when the next set of links appears again, the recently updated cliques has to be updated again and this process will go on. 

Now, it is trivial to observe that the updation procedure is same irrespective of the number of iterations. Hence, in this work we primarily focus to update the cliques till time stamp $T_{1}$ to time stamp $T_{2}$. Finally, we introduce the Maximal $(\Delta, \gamma)$\mbox{-}Clique Updation Problem in Definition \ref{Def:Prob} that we worked in this paper.
\begin{mydef}[Maximal $(\Delta, \gamma)$\mbox{-}Clique Updation Problem] \label{Def:Prob}
Given a temporal network $\mathcal{G}(V, E, \mathcal{T})$ with its $(\Delta,\gamma)$-clique till time stamp $T_{1}$, and the links from time stamp $T_{1}$ to $T_{2}$, the problem of Maximal $(\Delta, \gamma)$\mbox{-}Clique Updation is to update the $(\Delta, \gamma)$-cliques till time stamp $T_{1}$ to obtain the maximal $(\Delta, \gamma)$-cliques till time stamp $T_{2}$.
\end{mydef}

\section{Proposed Solution Approach} \label{Sec:PA}
In this section, we describe the proposed solution approach for updating maximal $(\Delta, \gamma)$\mbox{-}cliques. 
Let, $\mathcal{C}^{T_{1}}$, $\mathcal{C}^{T_{2}}$, and $\mathcal{C}^{T_{2} \setminus T_{1}}$ denote the set of maximal cliques till time $T_{1}$, $T_{2}$, and from $T_1$ to $T_2$, respectively. Also assume that for a clique $(\mathcal{X}, [t_a, t_b])$, the first and last $\gamma$-th occurrence timestamps of a static edge $(u, v)$ within $[t_a,t_b]$ are denoted by $f_{uv}^{\gamma}$ and $l_{uv}^{\gamma}$, respectively, where $u,v \in \mathcal{X}$. Initially, we establish the following important claims.

\begin{mylemma} \label{lemma:1}
 If both of the following conditions are true

 \begin{itemize}
    \item  $\forall (\mathcal{X}, [t_a, t_b]) \in \mathcal{C}^{T_{1}}$ if $\forall (u,v) \in \mathcal{X},  \nexists \ t$ such that $(u,v,t) \in \mathcal{E}{(\mathcal{G})}$, and $t \in [l^1_{uv} , l^\gamma_{uv} + \Delta]$ where $l^\gamma_{uv} +\Delta > T_1$, and
    \item $\forall (\mathcal{X}, [t_a, t_b]) \in \mathcal{C}^{T_{2} \setminus T_{1}}$  if $\forall (u,v) \in \mathcal{X},  \nexists \ t$ such that $(u,v,t) \in \mathcal{E}{(\mathcal{G})}$, and $t \in [f^\gamma_{uv} - \Delta,f^{1}_{uv} ]$ where $f^\gamma_{uv}-\Delta < T_1$.
\end{itemize}

 then $\mathcal{C}^{T_{2}}=\mathcal{C}^{T_{1}} \  \cup \ \mathcal{C}^{T_{2} \setminus T_{1}} \ \cup \ \mathcal{C}_{[T_{1}-\Delta,T_{1}+\Delta]}^{*}$, where $\mathcal{C}_{[T_{1}-\Delta,T_{1}+\Delta]}^{*}$ denotes the set of maximal cliques within the time interval $[T_{1}-\Delta,T_{1}+\Delta]$ which are not contained as a sub clique in any other maximal cliques in $\mathcal{C}^{T_{1}}$ or  $\mathcal{C}^{T_{2} \setminus T_{1}}$.

\end{mylemma}
\begin{proof}
\begin{enumerate}
    \item  Assume that $(\mathcal{X}, [t_a, t_b]) \in \mathcal{C}^{T_{1}}$. If the clique is extended beyond $T_{1}$, then any one or both of the following can happen:
    \begin{enumerate}
        \item The clique $(\mathcal{X}, [t_a, t_b])$ got extended till $t_c$, where $t_c > T_{1}$ and thus the new maximal clique becomes $(\mathcal{X}, [t_a, t_c])$. If this happens then by the definition of $(\Delta, \gamma)$\mbox{-}Clique, between every pair of vertices in $\mathcal{X}$, there must be at least $\gamma$ links in every $\Delta$ duration between $t_a$ and $t_c$. However, in Condition $1$, it is mentioned that for every $u,v \in \mathcal{X}$, there does not exist any link $(u,v,t)$ with  $t \in [l^1_{uv} , l^\gamma_{uv} + \Delta]$. Hence, the clique $(\mathcal{X}, [t_a, t_b])$ can not be extended beyond $t_b$.
        \item  The second thing that can happen is that the new maximal clique $(\mathcal{Y}, [t_a, t_d])$ is splitted out from the clique $(\mathcal{X}, [t_a, t_b])$,  where $\mathcal{Y} \subset \mathcal{X}$ and $t_d > T_{1}$. Now, it is important to observe that this case can only occur if there does not exist any maximal clique $(\mathcal{Y}, [t_a^{'}, t_b^{'}])$ in $\mathcal{C}^{T_{1}}$ with ($t_a^{'} < t_a$ and $ t_b^{'} \geq t_b$) or ($t_a^{'} \leq t_a $ and $t_b^{'} > t_b$). However, in Condition $1$, it is mentioned that for every $u,v \in \mathcal{Y}$, there does not exist any link $(u,v,t)$ with  $t \in [l^1_{uv} , l^\gamma_{uv} + \Delta]$. Hence, the clique $(\mathcal{Y}, [t_a, t_d])$ can not be formed.
    \end{enumerate}
    Both the sub cases together imply that neither the maximal clique $(\mathcal{X}, [t_a, t_b])$ nor any sub clique of this can be extended beyond $T_{1}$.
    \item Assume that $(\mathcal{X}, [t_a, t_b]) \in \mathcal{C}^{T_{2} \setminus T_{1}}$. If the clique is extended before $T_{1}$, then any one or both of the followings can happen:
    \begin{enumerate}
        \item The clique $(\mathcal{X}, [t_a, t_b])$ got extended till $t_c$, where $t_c < T_{1}$ and thus the new maximal clique becomes $(\mathcal{X}, [t_c, t_b])$
        \item  The new maximal clique $(\mathcal{Y}, [t_d, t_b])$ is formed from the clique $(\mathcal{X}, [t_a, t_b])$,  where $\mathcal{Y} \subset \mathcal{X}$ and $t_d < T_{1}$. Now, it is important to observe that this case can only occur if there does not exist any maximal clique $(\mathcal{Y}, [t_a^{'}, t_b^{'}])$ in $\mathcal{C}^{T_{2} \setminus T_{1}}$ with ($t_a^{'} < t_a$ and $ t_b^{'} \geq t_b$) or ($t_a^{'} \leq t_a $ and $t_b^{'} > t_b$).
    \end{enumerate}
    Similar to Case 1, it can be proved that such extension is not possible as there does not exist any link with $t  \in [f^{\gamma}_{uv} - \Delta, f^{1}_{uv}]$ for all $u, v \in \mathcal{X}$.
\end{enumerate}
Figure \ref{fig:lemma1} shows an example scenario of Lemma \ref{lemma:1} with $T_1 = 12$, $\Delta=4$, and $\gamma = 2$. The contents of $\mathcal{C}^{T_1}$, $\mathcal{C}^{T_2 \setminus T_1}$, and $\mathcal{C}_{[T_{1}-\Delta,T_{1}+\Delta]}^{*}$ are shown in colour blue, green, and red, respectively ($\mathcal{C}^{T_1} = \{ (\{v_1, v_2\}, [2, 11]), (\{v_2, v_3\}, [4, 13]), (\{v_3, v_4\}, [1, 9])\}$, $\mathcal{C}^{T_2 \setminus T_1} = \{ (\{v_1, v_2\}, [12, 21]), (\{v_1, v_3\}, [11, 20]), (\{v_1, v_2, v_3\}, [12, 20])\}$, $\mathcal{C}_{[T_{1}-\Delta,T_{1}+\Delta]}^{*} = \{ (\{v_3, v_4\}, [8, 16])\}$).

\end{proof}

\begin{figure}
    \centering
    \includegraphics[scale=1.0]{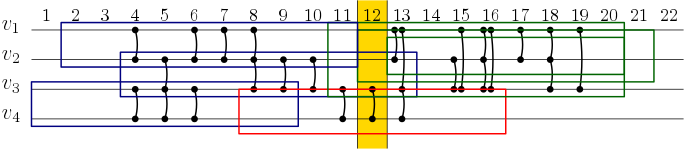}
    \caption{Demonstration for Lemma \ref{lemma:1}}
    \label{fig:lemma1}
\end{figure}


\begin{figure}
    \centering
    \includegraphics[scale=0.1]{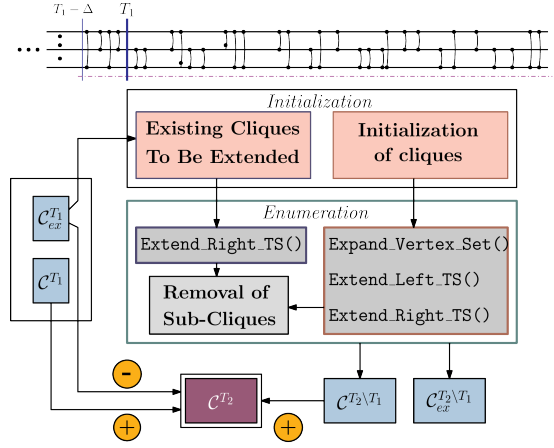}
    \caption{Block Diagram of the Proposed `Edge on Clique' Approach}
    \label{Fig:Demo_edge_on_clique}
\end{figure}

\par Assume that we have the maximal $(\Delta, \gamma)$-cliques for the links till time stamp $T_{1}$. Now, the links $\mathcal{E}^{(T_1, T_2]}$ are just arrived. As mentioned previously, the goal is to enumerate all the maximal $(\Delta, \gamma)$-cliques for the links till time stamp $T_2$. In our proposed methodology, as we update the existing $(\Delta, \gamma)$-cliques with the new set of links, we call our proposed methodology as `Edge on Clique' which goes like this. It takes the maximal clique set till time stamp $T_1$ (i.e.; $\mathcal{C}^{T_{1}}$), the possible clique set to be extended from the previous time stamp ($\mathcal{C}^{T_{1}}_{ex}$), the links arrived in the time duration from $T_1- \Delta$ to $T_2$ ($\mathcal{E}^{ [T_{1}-\Delta, T_{2}]}$), the time stamp till which the maximal cliques are processed ($T_1$), the time stamp up to which the recent links are just arrived ($T_2$),  $\Delta$, and $\gamma$ as inputs. It produces the maximal $(\Delta, \gamma)$-cliques till time stamp $T_{2}$ ($\mathcal{C}^{T_{2}}$), and the cliques that will be extended in the next update (i.e., when the links $\mathcal{E}^{(T_2, T_3]}$  for some $T_{3} > T_{2}$). The proposed method works in three parts; (i) First, it extends the right timestamp of all the cliques coming $\mathcal{C}^{T_1}_{ex}$; (ii) Second, it process the links $\mathcal{E}^{ [T_{1}-\Delta, T_{2}]}$ through initialization, extending the right and left timestamp, expanding the vertex set; (iii) Third, it removes the sub-cliques formed due to in-existence of the links before time $T_1 - \Delta$.  Figure \ref{Fig:Demo_edge_on_clique} demonstrates the proposed `edge on clique' approach.
\begin{algorithm}[h]
	\caption{Initialization Process of the `Edge on Clique' Approach}  \label{Algo:edge_on_clique_ini}
	\label{Algo:1}
	\KwData{ The Clique Set till Time $T_1$; i.e.; $\mathcal{C}^{T_{1}}$, $\mathcal{C}^{T_{1}}_{ex}$, The link set $\mathcal{E}^{ [T_{1}-\Delta, T_{2}]}$ of $\mathcal{G}(\mathcal{V}, \mathcal{E}, \mathcal{T}),  \Delta, \gamma$, $T_1$, \text{ and } $T_2$.}
	\KwResult{Initialized Clique Set for the links $\mathcal{E}^{T_{2} \setminus T_{1}}$. }
	$\mathcal{C}^{I} \longleftarrow  \mathcal{C}^{T_{1}}_{ex}$ \; $\mathcal{C}^{T_{2} \setminus T_{1}} \longleftarrow \emptyset, \ \mathcal{C}^{T_{2} \setminus T_{1}}_{ex} \longleftarrow \emptyset, \ \mathcal{C}_{im} \longleftarrow \mathcal{C}^{I}$\;
	\While{$\mathcal{C}^{I} \neq  \emptyset$}{
	  take and remove ($\mathcal{Z}$, [$t_x,t_y$]) from $\mathcal{C}^{I}$\;
	  ${r\_flag} = $ \texttt{Extend\_Right\_TS(}$(\mathcal{Z}, [t_x,t_y])$, $\mathcal{E}^{ [T_{1}-\Delta, T_{2}]}$\texttt{)}\;
	  \If{$ \text{r\_flag} == TRUE$}{
			add $(\mathcal{Z}, [t_x, t_y])$ to $\mathcal{C}^{T_{2} \setminus T_{1}}$\;
		}
		\If{$t_y \ge T_2$}{
		add $(\mathcal{Z}, [t_x, t_y])$ to $\mathcal{C}^{T_{2} \setminus T_{1}}_{ex}$\;
		  }
	  
	}
	$\mathcal{C}^{I} \longleftarrow  \text{Execute Algorithm 1 of \cite{banerjee2019enumeration} on the links } \mathcal{E}^{ [T_{1}-\Delta, T_{2}]}$\;
	$\mathcal{C}_{im} \longleftarrow \mathcal{C}_{im} \cup \mathcal{C}^{I}$\;
	\While{$\mathcal{C}^{I} \neq  \emptyset$}{
		take and remove ($\mathcal{Z}$, [$t_x,t_y$]) from $\mathcal{C}^{I}$\;
		\If{$t_y - t_x == \Delta$}{
		Prepare the static graph $G$ for the duration [$t_x,t_y$]\;
		Associate $N_{G}(\mathcal{Z})$ to ($\mathcal{Z}$, [$t_x,t_y$])\; 
		}
		${v\_flag} = $ \texttt{Expand\_Vertex\_Set(} $(\mathcal{Z}, [t_x,t_y])$, $N_{G}(\mathcal{Z})$ \texttt{)}\;
	
	    ${l\_flag} = $ \texttt{Extend\_Left\_TS(}$(\mathcal{Z}, [t_x,t_y])$, $\mathcal{E}^{ [T_{1}-\Delta, T_{2}]}$\texttt{)}\;
	    
	    ${r\_flag} = $ \texttt{Extend\_Right\_TS(}$(\mathcal{Z}, [t_x,t_y])$, $\mathcal{E}^{ [T_{1}-\Delta, T_{2}]}$\texttt{)}\;

		\If{$\text{v\_flag} \bigwedge \text{l\_flag} \bigwedge \text{r\_flag} == TRUE$}{
			add $(\mathcal{Z}, [t_x, t_y])$ to $\mathcal{C}^{T_{2} \setminus T_{1}}$\;
		}
		\If{$t_y \ge T_2$}{
		add $(\mathcal{Z}, [t_x, t_y])$ to $\mathcal{C}^{T_{2} \setminus T_{1}}_{ex}$\;
		  }
		}
		\texttt{EOC\_Remove\_Sub\_Cliques(}$T_1$\texttt{)}\;
		
		
		$\mathcal{C}^{T_{2}} \longleftarrow \mathcal{C}^{T_{2} \setminus T_{1}} \cup (\mathcal{C}^{T_{1}} \setminus \mathcal{C}^{T_{1}}_{ex}) $ \;
		
		\textbf{return} $\mathcal{C}^{T_{2}}$, $\mathcal{C}^{T_{2} \setminus T_{1}}_{ex}$\;

\end{algorithm}

\paragraph{Description of the Proposed Approach :} First, we initialize the clique set $\mathcal{C}^{I}$ with the cliques coming from $\mathcal{C}^{T_1}_{ex}$. These cliques are formed by the links of $\mathcal{E}^{T_1}$ which can be extended beyond time stamp $T_{1}$ without violating the properties of $(\Delta, \gamma)$-clique. Hence, we first process these cliques to extend the right timestamp with the links coming till timestamp $T_2$. For the enumeration process, four clique sets: $\mathcal{C}^{I}$ (for holding the cliques yet to be processed), $\mathcal{C}_{im}$ (for keeping the cliques already or yet to be processed), $\mathcal{C}^{T_{2} \setminus T_{1}}$ (for storing the maximal cliques, whose entire or partial links are present in $\mathcal{E}^{[T_1 - \Delta, T_2]}$), and $\mathcal{C}^{T_{2} \setminus T_{1}}_{ex}$ (for storing the cliques to be extended in next update) are maintained. The mentioned clique sets are initialized in Line 2. Here, we highlight that all these clique sets are `global', which means that the subroutines that are invoked in Algorithm \ref{Algo:1} will also have access. Next, in the while loop at Line 3 to 9, we process each of the cliques in $\mathcal{C}^{I}$, till $\mathcal{C}^{I}$ is empty. Next, an arbitrary clique $(\mathcal{Z}, [t_x,t_y])$ is taken out from $\mathcal{C}^{I}$ and tried to expand the right time stamp with \texttt{Extend\_Right\_TS()} Procedure (Procedure \ref{proc:rightts}). If the extension is possible, the new clique is added in $\mathcal{C}_{im}$ and $\mathcal{C}^{I}$, and set the $r\_flag$ as $FALSE$ to indicate $(\mathcal{Z}, [t_x,t_y])$ is non-maximal. Otherwise, $(\mathcal{Z}, [t_x,t_y])$ is maximal and added to $\mathcal{C}^{T_{2} \setminus T_{1}}$ in Line 7 of Algorithm \ref{Algo:1}. Now, if $t_y \geq T_2$ for the current popped clique, it is added into $\mathcal{C}^{T_{2} \setminus T_{1}}_{ex}$ as a possible candidate for the extension. This completes one step towards enumerating the all maximal cliques with the links till $T_2$.

\par Next, we execute the `initialization'  procedure, i.e., Algorithm 1 from \cite{banerjee2019enumeration} on the links $\mathcal{E}^{ [T_{1}-\Delta, T_{2}]}$ and the generated cliques are kept in $\mathcal{C}^{I}$. The cliques in $\mathcal{C}^{I}$ holds the following properties: i)the cardinality of the vertex set is 2, ii)the time interval of each $(\Delta, \gamma)$-clique is of exact duration $\Delta$, and iii) each clique has exactly $\gamma$ links within the time interval (Follows from Lemma 1 in \cite{banerjee2019enumeration}). In Line 10, the initialized cliques set $\mathcal{C}^{I}$ is the complete set for the enumeration process with the links $\mathcal{E}^{ [T_{1}-\Delta, T_{2}]}$. The correctness of the initialization can be proved by lemma 6 of \cite{banerjee2019enumeration}. In Line 11, $\mathcal{C}_{im}$ is updated with the cliques of $\mathcal{C}^{I}$ to start the next step of the enumeration process.

\par Next, an arbitrary clique $(\mathcal{Z}, [t_x,t_y])$ is taken out from $\mathcal{C}^{I}$ in Line $13$. If $t_y - t_x=\Delta$, the static graph $G$ is built with the links present within $[t_x, t_y]$ in Line 15. The neighboring vertices in $G$, which have communicated at least $\gamma$ times with any of the vertices in $\mathcal{Z}$, forms the candidate vertex set $\mathcal{N}_G(\mathcal{Z})$ and associated with the clique $(\mathcal{Z}, [t_x,t_y])$ in Line 16. Next, the algorithm tries to expand $(\mathcal{Z}, [t_x,t_y])$ in the following three ways: (i) by adding vertices (invoking the function \texttt{Expand\_Vertex\_Set()}), (ii) by stretching $t_x$ towards its left (invoking the function \texttt{Extend\_Left\_TS()}), and (iii) by stretching $t_y$ towards its right (invoking the function \texttt{Extend\_Right\_TS()}). In \texttt{Expand\_Vertex\_Set()}, it selects the neighbors of $\mathcal{Z}$ from the static graph $G$ within the interval $[t_x, t_y]$ and checks whether the new tuple(vertex pair, time interval) holds the $(\Delta, \gamma)$-clique property. If the vertex addition is possible, the new clique is added in $\mathcal{C}_{im}$ and $\mathcal{C}^{I}$, and it declares $(\mathcal{Z}, [t_x,t_y])$ as non-maximal by setting $v\_flag$ to $FALSE$. In \texttt{Extend\_Left\_TS()}, it tries to extend $t_x$ by $t_{xl} - \Delta$ ($t_{xl}$ is the latest time stamp from all the $\gamma$-th occurrence time stamps  within $[t_x-1, t_y]$ of each possible vertex pair in $\mathcal{Z}$) with the links $\mathcal{E}^{ [T_{1} - \Delta, T_{2}]}$. If the new tuple(vertex pair, time interval) holds the $(\Delta, \gamma)$-clique property, the new clique is added in $\mathcal{C}_{im}$ and $\mathcal{C}^{I}$, and it declares $(\mathcal{Z}, [t_x,t_y])$ as non-maximal by setting $l\_flag$ to $FALSE$. Similarly, in \texttt{Extend\_Right\_TS()}, $t_y$ is extended as $t_{yr} + \Delta$ ($t_{yr}$ is the earliest time stamp from all the last $\gamma$-th occurrence time stamps within $[t_x, t_y + 1]$ of each possible vertex pair in $\mathcal{Z}$) and sets $r\_flag$ to $FALSE$ if the extension is possible. If any of the function returns $FALSE$ that means the clique $(\mathcal{Z}, [t_x,t_y])$ is not maximal. Hence, we perform logical `AND' operation among the flags in Line 20, and if the outcome is $TRUE$, that means the clique $(\mathcal{Z}, [t_x,t_y])$ is a maximal clique, and it is added to $\mathcal{C}^{T_{2} \setminus T_{1}}$. Next, we verify whether the current clique $(\mathcal{Z}, [t_x, t_y])$ can be extended for the next update cycle in Line 22. If $t_y$ is greater than $T_2$, it is included in  $\mathcal{C}^{T_{2} \setminus T_{1}}_{ex}$. The correctness of this condition is given in Lemma \ref{lemma:extendend}. This process is repeated until $\mathcal{C}^{I}$ is empty.

Now, it is important to observe that the maximal clique set $\mathcal{C}^{T_{2}}$ can be obtained by the union of $(\mathcal{C}^{T_{1}} \setminus \mathcal{C}^{T_{1}}_{ex})$, and $\mathcal{C}^{T_{2} \setminus T_{1}}$. However, in doing so we may end up with getting some non\mbox{-}maximal cliques as well. So, it is important to remove such cliques to obtain the final clique set. Hence, the \texttt{EOC\_Remove\_Sub\_Cliques()} function is invoked. For a fixed $\mathcal{Z}$, it tries to get the maximum duration, and for a fixed $[t_x, t_y]$, it tries to keep the maximum number of vertices in $\mathcal{Z}$. In this way, it ends up with getting the maximal cliques only. Now, the complete maximal clique set, $\mathcal{C}^{T_{2}}$, is computed in Line 25. Finally, the algorithm returns the set of maximal cliques till $T_2$ ($\mathcal{C}^{T_{2}}$), and the possible cliques to be extended in next update cycle ($\mathcal{C}^{T_{2} \setminus T_{1}}_{ex}$), as output.

We highlight that for the first cycle there does not exist any prior clique sets and Algorithm \ref{Algo:1} will execute on the links $\mathcal{E}^{T_{1}}$. Hence, we redefine the changed inputs for the first cycle and introduce the notation $T_0$ for the simplicity in understanding. Here, we make $T_1$ of the second update cycle, as $T_0$. Similarly, we replace $T_2$ with $T_1$. So, the inputs of Algorithm \ref{Algo:1} become the previous cliques sets $\mathcal{C}^{T_{0}}$ and $\mathcal{C}^{T_{0}}_{ex}$ as $\emptyset$, the link set $\mathcal{E}^{T_{1}}$, $T_0$ as `$-1$', $T_1$, $\Delta$, and $\gamma$. Hence, the while loop from line 3 to 9 will not execute and the algorithm will run as the procedure described in \cite{banerjee2019enumeration} with the extra operations in line 22 and 23 to get the $\mathcal{C}^{T_{1} \setminus T_{0}}_{ex}$. The first update cycle will not require any removal of non-maximal cliques. So, it will return from the if condition in line 2 of the function \texttt{EOC\_Remove\_Sub\_Cliques(}$T_0 = -1$\texttt{)}. The line 25 in Algorithm \ref{Algo:1} will copy the entire maximal clique set $\mathcal{C}^{T_{1} \setminus T_{0}}$ to $\mathcal{C}^{T_{1}}$. Finally, it will return $\mathcal{C}^{T_{1}}$, and $\mathcal{C}^{T_{1} \setminus T_{0}}_{ex}$ as the outputs of the first cycle. The clique set $\mathcal{C}^{T_{1} \setminus T_{0}}_{ex}$ will be used as $\mathcal{C}^{T_{1}}_{ex}$ for the next update cycle. Now, based on the working principal of Algorithm \ref{Algo:1}, we have the following important observation which has been used to prove Lemma \ref{lemma:4}.

\begin{tcolorbox} \label{note:1}
\begin{itemize}
    \item \textbf{Note:} The clique extending procedures in Algorithm \ref{Algo:1} (i.e., \texttt{Expand\_Vertex\_Set()}, \texttt{Extend\_Left\_TS()}, and \texttt{Extend\_Right\_TS()}), are independent of their order of execution.
\end{itemize}
\end{tcolorbox}


\begin{procedure}[htb]
\caption{Expanding a clique by vertex addition()} \label{proc:nodeadd}
	\SetKwFunction{FMain}{Expand\_Vertex\_Set}
	  \SetKwProg{Pn}{Function}{:}{\KwRet}
      \Pn{\FMain{ $(\mathcal{Z}, [t_x, t_y])$, $N_G(\mathcal{Z})$}}{
       $\text{flag} = TRUE$\;
       \For{\text{All} $u \in N_G(\mathcal{Z}) \setminus \mathcal{Z}$}{
			\If{$(\mathcal{Z} \cup \{u\}, [t_x, t_y])$ is a $(\Delta, \gamma)$\mbox{-}clique}{
				$\text{flag} = FALSE$\;
				\If{$(\mathcal{Z} \cup \{u\}, [t_x, t_y]) \notin \mathcal{C}_{im}$}{
					add $(\mathcal{Z} \cup \{u\}, [t_x, t_y])$ to $\mathcal{C}^I$ and $\mathcal{C}_{im}$\;
				}	
			}
		}
         \KwRet \text{flag}\;
     }
\end{procedure}

\paragraph{Procedure: \texttt{Expand\_Vertex\_Set()}} This procedure takes a $(\Delta, \gamma)$\mbox{-}clique from $\mathcal{C}^{I}$ and its associated candidate vertex set for expanding the clique, as inputs, and returns a Boolean flag indicating whether the input clique is maximal or not. Now for a clique $(\mathcal{Z}, [t_x, t_y])$, we verify whether it is possible to add a vertex from $N_{G}(\mathcal{Z})$ to $\mathcal{Z}$, such that $\mathcal{Z} \cup \{u\}$ holds the $(\Delta, \gamma)$-clique property within $[t_x, t_y]$ in Line 4. If so, the flag is set to $FALSE$, and it is added to $\mathcal{C}^{I}$ and $\mathcal{C}_{im}$ if the new clique does not exists to $\mathcal{C}_{im}$. 

\begin{mylemma} \label{lemma:t1-delta}
For updating $\mathcal{C}^{T_1}$ to obtain $\mathcal{C}^{T_2}$, it is enough to have the links for the duration $[T_1 - \Delta, T_2]$ . 
\end{mylemma}
\begin{proof}
We have to show that it is enough to process $\mathcal{E}^{[T_1 - \Delta, T_2]}$ to get $\mathcal{C}^{T_2}$, along with the inputs $\mathcal{C}^{T_1}$ and $\mathcal{C}^{T_1}_{ex}$. To prove the statement, we need to show the following two points:
\begin{enumerate}
    \item It does not miss any clique to get the maximal cliques.
    \item It does not carry any redundant processing while building $\mathcal{C}^{T_2 \setminus T_1}$. 
\end{enumerate}
During the enumeration process, the initialized cliques are of size $2$, time duration is of $\Delta$. It holds exact $\gamma$ occurrences of the vertex pair within the time duration. They are also associated with their respective candidate set of nodes as the neighbors of the vertex pair in that $\Delta$ duration. Now, the possibility of the extension of a clique involves three directions, as discussed, (vertex addition, extending its start time to the left, extending the end time stamp to the right). In this process, it merges different initialized cliques to the extended one, till it reaches the maximality. In Algorithm \ref{Algo:1}, the cliques from $ \mathcal{C}^{T1}_{ex}$ are extended in the right time stamp only. Now, we know that the smallest possible $t_y$ for extension from $\mathcal{C}^{T_1}$ side to $\mathcal{C}^{T_2 \setminus T_1}$, is $T_1$. Assume, a clique $(\mathcal{Z}, [t_x, t_y]) \in \mathcal{C}^{T1}_{ex}$ with $t_y = T_1$ and $t_{yr}=T_1 - \Delta$. To extend $t_y$ in the right, it needs to observe the existence of $\gamma$ links in $[t_{yr} + 1, t_{yr} + 1 + \Delta  ]$, i.e, $[T_1 - \Delta + 1, T_1 + 1]$. This requires to have the links at least from the timestamp $T_1 - \Delta$. In Procedure \ref{proc:rightts}, it is calculated in Line 3. So, it is evident to have $\mathcal{E}^{[T_1 - \Delta, T_2]}$ to get the $\mathcal{C}^{T_2}$.
\par Another case can arise here is that the $\gamma$ edges of a vertex pair are separated by the partition at $T_1$. Hence, the cliques involving those vertices will not be initialized while building $\mathcal{C}^{T1}$. Hence, those cliques do not exists in $\mathcal{C}^{T1}_{ex}$. To check this possibility, for every new link in $(T1, T1+\Delta]$ it needs to verify the existence of $\gamma$ edges in the corresponding previous $\Delta$ duration for that vertex pair and initialize the $(\Delta, \gamma)$-cliques while building $\mathcal{C}^{T_2 \setminus T_1}$. So, it is necessary to have the set of links $\mathcal{E}^{[T_1 - \Delta, T_2]}$ in computing $\mathcal{C}^{T_2}$.
\par Now any other cliques initialized with the overlapping links are redundant. Let's assume, the set of links used for preparing $\mathcal{C}^{T_2 \setminus T_1}$ is $\mathcal{E}^{[T_1 - \Delta - k, T_2]}$, where $k > 0$. Here, the links between $T_1 - \Delta - k$ to $T_1$ are processed twice, once, while preparing $\mathcal{C}^{T_1}$ and another for $\mathcal{C}^{T_2 \setminus T_1}$. Now, we already shown that the links in $[T_1 - \Delta, T_1]$ are necessary to avoid the missing initialized cliques. Hence, adding more links from the left of $T_1-\Delta$ is redundant and $\mathcal{E}^{[T_1 - \Delta, T_2]}$ is sufficient to enumerate all the maximal cliques till $T_2$ by the `Edge on Cliques' method.
\end{proof}
\textbf{Note:} If any dataset follows the condition of Lemma \ref{lemma:1}, the overlapping links can be ignored.
\begin{mylemma} \label{lemma:3}
The candidate vertex set associated with each $(\Delta, \gamma)$-clique of $\mathcal{C}^{I}$, is complete for the final maximal clique set.
\end{mylemma}
\begin{proof}
To prove the statement by contradiction, we have to show for a clique $(\mathcal{Z}, [t_x, t_y])$, it's associated candidate vertex set is  incomplete. In Algorithm \ref{Algo:1}, the associated candidate set $N_{G}(\mathcal{Z})$ defines the set of vertices which are possible to be added in $\mathcal{Z}$ for vertex set expansion of the clique. Now, the $N_{G}(\mathcal{Z})$ contains the set of vertices where each has to be the neighbor (connected at least $\gamma$ times within $[t_x, t_y]$) of at least one vertex of $\mathcal{Z}$ in the duration of $t_x$ to $t_y$. It is formed during the initialization of the clique when $\vert \mathcal{Z} \vert = 2$ and $t_y - t_x = \Delta$, and propagated further in it's each expansion. Now, we need to show that the $N_{G}(\mathcal{Z})$ is complete. 
\par Assume, a vertex $u \notin N_{G}(\mathcal{Z})$ can be added to $\mathcal{Z}$ such that $(\mathcal{Z} \cup \{u\}, [t_x, t_y])$ forms a $(\Delta, \gamma)$-clique. So, $u$ has to be the neighbor of all the vertices from $\mathcal{Z}$. Let us assume, one such vertex is $v$. As $v \in \mathcal{Z}$, there has to be $\gamma$ edges of the vertex pair $(v, w)$ in $\Delta$ duration with $w \in \mathcal{Z}$, $[t_{x}^{'}, t_{y}^{'}] \subseteq [t_x, t_y]$, and  $t_{y}^{'} - t_{x}^{'} = \Delta$. Hence, $(\{v, w\}, [t_{x}^{'}, t_{y}^{'}])$ is a $(\Delta, \gamma)$-clique and can be one of the initialization to get $(\mathcal{Z}, [t_x, t_y])$. Hence, $u$ has to be in the candidate set of $(\{v, w\}, [t_{x}^{'}, t_{y}^{'}])$, which has to be carry forwarded in $N_{G}(\mathcal{Z})$. Hence, $N_{G}(\mathcal{Z})$ is complete.
\end{proof}
\begin{mylemma} \label{Lemma:5}
The candidate vertex set associated with each $(\Delta, \gamma)$-clique of $\mathcal{C}^{I}$, is correct for the final maximal clique set.
\end{mylemma}
\begin{proof}
In the proof of Lemma \ref{lemma:3}, it has been discussed that for any clique in $\mathcal{C}^I$, the candidate vertex set is associated at the time of it's initialization, and propagated further (without modification) to the new extended cliques, till it reaches it's maximality. To prove that the candidate vertex set is correct, we need to show that it does not miss any maximal $(\Delta, \gamma)$-clique in $\mathcal{C}^{T_{2}}$. Let's assume, $(\mathcal{Z}, [t_x, t_y]) \notin \mathcal{C}^{T_{2}}$ is a maximal $(\Delta, \gamma)$-clique. Now, there can be either of the following two situations.
\begin{itemize}
    \item \textbf{Case 1:} $[t_x, t_y] \subseteq [T_1 - \Delta, T_2]$
    \item \textbf{Case 2:} $t_x < T_1 - \Delta$ and $t_y \geq T_1$
\end{itemize}
\par For Case 1, as $(\mathcal{Z}, [t_x, t_y])$ is a $(\Delta, \gamma)$-clique, so all the vertices have to be linked at least $\gamma$ times in each $\Delta$ duration within $[t_x, t_y]$. So, there exist $ \binom{\vert \mathcal{Z} \vert}{2}$ possible vertex sets for the initialized cliques within $[t_x, t_y]$. Now, for any such initial clique $A_{0} = (\mathcal{Z}^{'}, [t_{x}^{'}, t_{y}^{'}])$ in $\mathcal{C}^{I}$, such that $\mathcal{Z}^{'} \subseteq \mathcal{Z}$, $\vert \mathcal{Z}^{'} \vert =2$, $[t_{x}^{'}, t_{y}^{'}] \subseteq [t_x, t_y]$, and $t_{y}^{'} - t_{x}^{'} = \Delta$, the associated candidate set $\mathcal{N}_G(\mathcal{Z}^{'})$ has to contain all the vertices of $\mathcal{Z}$. This is possible as all the vertices are connected in the static graph $G$ generated in $[t_{x}^{'}, t_{y}^{'}]$ with the link set $\mathcal{E}^{[T_1-\Delta, T_2]}$. Now, by executing Procedure \ref{proc:nodeadd} with $A_{0}$, it will generate all the cliques of size 3 as $A_{1} = (\mathcal{Z}^{''}, [t_{x}^{'}, t_{y}^{'}])$. Next, repeating this process it will form the sequence as $A_{0} \longrightarrow A_{1} \longrightarrow A_{2}, \longrightarrow \ldots \longrightarrow, A_{k-2}$, where $A_{k-2} = (\mathcal{Z}, [t_x^{'}, t_y^{'}])$ with $\vert \mathcal{Z} \vert = k$. Now, as per Lemma 7 of \cite{banerjee2019enumeration}, $(\mathcal{Z}, [t_x, t_y])$ will be obtained from $(\mathcal{Z}, [t_x^{'}, t_y^{'}])$. So, $(\mathcal{Z}, [t_x, t_y])$ will be in $\mathcal{C}^{T_2 \setminus T_1}$. Also, as the clique is maximal, it will not be removed. Hence, $(\mathcal{Z}, [t_x, t_y])$ will be present in $\mathcal{C}^{T_{2}}$.
\par For Case 2, the clique $(\mathcal{Z}, [t_x, t_y])$ has to be emerged from a clique initialized through $\mathcal{C}^{T_1}_{ex}$. Assume, there exist a clique $B_{0} = (\mathcal{Z}^{'}, [t_{x}^{'}, t_{y}^{'}])$ with $[t_x^{'}, t_y^{'}] \subset [T_0, T_1]$, $\vert \mathcal{Z}^{'} \vert =2$, and $t_{y}^{'} - t_{x}^{'} = \Delta$. Similar to Case 1, it can be shown that $B_{0}$ forms $B_{k-2} = (\mathcal{Z}, [t_{x}^{'}, t_{y}^{'}])$, while building $\mathcal{C}^{T_1}$. Now, as per Lemma 7 of \cite{banerjee2019enumeration}, $B_{k-2}$ reaches to it's maximal $B_{k} = (\mathcal{Z}, [t_{x}, t_{y}^{''}])$ with $\mathcal{E}^{T_1}$, where $t_y^{''} \geq T_1$. $B_{k}$ is added in $\mathcal{C}^{T_1}_{ex}$. Next, the cliques in $\mathcal{C}^{T_1}_{ex}$ are expanded in the right timestamp only to reach $t_y$. Hence, $(\mathcal{Z}, [t_x, t_y])$ will be present in $\mathcal{C}^{T_2}$. This completes the proof of the lemma statement.
\end{proof}
Together Lemma \ref{lemma:3}, and \ref{Lemma:5} imply that the candidate vertex set associated with each of the cliques of $\mathcal{C}^{I}$ is correct and complete.
\begin{mylemma} \label{lemma:4}
There exist a maximal clique $(\mathcal{Z}, [t_x, t_y]) \in \mathcal{C}^{T_2}$, such that $t_x < T_1$ and $t_y \geq T_1$, iff there exist a clique $(\mathcal{Z}, [t_x, t_{y}^{'}]) \in \mathcal{C}^{T_1}_{ex}$ with $t_{y}^{'} \leq t_{y}$.
\end{mylemma}
\begin{proof}
First, we prove the forward direction of the lemma statement, i.e., if $(\mathcal{Z}, [t_x, t_{y}^{'}]) \in \mathcal{C}^{T_1}_{ex}$, then $(\mathcal{Z}, [t_x, t_y]) \in \mathcal{C}^{T_2}$. It leads to make the conclusion that the right timestamp extension is correct in Procedure \ref{proc:rightts} and possible while building $\mathcal{C}^{T_2 \setminus T_1}$. Now, given the links that appeared from the last $\Delta$ duration of any clique, the correctness of Procedure \ref{proc:rightts} is self-explanatory. Also, the possibility of extension for the cliques that are coming from $\mathcal{C}^{T_1}_{ex}$ is ensured in Lemma \ref{lemma:t1-delta}. Hence, $(\mathcal{Z}, [t_x, t_y])$ will be in $\mathcal{C}^{T2 \setminus T_1}$.  Now, as $(\mathcal{Z}, [t_x, t_y])$ is maximal, it will not be contained in any other cliques either by vertex sets or in the time interval. So, it has to be present in $ \mathcal{C}^{T_2}$.
\par Next, we prove the reverse direction , i.e., if $(\mathcal{Z}, [t_x, t_y]) \in \mathcal{C}^{T_2}$, then $(\mathcal{Z}, [t_x, t_{y}^{'}]) \in \mathcal{C}^{T_1}_{ex}$. Assume that, we have the entire link set till time stamp $T_2$. Then by Lemma 7 of \cite{banerjee2019enumeration}, for $(\mathcal{Z}, [t_x, t_y])$ there must exist a clique $A_1= (\mathcal{Z}, [t_x^{'}, t_x^{'} + \Delta])$, where $t_{x}^{'} \in [t_x, t_y]$. In Lemma \ref{lemma:3} and \ref{Lemma:5}, we have already shown that the candidate vertex set associated with each $(\Delta, \gamma)$-clique is correct and complete. So, for the duration $[t_{x}^{'}, t_{x}^{'} + \Delta]$ there can be $\binom{| \mathcal{Z} |}{2}$ possible initialized cliques, and each of them will have the candidate set containing all the vertices of $\mathcal{Z}$. Let's say, one such clique is $A_0= (\{u,v\}, [t_x^{'}, t_x^{'} + \Delta])$. By executing only Procedure \ref{proc:nodeadd}, the initialized clique $A_0$ will generate $A_1= (\mathcal{Z}, [t_x^{'}, t_x^{'} + \Delta])$. As highlighted in the description of Algorithm \ref{Algo:1} that the order of execution of the enumeration process is irrelevant, we make the following arguments.
By Procedure \ref{proc:leftts}, $A_1$ will be emerged to $A_2 = (\mathcal{Z}, [t_x, t_x^{'} + \Delta])$ by extending it's left time stamp. Next, $A_2$ will be extended to $(\mathcal{Z}, [t_x, t_y])$ by Procedure \ref{proc:rightts}. This extension will verify the $(\Delta, \gamma)$-clique property in every last $\Delta$ duration and cross through a clique state $A_3 = (\mathcal{Z}, [t_x, t_y^{'}])$. 
Now, while processing $\mathcal{C}^{T_1}$, $A_2$ will not be able to reach $(\mathcal{Z}, [t_x, t_y])$. However, $A_3$ must be reached from $A_2$, while building $\mathcal{C}^{T_1}$. Hence, it will be included in $\mathcal{C}^{T_1}_{ex}$ as $t_y \geq T_1$. Thus, the extension of $A_3$ will be performed in building $\mathcal{C}^{T_2 \setminus T_1}$. So, $\mathcal{C}^{T_1}_{ex}$ will contain the clique $(\mathcal{Z}, [t_x, t_{y}^{'}])$. This concludes the proof of the lemma statement.
\end{proof}

\begin{mylemma} \label{lemma:6}
It is sufficient to expand only the right timestamp of a clique in $\mathcal{C}^{T_1}_{ex}$ to build the final maximal clique set $\mathcal{C}^{T_2}$.
\end{mylemma}
\begin{proof}
To prove this statement, we need to show that Algorithm \ref{Algo:1} does not miss any maximal cliques in $\mathcal{C}^{T_2}$, whose left timestamp is less than $T_1$ and right time stamp is greater than or equal to $T_1$. We prove this statement by contradiction. Let's assume, $(\mathcal{Z}, [t_x, t_y]) \notin \mathcal{C}^{T_{2}}$ is a maximal $(\Delta, \gamma)$-clique, and $t_x < T_1$ and $t_y \geq T_1$.
\par Now, $\mathcal{C}^{T_1}_{ex}$ contains all the cliques whose $t_y \geq T_1$ and formed while running Algorithm \ref{Algo:1} with the link set $\mathcal{E}^{T_1}$. Hence, we need to show that there must be a clique $(\mathcal{Z}, [t_x, t_y^{'}]) \in \mathcal{C}^{T_1}_{ex}$ with $t_{y}^{'} \leq t_y$, whose right time stamp extension is sufficient to get the maximal clique $(\mathcal{Z}, [t_x, t_y])$. According to Lemma \ref{lemma:4}, such $(\mathcal{Z}, [t_x, t_y^{'}])$ exists in $\mathcal{C}^{T_1}_{ex}$. Alternatively, $(\mathcal{Z}, [t_x, t_y])$ belongs to $\mathcal{C}^{T_2}$. Hence, the statement is proved.
\end{proof}

\begin{procedure}[htb]
\caption{Extending a clique towards right in the time horizon()} \label{proc:rightts}
     \SetKwFunction{FRight}{Extend\_Right\_TS}
	  \SetKwProg{Pn}{Function}{:}{\KwRet}
      \Pn{\FRight{$(\mathcal{Z}, [t_x,t_y])$, $\mathcal{E}^{ [T_{1}-\Delta, T_{2}]}$}}{
           $\text{flag} = TRUE$\;
		$t_{yr} = min_{u,v \in \mathcal{Z}} \ t_{yuv}$  \tcp*{last $\gamma^{th}$ occurrence time of an edge $(u,v)$ within $[t_x, t_y+1]$}
		\If{$t_{yr}+\Delta > t_y$}{
			$\text{flag} = FALSE$\;
			\If{$(\mathcal{Z}, [t_x, t_{yr}+\Delta]) \notin \mathcal{C}_{im}$}{
				add $(\mathcal{Z}, [t_x, t_{yr}+\Delta])$ to $\mathcal{C}^{I}$ and $\mathcal{C}_{im}$\;
			}
		}
		\KwRet \text{flag}\;
      }
\end{procedure}
\paragraph{Procedure \texttt{Expand\_Right\_TS()}: } This procedure takes a $(\Delta, \gamma)$\mbox{-}clique from $\mathcal{C}^{I}$ and the edge list for the current update cycle, as inputs, and returns a Boolean flag indicating whether the input clique is maximal or not. For a clique $(\mathcal{Z}, [t_x, t_y])$, the trivial way of extending $t_y$ is as follows: For every pair of vertices $u,v \in \mathcal{Z}$, the last $\gamma$-th occurrence time stamp within $[t_x, t_y + 1]$ is $t_{yuv}$. If the resultant after adding $\Delta$ to the earliest of all $t_{yuv}$ is more than $t_y$, then $(\mathcal{Z}, [t_x, t_y])$ is not maximal and the new clique $(\mathcal{Z}, [t_x, t_{yr} + \Delta])$ is formed. Now, for the cliques which are initialized in the current update cycle or whose time interval $[t_x, t_y]$ falls within $[T_1 - \Delta, T_2]$, it is easy get $t_{yr}$ with the current set of links. Similar to Lemma \ref{lemma:t1-delta}, it can be concluded that for the cliques initialized by $\mathcal{C}^{T_1}_{ex}$, it is possible to get the $t_{yr}$ values with the links $\mathcal{E}^{ [T_{1}-\Delta, T_{2}]}$. This results in the following theorem.
\begin{mytheorem}
 Procedure \ref{proc:rightts} is able to extend any $(\Delta, \gamma)$-clique using the set of links in the current update cycle only, irrespective of it's initialization. 
\end{mytheorem}
\begin{proof}
In the theorem statement, `irrespective of it's initialization' refers two cases: i) cliques coming from $\mathcal{C}^{T_1}_{ex}$, and ii) cliques initialized with link set $\mathcal{E}^{[T_1- \Delta, T_2]}$. For Case 1, it is already shown in lemma \ref{lemma:t1-delta}, that the links from $T_1 - \Delta$ is sufficient to extend the right  timestamp of the cliques in $\mathcal{C}^{T_1}_{ex}$. Now, for the later case, Lemma 6 and 7 of \cite{banerjee2019enumeration} together prove the correctness of the initialized cliques and it's extension in right time stamp. Hence, the statement of the theorem is proved. 
\end{proof}
\begin{procedure}[htb]
\caption{Extending a clique towards left in the time horizon()}  \label{proc:leftts}
     \SetKwFunction{FM}{Extend\_Left\_TS}
	  \SetKwProg{Pn}{Function}{:}{\KwRet}
      \Pn{\FM{$(\mathcal{Z}, [t_x,t_y])$, $\mathcal{E}^{ [T_{1}-\Delta, T_{2}]}$}}{
           $\text{flag} = TRUE$\;
           $t_{xl} = max_{u,v \in \mathcal{Z}} \ t_{xuv}$  \tcp*{first $\gamma^{th}$ occurrence time of an edge $(u,v)$ within $ [t_x-1, t_y]$}
		\If{$t_{xl} - \Delta < t_x$}{
			$\text{flag} = FALSE$\;
			\If{$(\mathcal{Z}, [t_{xl}-\Delta, t_y]) \notin \mathcal{C}_{im}$}{
				add $(\mathcal{Z}, [t_{xl}-\Delta, t_y])$ to $\mathcal{C}^I$ and $\mathcal{C}_{im}$\;
			}
		}
		\KwRet \text{flag}\;
      }
\end{procedure}
\paragraph{Procedure \texttt{Expand\_Left\_TS()}:} Similar to Procedure \ref{proc:rightts}, this procedure takes a $(\Delta, \gamma)$\mbox{-}Clique from $\mathcal{C}^{I}$ and the edge list for the current update cycle, as inputs, and returns a Boolean flag indicating whether the input clique is maximal or not. For a clique $(\mathcal{Z}, [t_x, t_y])$, the trivial way of extending $t_x$ is as follows: For every pair of vertices $u,v \in \mathcal{Z}$, let the first $\gamma$-th occurrence time stamp within $[t_x -1, t_y]$ is $t_{xuv}$. If subtracting $\Delta$ from the latest of all $t_{xuv}$ is less than $t_x$, then $(\mathcal{Z}, [t_x, t_y])$ is not maximal and the new clique $(\mathcal{Z}, [t_{xl} - \Delta, t_{y}])$ is formed. Assume that, the entire set of links till $T_2$ ($\mathcal{E}^{T_2}$) is being processed to get $\mathcal{C}^{T_2}$. Now, it is easy to observe that for a maximal clique $(\mathcal{Z}, [t_x, t_y])$ in $\mathcal{C}^{T_2}$, if $[t_x, t_y]$ lies entirely within $ [T_1 - \Delta, T_2]$, Procedure \ref{proc:leftts} can reach to the maximal clique by extending the start time stamp towards left. However, if $T_1 - \Delta$ lies within the $[t_x, t_y]$ and the first gamma edges of every pair of vertices in $\mathcal{Z}$ does not belong to $\mathcal{E}^{ [T_{1}-\Delta, T_{2}]}$, Procedure \ref{proc:leftts} will fail to get the maximal cliques from $\mathcal{C}^I$, initialized using $\mathcal{E}^{ [T_{1}-\Delta, T_{2}]}$. This scenario returns a non-maximal cliques by identifying them as falsely maximal. However, we do not miss the maximal ones as the algorithm uses the cliques from $\mathcal{C}^{T_1}_{ex}$, which has it's $t_x$ fixed and the $t_y$ is extended correctly by Procedure \ref{proc:rightts}. The following lemma highlights this claim.

\begin{procedure}[h]
\caption{Removal of the sub cliques()}  \label{proc:removal}
\SetKwFunction{FMain}{EOC\_Remove\_Sub\_Cliques}
	  \SetKwProg{Pn}{Function}{:}{\KwRet}
      \Pn{\FMain{$T_1$}}{
      \If{$T_1 == -1 $}{
      \KwRet \;
       }
      $\mathcal{C}_{check} \longleftarrow \emptyset$  \;
      $R_{dic} = dict()$ \;
      \For{all $(\mathcal{Z}, [t_x, t_y]) \in \mathcal{C}^{T_2 \setminus T_1}$}{
      \If{$\mathcal{Z} \notin R_{dic}.keys()$}{
      $\text{The new `key' } \mathcal{Z} \text{ is added to } R_{dic}$\;
      $R_{dic}[\mathcal{Z}] \longleftarrow \emptyset$\;
      }
      add $[t_x, t_y]$ in $R_{dic}[\mathcal{Z}]$\;
      \If{$t_x \leq T_1 $}{
        add $(\mathcal{Z}, [t_x, t_y])$ in $\mathcal{C}_{check}$\;
        }
      }
      
      \For{$(\mathcal{Z}, [t_x, t_y]) \in \mathcal{C}_{check}$}{
      \If{$\mathcal{Z} \in R_{dic}.keys()$}{
      \If{$|\{ [t_{{x}^{'}}, t_{{y}^{'}}]:  [t_{{x}^{'}}, t_{{y}^{'}}] \in R_{dic}[\mathcal{Z}] \text{ and }  [t_x, t_y] \subset  [t_{{x}^{'}}, t_{{y}^{'}}]  \}| \geq 1 $}{
         remove $(\mathcal{Z}, [t_x, t_y])$ from $\mathcal{C}^{T_{2} \setminus T_{1}}$\;
         \textbf{continue}\;
         }
        }
        $temp = \{(\mathcal{Z}^{'}, [t_{{x}^{'}}, t_{{y}^{'}}]) : (\mathcal{Z}^{'}, [t_{{x}^{'}}, t_{{y}^{'}}]) \in \mathcal{C}^{T_2 \setminus T_1} \text{ and } \mathcal{Z} \subset \mathcal{Z}^{'} \}$\;
        \For{$ \text{All }(\mathcal{Z}^{'}, [t_{{x}^{'}}, t_{{y}^{'}}]) \in temp$}{
          \If{$[t_x, t_y] \subseteq  [t_{{x}^{'}}, t_{{y}^{'}}]$}{
          remove $(\mathcal{Z}, [t_x, t_y])$ from $\mathcal{C}^{T_{2} \setminus T_{1}}$\;
          \textbf{break}\;}
        }
      }
      \KwRet \;
     }
\end{procedure}
\paragraph{Procedure \texttt{EOC\_Remove\_Sub\_Cliques()}:} In the description of Procedure \ref{proc:leftts}, we have realized that some of the non-maximal cliques are declared as falsely maximal of $\mathcal{C}^{T_2 \setminus T_1}$ before invoking Procedure \ref{proc:removal}. However, it removes such cliques and generate the final maximal clique set. It identifies the cliques with $t_x \leq T_1$, which are possibly non-maximal candidates (verified in Lemma \ref{lemma:maximality}). The cliques which are formed with the links $\mathcal{E}^{ [T_{1}-\Delta, T_{2}]}$, unable to verify it's extendibility towards left. This results in having $t_x$ greater than its actual value from it's maximal counterpart. So, for an identified clique $(\mathcal{Z}, [t_x, t_y])$, there are two possibilities; i) with same $\mathcal{Z}$ their exist a $[t_{x^{'}}, t_{y^{'}}] $ which contains $[t_x, t_y]$, or ii) there exist a clique $(\mathcal{Z}^{'}, [t_{x^{'}}, t_{y^{'}}])$ such that $\mathcal{Z} \subset \mathcal{Z}^{'}$ and $[t_x, t_y] \subseteq [t_{x^{'}}, t_{y^{'}}]$, resulting $(\mathcal{Z}, [t_x, t_y])$ as non-maximal. Hence, $(\mathcal{Z}, [t_x, t_y])$ is removed from $\mathcal{C}^{T_{2} \setminus T_{1}}$.
\par Procedure \ref{proc:removal} takes $T_1$ as input, the time stamp till which the maximal cliques are enumerated previously. When the algorithm runs for the first time (or with first batch of links), then the input of the procedure becomes $-1$ and returns from Line 3. For  other cases, the removal of sub-cliques are done from Line 4 to 25. We keep all the cliques to be checked for maximality condition in $\mathcal{C}_{check}$ and set it to $\emptyset$ in Line 4. We prepare a dictionary $R_{dic}$ to hold all the cliques produced in $\mathcal{C}^{T_2 \setminus T_1}$. The keys of $R_{dic}$ is the vertex set of a clique and the value is the list containing all the time intervals when the corresponding vertex set has formed the $(\Delta, \gamma)$-cliques. For all the cliques in $\mathcal{C}^{T_2 \setminus T_1}$, it is added in $R_{dic}$ from line 7 to 10. Also, if left time stamp of the clique is less than or equal to $T_1$, it is added in $\mathcal{C}_{check}$ in line 11.
Now, for each clique $(\mathcal{Z},[t_x, t_y])$ in $\mathcal{C}_{check}$, two of the following things are checked. (i) $[t_x, t_y]$ is proper subset of any of the values in $R_{dic}[\mathcal{Z}]$, (ii) $\mathcal{Z}$ is subset of any other keys in $R_{dic}$ with $[t_x, t_y]$ is subset of that cliques as well. The clique can be removed if any of the mentioned case becomes true. The first case is checked within Line 14 to 17. If $\mathcal{Z}$ appears in $R_{dic}.keys()$, it computes the set having the intervals which are proper superset of $[t_x, t_y]$ in Line 15. Now, if the cardinality of that set is greater than 1, i.e., there exist at least one clique which contains the $(\mathcal{Z}, [t_x, t_y])$ temporally within it. Hence, $(\mathcal{Z}, [t_x, t_y])$ is not maximal and removed from $\mathcal{C}^{T_2 \setminus T_1}$ in line 16, and continues to pick next clique from $\mathcal{C}_{check}$. If Case (i) is false, it tries to check for case (ii). In line 18, it constructs a set $temp$ with the cliques from $\mathcal{C}^{T_2 \setminus T_1}$, whose vertex set is proper superset of $\mathcal{Z}$. Now, for each cliques in $temp$, it is checked if the time interval is also super set of $[t_x, t_y]$ in Line 20. If it is true, $(\mathcal{Z},[t_x, t_y])$ is not maximal, as it is contained into another clique both in terms of set of vertices and temporally. Hence, it is removed from $\mathcal{C}^{T_2 \setminus T_1}$, and the for loop at line 19 breaks. Finally, all the non-maximal cliques are removed from $\mathcal{C}^{T_2 \setminus T_1}$. Now, we state and prove few lemmas to show the correctness of Procedure \ref{proc:removal}.
\begin{mylemma}\label{lemma:maximality}
It is sufficient to identify the cliques having $t_x \leq T_1$, as possible candidates for checking the maximality condition.
\end{mylemma}
\begin{proof}
Before executing Procedure \ref{proc:removal}, a non-maximal clique $(\mathcal{Z}, [t_x, t_y])$ in $\mathcal{C}^{T_2 \setminus T_1}$ can come from either of the two following sources: i) $\mathcal{C}^{T_1}_{ex}$, and ii) the cliques initialized in Line 10 of Algorithm \ref{Algo:1}.
\par In the first case, we consider all the cliques having $t_y \geq T_1$ in $\mathcal{C}^{T_1}_{ex}$. While preparing $\mathcal{C}^{T_2 \setminus T_1}$, Algorithm \ref{Algo:1} only extends the right time stamp for the cliques in $\mathcal{C}^{T_1}_{ex}$. So, it is required to verify their maximality. Now, we need to show, what the maximum value of $t_x$ is possible for the cliques in $\mathcal{C}^{T_1}_{ex}$. Consider one scenario, when for a particular vertex pair $\{u, v\}$, all the $\gamma$ links appeared consecutively at each time stamp within $[T_1 -\gamma + 1, T_1]$. By the initialization algorithm of \cite{banerjee2019enumeration}, it will generate two cliques $(\{u,v\}, [T_1 -\gamma + 1, T_1 -\gamma + 1 + \Delta])$ and $(\{u,v\}, [ T_1 -\Delta, T_1])$. For $\gamma \in [1, \Delta+1]$, we observe the condition $T_1 - \Delta \leq T_1 -\gamma + 1 \leq T_1$ as true.
\par For the second case, the cliques will be extended in all the possible three ways of expansions. As it has the link set $\mathcal{E}^{[T_1 - \Delta, T_2]}$, expanding the left time stamp is not correct with respect to all the links till $T_2$. Now, the minimum possible value for the $t_x$, can be less than $T_1 - \Delta$. Consider one scenario, for a particular vertex pair $\{u, v\}$, all the $\gamma$ links only appeared consecutively at each time stamp within $[T_1 -\gamma + 1, T_1]$. Then, the maximal clique generated out of it, will be $A = (\{u, v\}, [T_1 - \Delta, T_1 -\gamma + 1 + \Delta ])$. Now, if $(u,v)$ appears on $T_1 - \Delta - 1$, then $A$ is not maximal and it's left time stamp should be expanded. So, the maximum value possible for such non-maximal cliques will be $T_1 - \Delta$. Hence, it is required to check all the cliques with $t_x \leq T_1 - \Delta$.
\par Combining both the cases, it is sufficient to identify the cliques having $t_x \leq T_1$, as possible candidates for checking the maximality condition.
\end{proof}
\begin{mylemma} \label{lemma:extendend}
$\mathcal{C}^{T_{1}}_{ex}$ contains all the cliques for updating $\mathcal{C}^{T_1}$ to $\mathcal{C}^{T_2}$.
\end{mylemma}
\begin{proof}
While building the maximal cliques till time stamp $T_1$, all the intermediate cliques having $t_y \geq T_1$ are kept in $\mathcal{C}^{T_{1}}_{ex}$. To prove the lemma statement, we need to show that $t_y \geq T_1$ is sufficient condition, to have all the cliques for extending in right (Refer Lemma \ref{lemma:6}). Here, we divide the proof in two parts: i) $t_y \ge T_1 - k$, ii) $t_y \ge T_1 + k$, where $k \in \mathbb{Z}^{+}$. 
For part (i), it needs to have the links from $t_y - \Delta$ to $T_2$, to extend it's right time stamp by Procedure \ref{proc:rightts}. According to Lemma \ref{lemma:t1-delta}, it will have to process the link set $\mathcal{E}^{[T_1 - \Delta -k, T_2]}$, which will incur the redundant processing.
For part (ii), the maximum possible value for $t_y$ for the cliques generated with link set $\mathcal{E}^{T_1}$ can be greater than $T_1$. Now consider a scenario, where for a vertex pair $\{u, v\}$, the consecutive $\gamma$ links have only appeared within $[T1-\Delta , T_1 - \Delta + \gamma - 1]$. Hence, two cliques will be generated with the initialization of \cite{banerjee2019enumeration}, as $A_0 = (\{u, v\}, [T_1 - \Delta, T_1])$ and $A_1 = (\{u, v\}, [T_1 - 2\Delta + \gamma - 1, T_1 - \Delta + \gamma -1])$. Now if there exist a link $(u,v,t)$ with $t = T_1 + 1$, then the maximal clique becomes $A=(\{u, v\}, [T_1 - 2\Delta + \gamma - 1, T_1 + 1])$. To get the maximal clique $A$, $\mathcal{C}^{T_{1}}_{ex}$ must contain the clique $A_3=(\{u, v\}, [T_1 - 2\Delta + \gamma - 1, T_1])$. Hence, it is required to have the smallest possible value for the extension of $t_y$ is $T_1$.
From part (i) and (ii), it is proved that the $t_y \geq T_1$ is sufficient condition, to have all the cliques for extending in right. So, $\mathcal{C}^{T_{1}}_{ex}$ contains all the cliques for updating $\mathcal{C}^{T_1}$ to $\mathcal{C}^{T_2}$.
\end{proof}
\begin{mylemma} \label{lemma:9}
Without execution of Procedure \ref{proc:removal}, $\mathcal{C}^{T_{2} \setminus T_{1}} \cup (\mathcal{C}^{T_{1}} \setminus \mathcal{C}^{T_{1}}_{ex})$ contains all the maximal cliques along with some non-maximal ones.
\end{mylemma}
\begin{proof}
We denote $\mathcal{C}^{T_{2} \setminus T_{1}} \cup (\mathcal{C}^{T_{1}} \setminus \mathcal{C}^{T_{1}}_{ex})$ as $\hat{\mathcal{C}}^{T_{2}}$. Assume, a maximal clique $(\mathcal{Z}, [t_x, t_y])$ is not present in $\hat{\mathcal{C}}^{T_{2}}$. We prove the lemma statement with contradiction that $(\mathcal{Z}, [t_x, t_y])$ should be in $\hat{\mathcal{C}}^{T_{2}}$. To simplify the proof, we classify $(\mathcal{Z}, [t_x, t_y])$ to be in any of the three of the following cases.
\begin{itemize}
    \item \textbf{Case 1:} $[t_x, t_y] \subseteq [T_0, T_1)$,
    \item \textbf{Case 2:} $[t_x, t_y] \subseteq (T_1, T_2]$, 
    \item \textbf{Case 3:} $\{t_x < T_1 \And t_y \geq T_1\}$, or $\{t_x \leq T_1 \And t_y > T_1 \}$.
\end{itemize}
For both Case 1 and 2, it can be proved from Lemma 6 of \cite{banerjee2019enumeration} that the initialization is correct to get the maximal cliques with link set $\mathcal{E}^{T_1}$ and $\mathcal{E}^{[T_1 - \Delta, T_2]}$, respectively. Now, Lemma \ref{lemma:3} and \ref{Lemma:5}, together verifies that the candidate set associated with the initialized cliques is correct and complete. It shows that from any initialized clique $(\{u, v\}, [t_x^{'}, t_x^{'}+\Delta])$ where $u, v  \in \mathcal{Z}$ and $ [t_x^{'}, t_x^{'}+\Delta] \subseteq [t_x, t_y]$, $(\mathcal{Z}, [t_x^{'}, t_x^{'}+\Delta])$ will be generated. Now, Lemma 7 of \cite{banerjee2019enumeration} show that $(\mathcal{Z}, [t_x, t_y])$ will be formed from $(\mathcal{Z}, [t_x^{'}, t_x^{'}+\Delta])$. Hence, $(\mathcal{Z}, [t_x, t_y])$ will be in $\mathcal{C}^{T_1}$ (also not in $\mathcal{C}^{T_1}_{ex}$), and $\mathcal{C}^{T_{2} \setminus T_{1}}$ (also not in $\mathcal{C}^{T_{2} \setminus T_{1}}_{ex}$) for Case 1 and Case 2, respectively. Hence, $(\mathcal{Z}, [t_x, t_y])$ should be in $\hat{\mathcal{C}}^{T_{2}}$.
\par For Case 3, $(\mathcal{Z}, [t_x, t_y])$ has to be initialized from $\mathcal{C}^{T_1}_{ex}$. According to Lemma \ref{lemma:4} and \ref{lemma:extendend}, for the maximal clique $(\mathcal{Z}, [t_x, t_y])$ there will be a clique $(\mathcal{Z}, [t_x, t_y^{'}])$ with $t_y^{'} \leq t_y$ in $\mathcal{C}^{T_1}_{ex}$. Hence, $(\mathcal{Z}, [t_x, t_y])$ will only be in $\mathcal{C}^{T_2 \setminus T_1}$ and not in $(\mathcal{C}^{T_{1}} \setminus \mathcal{C}^{T_{1}}_{ex})$. So, $(\mathcal{Z}, [t_x, t_y])$ is present in $\hat{\mathcal{C}}^{T_{2}}$.
\par To prove that $\hat{\mathcal{C}}^{T_{2}}$ will contain non-maximal cliques as well, it is enough to show one such non-maximal clique exists in $\hat{\mathcal{C}}^{T_{2}}$. The proof of Lemma \ref{lemma:maximality} shows that such non-maximal clique will exist. It concludes the proof of the lemma statement.
\end{proof}
\begin{mylemma} \label{lemma:10}
Procedure \ref{proc:removal} correctly removes all the non-maximal cliques, while building $\mathcal{C}^{T_2}$ from $\mathcal{C}^{T_1}$.
\end{mylemma}
\begin{proof}
The cliques having it's left timestamp less than or equal to $T_1$, are the only possible candidates, which can become a non-maximal clique in $\mathcal{C}^{T_2}$. It follows from Lemma \ref{lemma:maximality}. According to Lemma \ref{lemma:extendend}, it is evident that $\mathcal{C}^{T_2 \setminus T_1}$ will contain the maximal cliques along with some non-maximal cliques. By definition, the non-maximal $(\Delta, \gamma)$-cliques will contain temporally or by vertex sets into the maximal ones. Now, from the description of Procedure \ref{proc:removal}, it is easy to observe that it can remove a clique if this is contained into at least one another clique by temporally or by vertex sets. Hence, Procedure \ref{proc:removal} will remove all the non-maximal cliques while building $\mathcal{C}^{T_2}$ from $\mathcal{C}^{T_1}$.
\end{proof}

From Lemma \ref{lemma:10}, it is clear that $\mathcal{C}^{T_2 \setminus T_1}$ will not contain any non-maximal cliques of $\mathcal{C}^{T_2}$. Now, $(\mathcal{C}^{T_1} \setminus \mathcal{C}^{T_1}_{ex})$ will contain the maximal cliques only, as the minus operation will remove the cliques which are declared as falsely maximal due to the non-availability of the entire link set till $T_2$. Hence, $\mathcal{C}^{T_{2} \setminus T_{1}} \cup (\mathcal{C}^{T_{1}} \setminus \mathcal{C}^{T_{1}}_{ex})$ will contain the maximal cliques only till time stamp $T_2$. From these, we state the following theorems.
\begin{mytheorem}\label{Th:2}
 All the cliques in $\mathcal{C}^{T_{2}}$ are maximal cliques till time stamp $T_2$.
\end{mytheorem}
\begin{mytheorem}\label{Th:3}
  $\mathcal{C}^{T_{2}}$ contains all the maximal cliques till time stamp $T_2$.
\end{mytheorem}
Together Theorem \ref{Th:2}, and \ref{Th:3} complete the correctness of the proposed methodology. Now, we analyze the time and space complexity of the proposed methodology. Initially, we start with Algorithm \ref{Algo:1}. Let, $m_{1}$, $m_{2}$, and $m_{3}$ denote the number of links till time stamp $T_{1}-\Delta$, from $T_{1}-\Delta$ to $T_{1}$, and from $T_{1}$ to $T_{2}$, respectively. Now, the number of links from time stamp $T_{1}-\Delta$ to $T_{2}$ are $(m_2+m_3)$. As per the analysis shown in  \cite{banerjee2019enumeration}, in the worst case, the size of $\mathcal{C}_{ex}^{T_{1}}$ can be $\mathcal{O}(2^{n}(m_1+m_2 - \gamma+1))$. Each clique can be of size $\mathcal{O}(n)$. Hence, copying the cliques from $\mathcal{C}_{ex}^{T_{1}}$ to $\mathcal{C}^{I}$ requires $\mathcal{O}(2^{n}n(m_1+m_2 - \gamma+1))$ time. All the statements in Line $2$ are intialization statement, where the first two requires $\mathcal{O}(1)$ time, whereas the third one requires $\mathcal{O}(2^{n}n(m_1+m_2 - \gamma+1))$ time. In Line $4$, removing a clique requires $\mathcal{O}(n)$ time. Complexity analysis of the \texttt{Extend\_Right\_TS} Procedure (i.e., Procedure 3) has been analyzed little later. Checking the condition of the \texttt{if} statement in Line $6$ and $8$ requires $\mathcal{O}(1)$ time, and putting the cliques into $\mathcal{C}^{T_2 \setminus T_1}$ and  $\mathcal{C}_{ex}^{T_{1}}$ requires $\mathcal{O}(n)$ time. Now, it is important to understand how many times the \texttt{while} loop of Line $3$ will execute. In the worst case, all the cliques of $\mathcal{C}_{ex}^{T_{1}}$ may extend, and hence the  \texttt{while} loop will execute for $\mathcal{O}(|\mathcal{C}_{ex}^{T_{1}}|.(T_2 - T_1))=\mathcal{O}(2^{n}(m_1+m_2 - \gamma+1)(T_2 - T_1))$. As per Lemma $1$ of \cite{banerjee2019enumeration} complexity of executing Line $10$ requires $\mathcal{O}(\gamma(m_2+m_3))$ time. In the worst case the size of $\mathcal{C}^{I}$ could be $\mathcal{O}(\gamma(m_2+m_3))$. Hence, copying the cliques of $\mathcal{C}^{I}$ into $\mathcal{C}_{im}$ requires $\mathcal{O}(\gamma(m_2+m_3))$ time.
\par Inside the next \texttt{while} loop, removing a clique in Line $13$ from $\mathcal{C}^{I}$ requires $\mathcal{O}(n)$ time. In Line $14$ condition checking of the \texttt{if} statement requires $\mathcal{O}(1)$ time. Now, the \texttt{if} condition checks whether $t_y - t_x == \Delta$ or not. Hence in the worst case, there can be $\Delta + 1$ links between any two vertices and all the vertices are connected within that $\Delta$ duration. So, preparing the static graph at Line 15, requires $\mathcal{O}(n^2\Delta)$ time as identified by the maximum possible number of links within a $\Delta$ duration. Associating $N_{G}(\mathcal{Z})$ to $(\mathcal{Z},[t_x,t_y])$ requires $\mathcal{O}(|\mathcal{Z}|(n-|\mathcal{Z}|))$ time. In the worst case, this quantity will be $\mathcal{O}(n^{2})$. Now, as per the sequential steps of Algorithm \ref{Algo:1}, subsequently we proceed to analyze the time and space requirement for the Procedures \ref{proc:nodeadd},\ref{proc:rightts},\ref{proc:leftts}, and \ref{proc:removal}.
\par Now, we start with Procedure \ref{proc:nodeadd}. As mentioned in the analysis of Algorithm 2 of \cite{banerjee2019enumeration}, the maximum number of intermediate cliques are $\mathcal{O}(2^{n}(m_2+m_3 -\gamma +1))$. For each of these cliques, time requirement to execute Procedure \ref{proc:nodeadd} is as follows. As mentioned previously, for any intermediate clique $(\mathcal{Z},[t_x,t_y])$, $|N_{G}(\mathcal{Z})|$ can be at most $\mathcal{O}(n)$. Hence, the \texttt{for} loop in Line $3$ will execute $\mathcal{O}(n)$ times in the worst case. To check any clique $(\mathcal{Z},[t_x,t_y])$ holds the $(\Delta, \gamma)$\mbox{-}Clique property or not, we need to check for all the links for the vertices of $\mathcal{Z}$. In the worst case, this may require $\mathcal{O}(n(m_2+m_3))$ time. So, the condition checking of the \texttt{if} statement in Line $4$ requires $\mathcal{O}(n(m_2+m_3))$ time. Setting the `flag'
in Line $5$ requires $\mathcal{O}(1)$ time. Now, as the maximum number of intermediate cliques are $\mathcal{O}(2^{n}(m_2+m_3 -\gamma +1))$, hence in the condition checking of the \texttt{if} statement in Line $6$, the clique $(\mathcal{Z}\cup \{u\},[t_x,t_y])$ needs to be compared with  $\mathcal{O}(2^{n}(m_2+m_3 -\gamma +1))$ number of cliques. If the vertex ids of the clique are always stored in the sorted order then two $(\Delta, \gamma)$\mbox{-}cliques can be compared in $\mathcal{O}(n)$ time. If the cliques in the $\mathcal{C}_{im}$ are stored in the sorted of $t_x$ then the number of comparisons will be $\mathcal{O}(\log (2^{n}(m_2+m_3 -\gamma +1)))=\mathcal{O}(n+ \log(m_2+m_3 -\gamma +1))$. Hence, the total time requirement for the condition checking of the \texttt{if} statement in Line $6$ requires $\mathcal{O}(n^{2}+ n.\log(m_2+m_3 -\gamma +1))$ time. Adding the new clique in $\mathcal{C}_{im}$ and $\mathcal{C}^{I}$ requires $\mathcal{O}(n)$ time. Hence, the total time requirement for Procedure \ref{proc:nodeadd} is of $\mathcal{O}(n(n(m_2+m_3)+n^{2}+ n.\log(m_2+m_3 -\gamma +1) +n))= \mathcal{O}(n^{2}(m_2+m_3)+ n^{3})$.
\par Next, we analyze Procedure \ref{proc:rightts}. Setting the \texttt{flag} in Line $2$ requires $\mathcal{O}(1)$ time. Finding $t_{yr}$ in Line $3$ requires $\mathcal{O}(m_2+m_3)$ time. Condition checking of the \texttt{if} statement in Line $4$ requires $\mathcal{O}(1)$ time. As mentioned in the analysis of Procedure \ref{proc:nodeadd}, checking for the belongingness of any clique  $(\mathcal{Z},[t_x,t_y])$ requires $\mathcal{O}(n^{2}+ n.\log(m_2+m_3 -\gamma +1))$ time. Hence, the running time for the Procedure \ref{proc:rightts} is of $\mathcal{O}(m_2+m_3+n^{2}+ n.\log(m_2+m_3 -\gamma +1))$ time. As Procedure \ref{proc:leftts} is identical to Procedure \ref{proc:rightts}, hence time requirement for Procedure \ref{proc:leftts} will also be of $\mathcal{O}(m_2+m_3+n^{2}+ n.\log(m_2+m_3 -\gamma +1))$. 
\par Now, we analyze Procedure \ref{proc:removal}. It is easy to follow that all the statements from Line $2$ to $5$ require $\mathcal{O}(1)$ time. Number of keys in the dictionary $R_{dic}$ is of $\mathcal{O}(2^{n})$. It is easy to follow that the number of times the \texttt{for} loop in Line $6$ will execute $\mathcal{O}(2^{n}(m_2+m_3-\gamma+1))$ times. Condition checking of the \texttt{if} statement in Line $7$ requires $\mathcal{O}(2^{n}.n)$ time. Executing Line $8$ and $9$ require $\mathcal{O}(n)$, and $\mathcal{O}(1)$ time, respectively. Appending the time duration of the clique corresponding to the vertex set of the clique as `key' requires $\mathcal{O}(1)$ time. The condition checking of the \texttt{if} statement and at Line $11$ requires $\mathcal{O}(1)$ time, and adding clique $(\mathcal{Z}, [t_x,t_y])$ at Line $12$ requires $\mathcal{O}(n)$ time. Hence, time requirement from Line $2$ to $12$ is of $\mathcal{O}(2^{n}(m_2+m_3-\gamma+1)(n.2^{n}+n)) = \mathcal{O}(2^{2n}.n.(m_2+m_3-\gamma+1))$. Now, it is easy to follow that the in the worst case the size of $\mathcal{C}_{check}$ will be $\mathcal{O}(2^{n}(m_2+m_3-\gamma+1))$. So, the \texttt{for} loop in Line $13$ will run $\mathcal{O}(2^{n}(m_2+m_3-\gamma+1))$ times. Condition checking of the \texttt{if} statement in 
Line $14$ requires $\mathcal{O}(2^{n}.n)$ time. Now, let $f_{max}$ denotes the maximum number of cliques with the same vertex set. Hence, the condition checking of the \texttt{if} statement in Line $15$ requires $\mathcal{O}(f_{max})$ time. Then, removing the clique from $\mathcal{C}^{T_2 \setminus T_1}$ requires $\mathcal{O}(n)$ time. The cardinality of `temp' in Line $18$ can be given by the following equation:
\begin{equation}
    |temp|= \bigg [\binom{n}{|\mathcal{Z}|+1}+ \binom{n}{|\mathcal{Z}|+2}+ \ldots + \binom{n}{|\mathcal{Z}|+(n-|\mathcal{Z}|)} \bigg ]. \mathcal{O}(f_{max})
\end{equation}
In the worst case $|temp|$ may converges to $\mathcal{O}(2^{n}.f_{max})$. Hence, the \texttt{for} loop in Line $19$ will run for $\mathcal{O}(2^{n}.f_{max})$ times. It is easy to observe that execution of the condition checking of the \texttt{if} statement, removing the clique from $\mathcal{C}^{T_2 \setminus T_1}$, require $\mathcal{O}(1)$ and $\mathcal{O}(n)$ respectively. Hence, the total time requirement from Line $13$ to $22$ is as follows: $\mathcal{O}(2^{n}(m_2+m_3-\gamma+1)(n.2^{n}+ n.2^{n}(f_{max} + n) + n.2^{n}.f_{max}))=\mathcal{O}(2^{2n}.n.(m_2+m_3-\gamma+1).(f_{max}+n))$. Hence, the total running time of Procedure \ref{proc:removal} is $\mathcal{O}(2^{2n}.n.(m_2+m_3-\gamma+1).(f_{max}+n))$.
\begin{table}[]
    \centering
    \begin{tabular}{|c|c|}
    \hline
        Procedure & Time  \\ \hline
        \texttt{Expand\_Vertex\_Set} \ref{proc:nodeadd} & $\mathcal{O}(n^{2}(m_2+m_3)+ n^{3})$  \\ \hline
        \texttt{Extend\_Left\_TS} \ref{proc:leftts} & $\mathcal{O}(m_2+m_3+n^{2}+ n.\log(m_2+m_3 -\gamma +1))$ \\ \hline
        \texttt{Extend\_Right\_TS} \ref{proc:rightts} & $\mathcal{O}(m_2+m_3+n^{2}+ n.\log(m_2+m_3 -\gamma +1))$ \\ \hline
        \texttt{EOC\_Removal\_Sub\_Cliques} \ref{proc:removal} & $\mathcal{O}(2^{2n}.n.(m_2+m_3-\gamma+1).(f_{max}+n))$ \\ \hline
    \end{tabular}
    \caption{Computational Time Required by the Procedures}
    \label{tab:time_proc}
\end{table}
\par Now, the final task is to add up the step wise running time of Algorithm \ref{Algo:1} to get the running time of the proposed methodology. After adding up running time from Line $1$ to $11$ will be of $\mathcal{O}(2^{n}(m_1+m_2-\gamma+1)(T_2-T_1)(m_2+m_3+n^{2}+n \log(m_2+m_3-\gamma+1)))$. Also, running time of Line $12$ to $23$ of Algorithm \ref{Algo:1} requires $\mathcal{O}(2^{n}.(m_2+m_3-\gamma+1) (n^3 + n^2(m_2 + m_3)) )$. In Line $25$, we are performing set minus operation between $\mathcal{C}^{T_{1}}_{ex}$ and $\mathcal{C}^{T_{1}}$. Now, performing set minus between two sets with $k_1$ and $k_2$ elements requires $\mathcal{O}(k_1, k_2)$ number of operations. In the worst case size of both $\mathcal{C}^{T_{1}}_{ex}$ and $\mathcal{C}^{T_{1}}$ can be $\mathcal{O}(2^{n}(m_1+m_2-\gamma+1))$ and the size of a $(\Delta,\gamma)$-clique could be $\mathcal{O}(n)$, performing this set minus operation requires $\mathcal{O}(2^{2n}n^{2}(m_1+m_2-\gamma+1)^{2})$. Now, the number of cliques in $\mathcal{C}^{T_2 \setminus T_1}$ can be $\mathcal{O}(2^{n}(m_2+m_3-\gamma+1) + 2^{n}(m_1+m_2-\gamma+1)(T_2 - T_1)) $. The line $25$ can be executed by copying the elements of $\mathcal{C}^{T_1} \setminus \mathcal{C}^{T_1}_{ex}$ into $\mathcal{C}^{T_2 \setminus T_1}$ and add the reference to a new variable $\mathcal{C}^{T_2}$. Now, as the number elements in $\mathcal{C}^{T_1} \setminus \mathcal{C}^{T_1}_{ex}$ can be $\mathcal{O}(2^{n}(m_1+m_2-\gamma+1))$, copying that requires $\mathcal{O}(n 2^{n}(m_1+m_2-\gamma+1))$ time. Hence, the total time of line $25$ is $\mathcal{O}(2^{2n}n^{2}(m_1+m_2-\gamma+1)^{2} + n 2^{n}(m_1+m_2-\gamma+1) ) = \mathcal{O}(2^{2n}n^{2}(m_1+m_2-\gamma+1)^{2})$.  The time complexity of the proposed methodology is of $\mathcal{O}(2^{n}(m_1+m_2-\gamma+1)(T_2-T_1)(m_2+m_3+n^{2}+n \log(m_2+m_3-\gamma+1)) + 2^{n}.(m_2+m_3-\gamma+1) (n^3 + n^2(m_2 + m_3)) + 2^{2n}n^{2}(m_1+m_2-\gamma+1)^{2})$.
\par Now, we turn our attention to the space requirement. It can be observed from the Algorithm \ref{Algo:1} that the space requirement of the proposed methodology is basically the sum of the space requirement of the following individual structures: $\mathcal{C}^{I}$, $\mathcal{C}_{im}$, $\mathcal{C}^{T_{2} \setminus T_{1}}$, $\mathcal{C}^{T_{2} \setminus T_{1}}_{ex}$, $\mathcal{R}_{dic}$, $\mathcal{C}_{check}$, and $Temp$. Among them the last three structures has been used in the \texttt{EOC\_Removal\_Sub\_Cliques} subroutine. Table \ref{tab:space} contains individual space requirement by different structures. So, the total space requirement of the proposed methodology is of $\mathcal{O}(n2^{n}(m_2 + m_3 - \gamma + 1) + n2^{n}(m_1 + m_2 - \gamma + 1)(T_2 - T_1)) $.

\begin{table}[]
    \centering
    \begin{tabular}{|c|c|}
    \hline
      Structure   &  Space \\ \hline
        $\mathcal{C}^I$  & $\mathcal{O}(n2^{n}(m_2 + m_3 - \gamma + 1)) $ \\ \hline
        $\mathcal{C}_{im}$  & $\mathcal{O}(n2^{n}(m_2 + m_3 - \gamma + 1) + n2^{n}(m_1 + m_2 - \gamma + 1)(T_2 - T_1)) $ \\ \hline
     $\mathcal{C}^{T_2 \setminus T_1}$  & $\mathcal{O}(n2^{n}(m_2 + m_3 - \gamma + 1)) $ \\ \hline
     $\mathcal{C}^{T_2 \setminus T_1}_{ex}$  & $\mathcal{O}(n2^{n}(m_2 + m_3 - \gamma + 1)) $ \\ \hline
     $\mathcal{R}_{dic}$  & $\mathcal{O}(n2^{n}(m_2 + m_3 - \gamma + 1)) $ \\ \hline
     $\mathcal{C}_{check}$  & $\mathcal{O}(n2^{n}(m_2 + m_3 - \gamma + 1)) $ \\ \hline
     $temp$ at Line 18 in Procedure \ref{proc:removal}  & $\mathcal{O}(n2^{n}(m_2 + m_3 - \gamma + 1)) $ \\ \hline
    \end{tabular}
    \caption{Computational space required to store different structures}
    \label{tab:space}
\end{table}
\section{Experimental Evaluation}\label{Sec:EE}
Here we describe the experimental evaluation of the proposed solution approach. This section has been arranged in the following way. Subsection \ref{Sec:DD} contains the description of the datasets. Subsection \ref{SubSec:Exp_Set_Up} contains the set up for the experimentation. The goals of the experimentation have been listed out in Subsection \ref{Sec:GE}. Finally, Subsection \ref{Sec:ER} contains the experimental results with a detailed discussion. 

\subsection{Dataset Description} \label{Sec:DD}
In this study, we use the following four publicly available temporal network datasets: 
\begin{itemize}
    \item \textbf{Infectious \cite{isella2011s}:} This dataset contains the dynamic contact information collected at the time of Infectious SocioPattern event at the science gallery of Dublin city. Content of this dataset is the collection of tuples of type $(t,u,v)$ signifying a contact between $u$ and $v$ at time $t$.

    \item \textbf{Hypertext \cite{isella2011s}:} This dataset was collected during the ACM Hypertext conference 2009, where the conference attendees voluntarily weared wireless devises and their contacts (when two attendees come to a close proximity) during the conference days are captured in this dataset.

    \item \textbf{College Message \cite{panzarasa2009patterns}:} This dataset contains the interaction information among a group of students from University of California, Irvine.
\item \textbf{Autonomous Systems (AS180) \cite{leskovec2005graphs}:} The dataset contains the daily traffic flow between routers in a communication network. The data was collected from University of Oregon Route Views Project - Online data and reports. The dataset contains 733 daily instances which span an interval of 785 days from November 8 1997 to January 2 2000. As the number of links on each day is very large,  we consider the data for the first 180 days only for the experiment. 
\end{itemize}
The first two datasets are downloaded from \url{http://www.sociopatterns.org}, and the last two from \url{https://snap.stanford.edu/data/index.html}. Table \ref{Tab:Data_Stat} gives a brief description of the datasets and Figure \ref{fig:dataset_description} shows the number of links present at each time stamp for the entire lifetime of the temporal network.

\begin{table}
\centering
\caption{Basic statistics of the datasets}
\label{Tab:Data_Stat}
    \begin{tabular}{ | p{2.5 cm} | p{1.3cm} | p{1.5cm} | p{2cm} | p{2cm} |}
    \hline
    Datasets & \#Nodes($n$) & \#Links($m$) & \#Static Edges & Lifetime/Total Duration \\ \hline
    Infectious & 410 & 17298 & 2765 & 8 Hours \\ \hline
    Hypertext & 113 & 20818 & 2196 & 2.5 Days \\ \hline
    College Message & 1899 & 59835 & 20296 & 193 Days \\ \hline
    AS180 & 4002  & 2127983 & 8957 & 180 Days \\ 
    \hline
    \end{tabular}
\end{table}

\begin{figure}
    \centering
\begin{tabular}{cc}
    \includegraphics[scale=0.37]{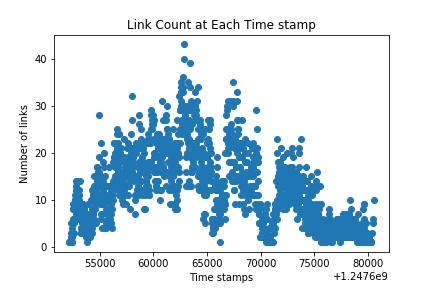}  & 
    \includegraphics[scale=0.37]{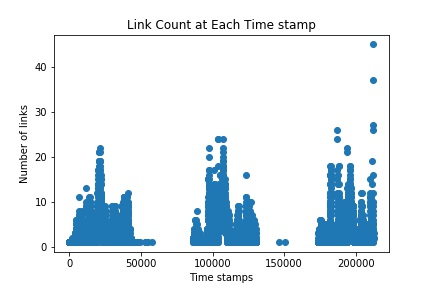} \\
    (a) Infectious & (b) Hypertext \\
     \includegraphics[scale=0.37]{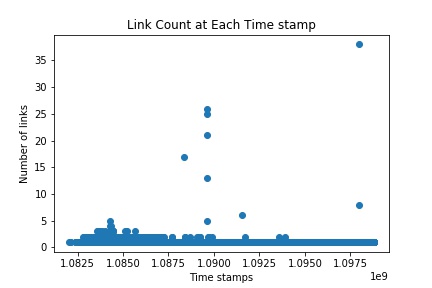} & 
     \includegraphics[scale=0.37]{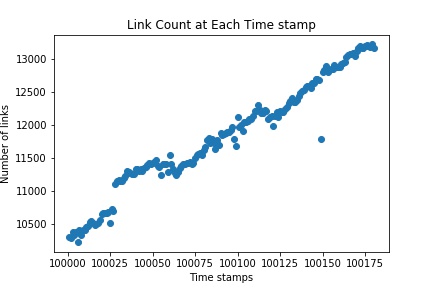} \\
     (c) College Message & (d) AS180 \\
\end{tabular}
   \caption{The link count at each time $t$ in the entire life-cycle of the datasets}
    \label{fig:dataset_description}
\end{figure}

\subsection{Experimental Setup} \label{SubSec:Exp_Set_Up}
Here, we describe the set up for our experimentation. Only the set up required in our experiments is to partition the dataset. We use the following two partitioning technique:
\begin{itemize}
    \item \textbf{Uniform Time Interval\mbox{-}Based Partitioning (EOC-UT)}: In this technique, the whole dataset is splitted into parts, where in each part the number of time stamps are equal.
    \item \textbf{Uniform Link Count\mbox{-}Based Partitioning (EOC-ULC)}: In this technique, the whole dataset is splitted into parts, where the number of time stamps in each part are equal.
\end{itemize}
Here, EOC stands for the proposed `Edge on Clique' procedure. We perform the experiments by making two partitions of the datasets, based on the mentioned partitioning schemes. We implement our proposed methodology in Python 3.6 along with NetworkX 2.2 environment, and all the experiments have been carried out in a 32-core server with 256GB RAM and 2.2 GHz processing speed.

\subsection{Goals of the Experiments} \label{Sec:GE}
The goal of the experimentation is to address the following research questions:
\begin{itemize}
    \item To understand the change in the size of different clique sets used in Algorithm \ref{Algo:1}, with respect to the change of $\Delta$ and $\gamma$ value.
    \item To understand the change in computational time and space required \emph{with} and \emph{without partition} for different partition schemes, with respect to the change of $\Delta$ and $\gamma$ value. 
    \item To understand the change in computational time and space with respect to the number of partitions.
\end{itemize}
All the experiments are done in two ways; (i) changing $\Delta$, with fixed $\gamma$ value, and (ii) varying $\gamma$, with fixed $\Delta$ value.

\subsection{Experimental Results with Discussion} \label{Sec:ER}
Here, we report the experimental results with detailed analysis, for the identified goals.
\subsubsection{Change in the Size of Different Clique Sets}
\begin{figure}
    \centering
    \begin{tabular}{cc}
       \includegraphics[scale=0.2]{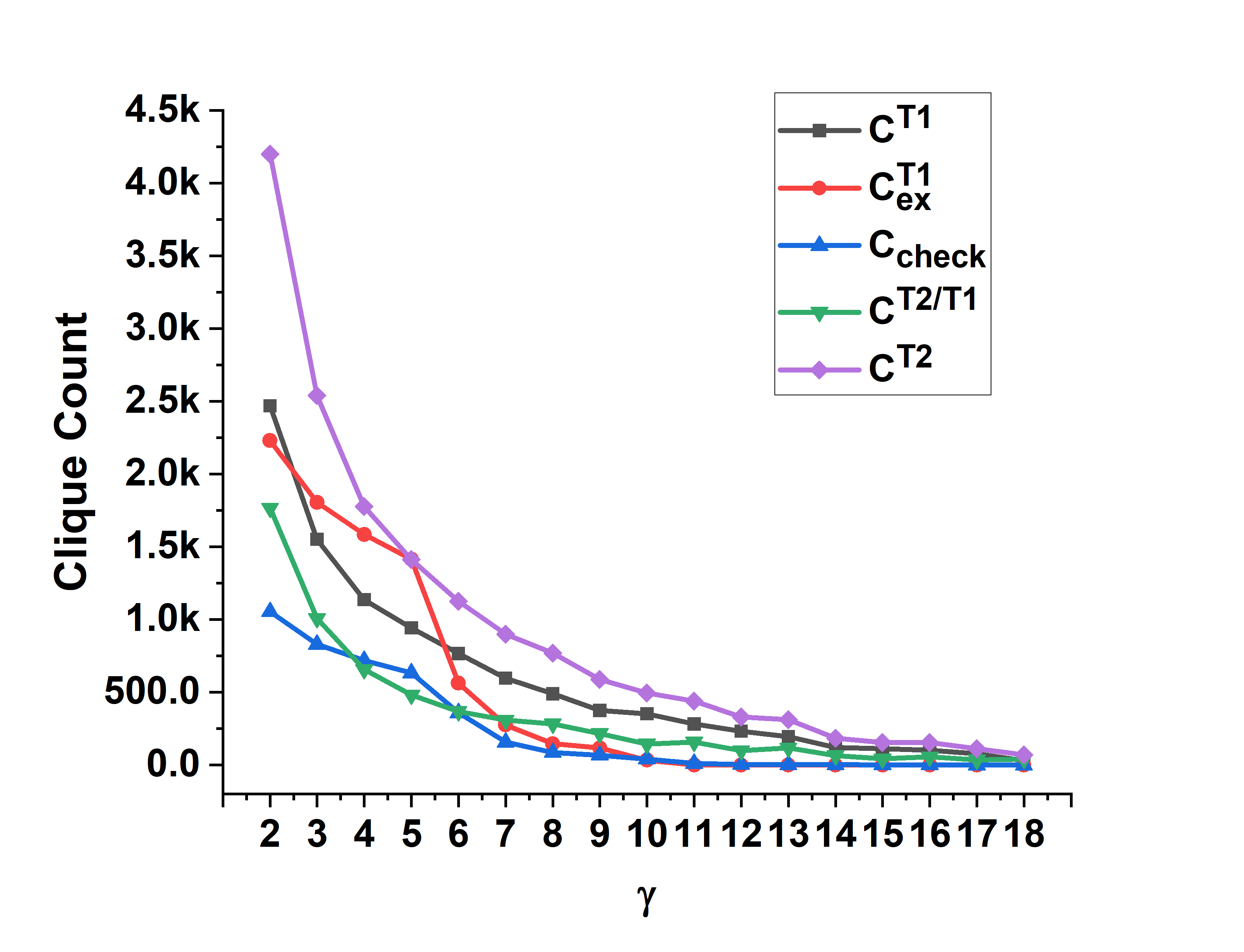}  & \includegraphics[scale=0.2]{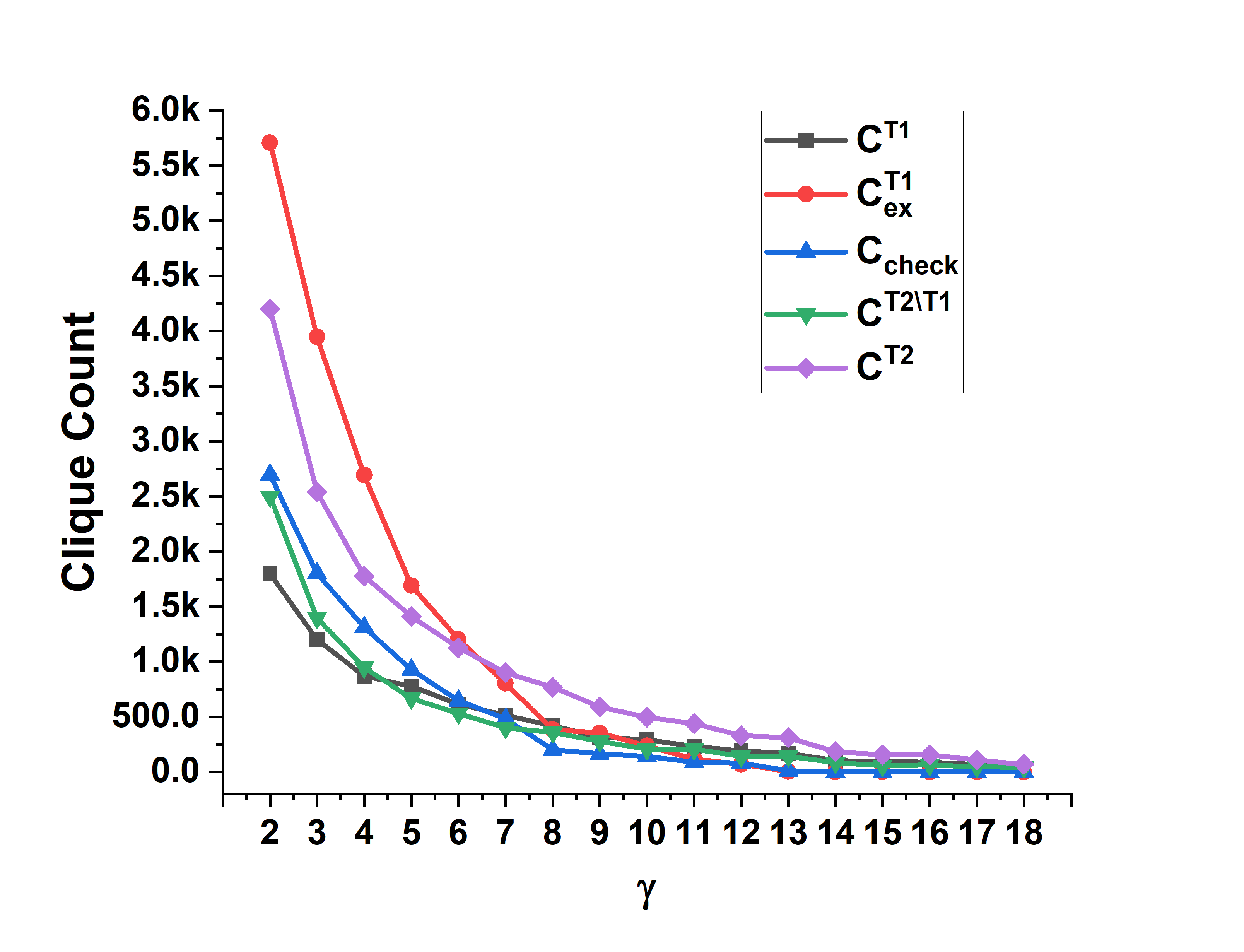}  \\
       (a) Uniform Interval-Based & (b) Uniform Link Count-Based \\
       \includegraphics[scale=0.2]{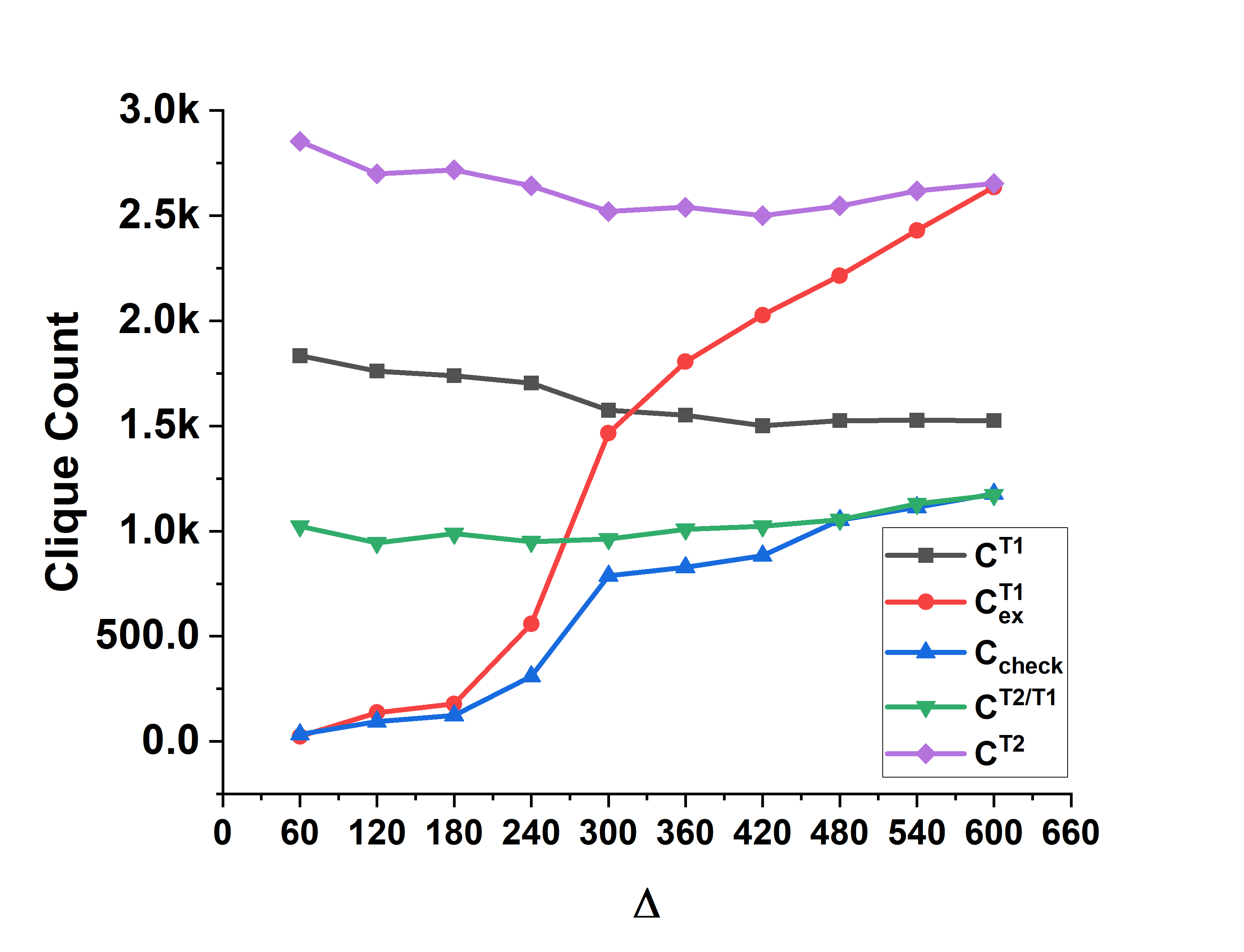}  & \includegraphics[scale=0.2]{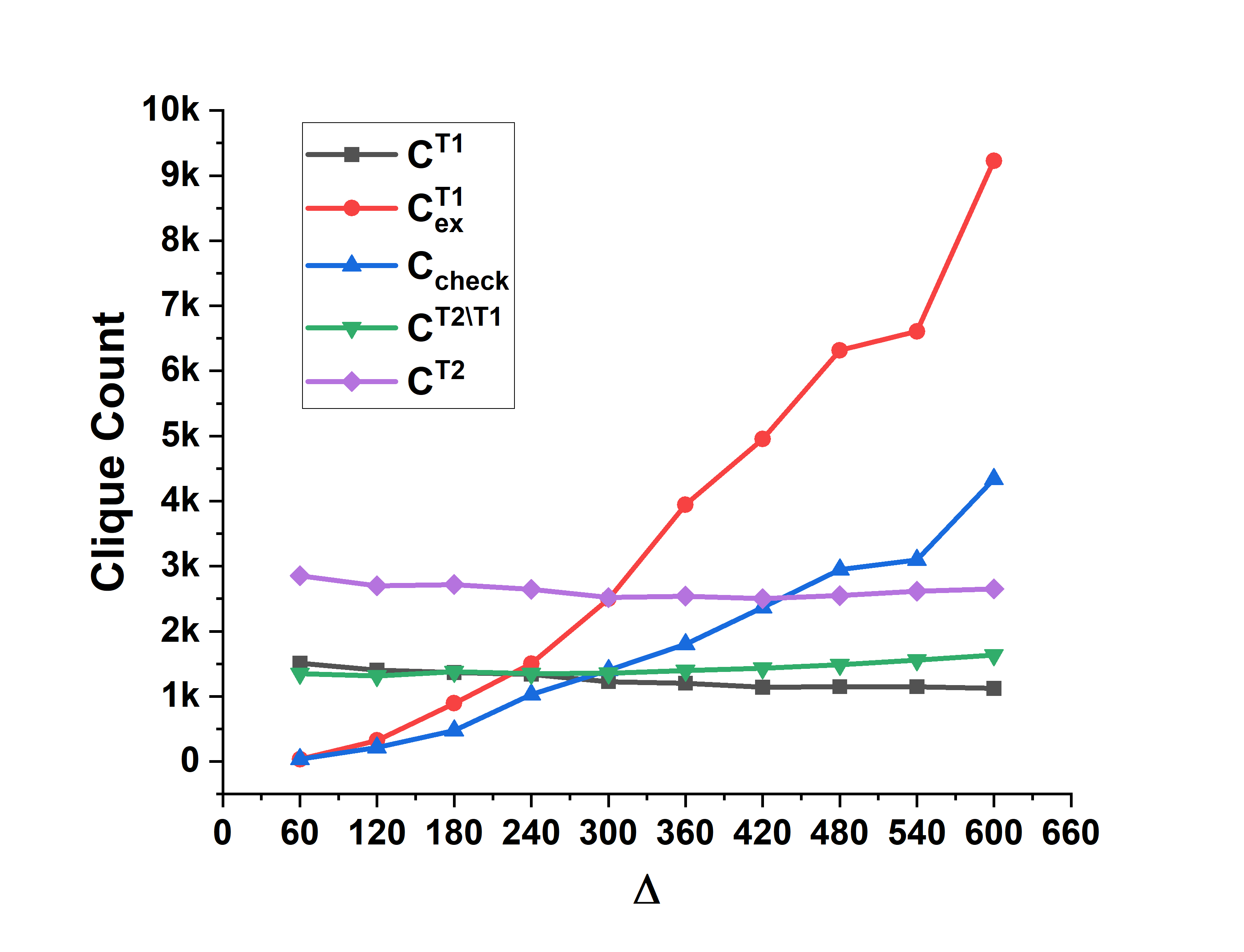} \\
       (c) Uniform Interval-Based & (d) Uniform Link Count-Based 
    \end{tabular}
    \caption{Result for the change of clique count w.r.t $\Delta$ and $\gamma$ for Infectious Dataset; (a)-(b) fixed $\Delta = 360$; (c)-(d) infectious fixed $\gamma = 3$ }
    \label{fig:Infectious}
\end{figure}

\begin{figure}
    \centering
    \begin{tabular}{cc}
       \includegraphics[scale=0.2]{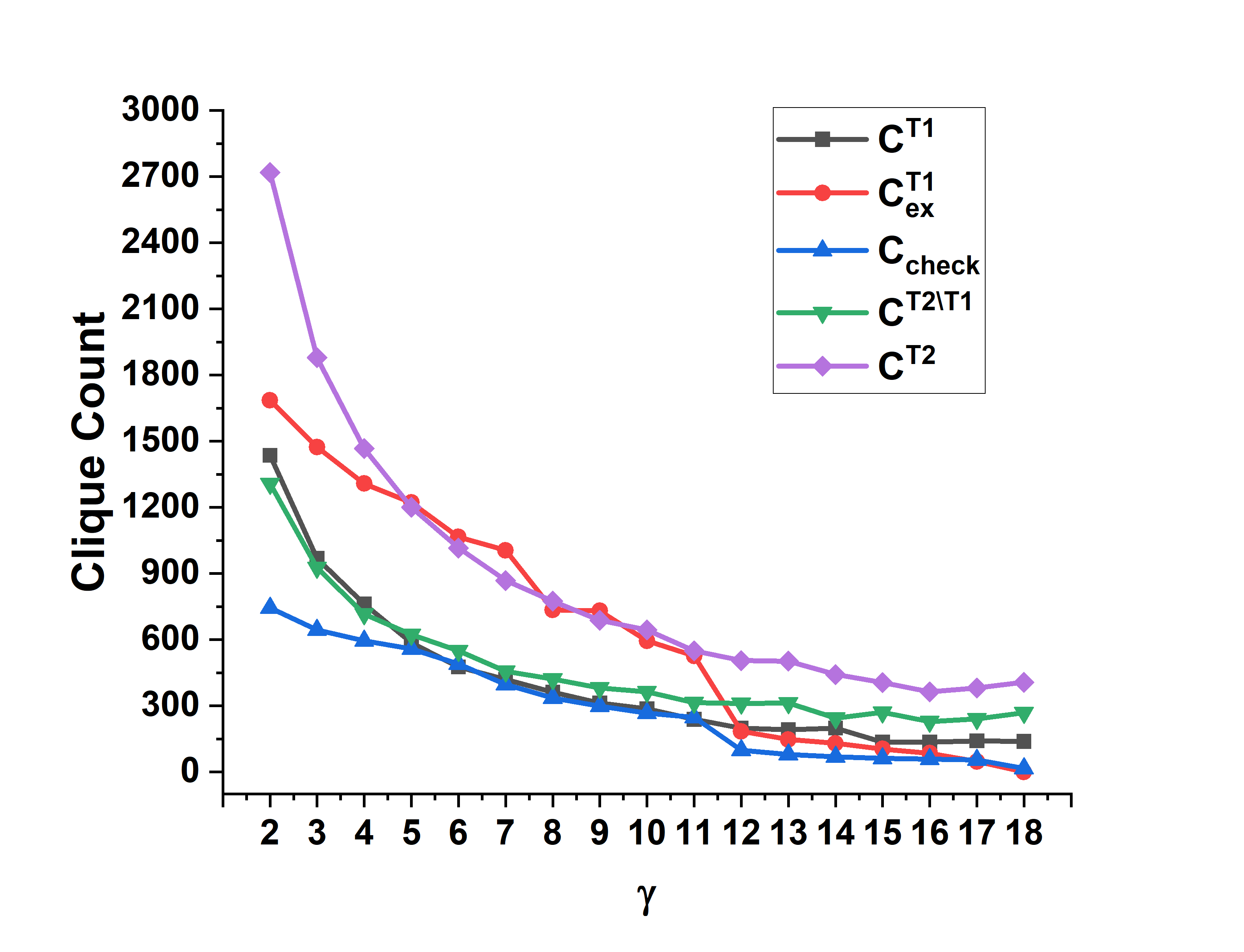}  & \includegraphics[scale=0.2]{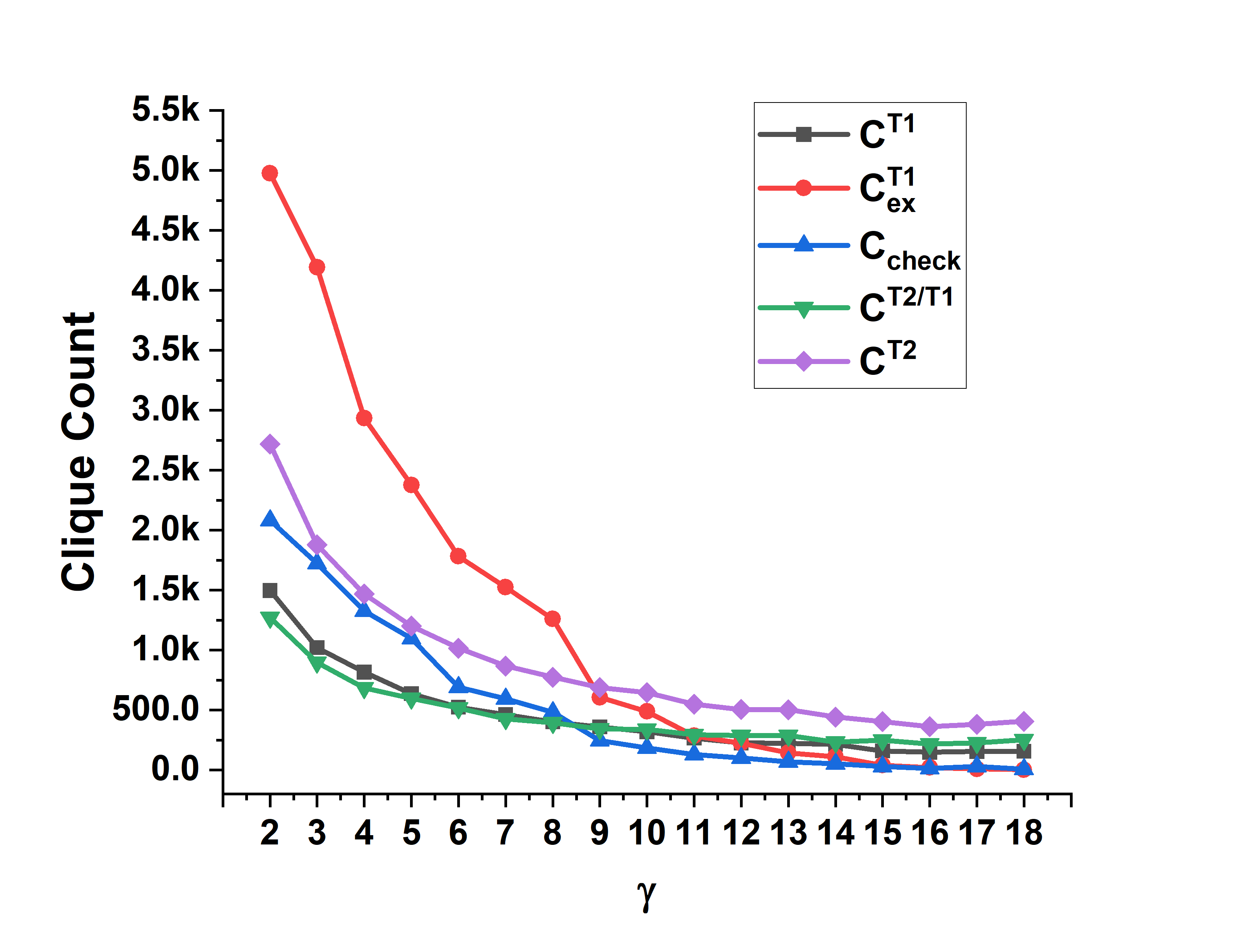}  \\
       (a) Uniform Interval-Based & (b) Uniform Link Count-Based \\
       \includegraphics[scale=0.2]{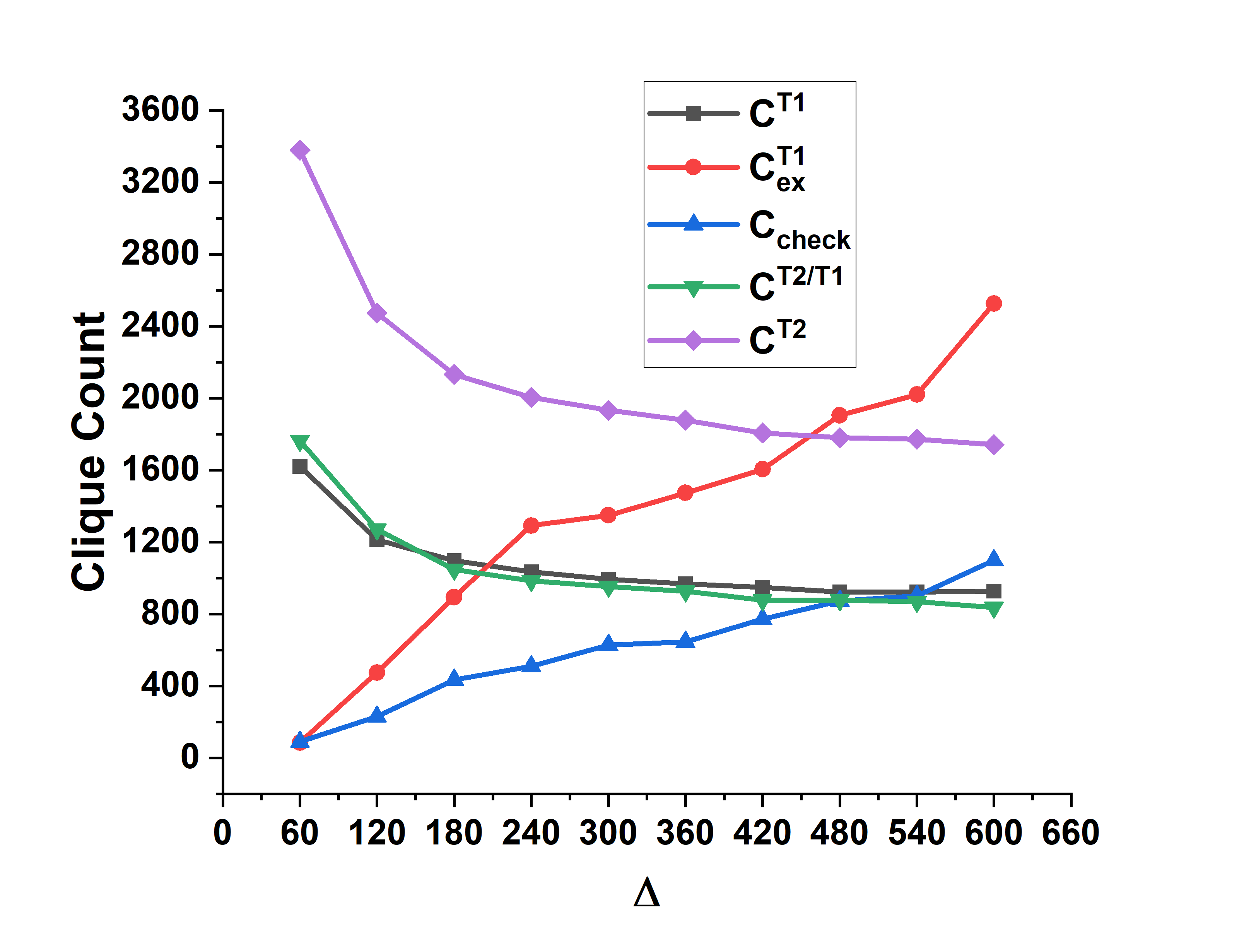}  & \includegraphics[scale=0.2]{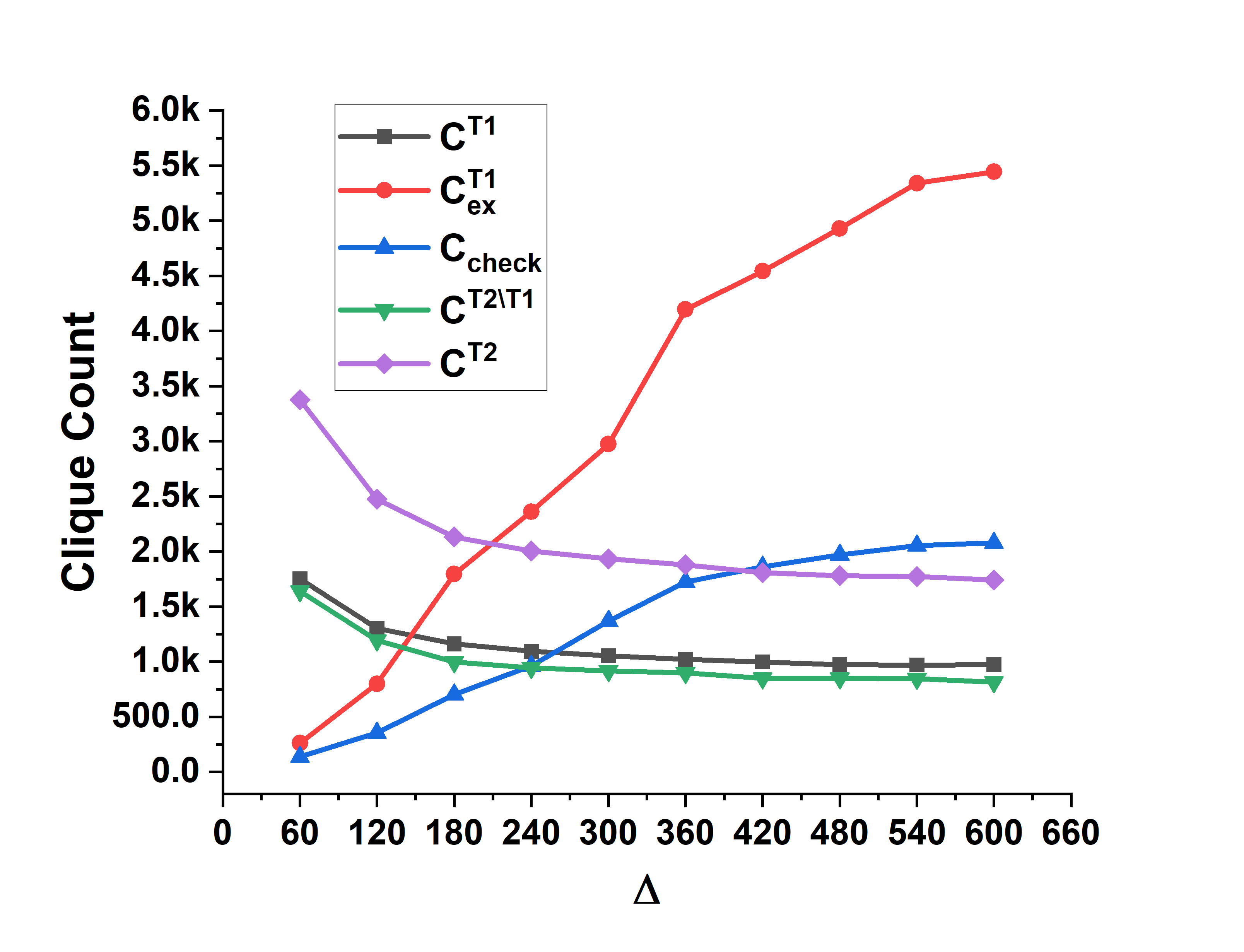} \\
       (c) Uniform Interval-Based & (d) Uniform Link Count-Based 
    \end{tabular}
    \caption{Result for the change of clique count w.r.t $\Delta$ and $\gamma$ for Hypertext Dataset; (a)-(b) fixed $\Delta = 360$; (c)-(d) Hypertext fixed $\gamma = 3$ }
    \label{fig:ht_cintact}
\end{figure}

 Figure \ref{fig:Infectious} a, and b show the plots for the change in cardinality of the clique sets $\mathcal{C}^{T_{1}}$, $\mathcal{C}_{ex}^{T_{1}}$, $\mathcal{C}^{T_{2} \setminus T_{1}}$, $\mathcal{C}^{T_{2}}$, $\mathcal{C}_{check}$ with the change in $\gamma$ for a fixed $\Delta$ (we consider $\Delta=360$), for the Infectious dataset with the uniform interval\mbox{-}based, and  uniform link count\mbox{-}based partitioning, respectively. Figure \ref{fig:Infectious} c, and d show the plots for the same in change with $\Delta$ for a fixed $\gamma$ (we consider $\gamma=3$). From the Figure \ref{fig:Infectious} a, and b it has been observed that in both the partitioning schemes, for a fixed $\Delta$, when the $\gamma$ value is increased, the cardinality of $\mathcal{C}^{T_{1}}$, $\mathcal{C}_{ex}^{T_{1}}$, $\mathcal{C}^{T_{2} \setminus T_{1}}$, $\mathcal{C}^{T_{2}}$, $\mathcal{C}_{check}$ are decreasing. The reason behind this is quite intuitive. For a fixed duration, if the frequency of contacts increases then certainly, the number of maximal cliques following this requirement is going to decrease. The decrements are very sharp till $\gamma=8$. From $\gamma=9$ to $13$ the decrements are quite gradual, and beyond $\gamma \geq 14$, the change is very less. As an example, for uniform interval\mbox{-}based partitioning scheme, when  $\gamma=2$, the value of $|\mathcal{C}^{T_{2}}|$ is $4199$ and the same for $\gamma=8$ is $569$. However, the value of $|\mathcal{C}^{T_{2}}|$ for $\gamma=9$ and $13$ are $589$ and $311$ and also for $\gamma=14$ and $18$ are $185$ and $69$. 
\par From Figure \ref{fig:Infectious} c and d, it can be observed that for a fixed $\gamma$ ($=3$ in our experiments) if the $\Delta$ value is increased gradually, the change in cardinality of $\mathcal{C}^{T_{1}}$, $\mathcal{C}^{T_{2}}$, $\mathcal{C}^{T_{2} \setminus T_{1}}$ are decreasing very slightly. As an example, for uniform time interval based partitioning scheme when the value of $\Delta =60$, $360$ and $600$, the value of $\mathcal{C}^{T_{1} }$ are $1834$, $1552$, and $1526$, respectively. The reason behind this is as follows. For a given frequency of contacts, when the duration is increased, there is a high possibility that the two or more maximal cliques get merged and this leads to the decrement in maximal clique count.  However, the gradual increment of $\Delta$ value leads to the increase of $|\mathcal{C}_{ex}^{T_{1}}|$ and $|\mathcal{C}_{check}|$. The increment of $|\mathcal{C}_{ex}^{T_{1}}|$ is sharper than that of $|\mathcal{C}_{check}|$, particularly for uniform link count\mbox{-}based partitioning. As an example, for uniform interval\mbox{-}based partitioning, when $\Delta=60$ the value of $|\mathcal{C}_{check}|$ and $|\mathcal{C}_{ex}^{T_{1}}|$ is $33$ and $24$, respectively. However, the same for $\Delta=660$, are $1178$ and $2636$. In case of uniform link count\mbox{-}based partitioning,  the value of $|\mathcal{C}_{check}|$ and $|\mathcal{C}_{ex}^{T_{1}}|$ for $\Delta=60$ are $39$ and $36$, respectively. However, when the $\Delta$ value is increased to $600$, these become $4337$ and $9228$, respectively.

\begin{figure}
    \centering
    \begin{tabular}{cc}
       \includegraphics[scale=0.2]{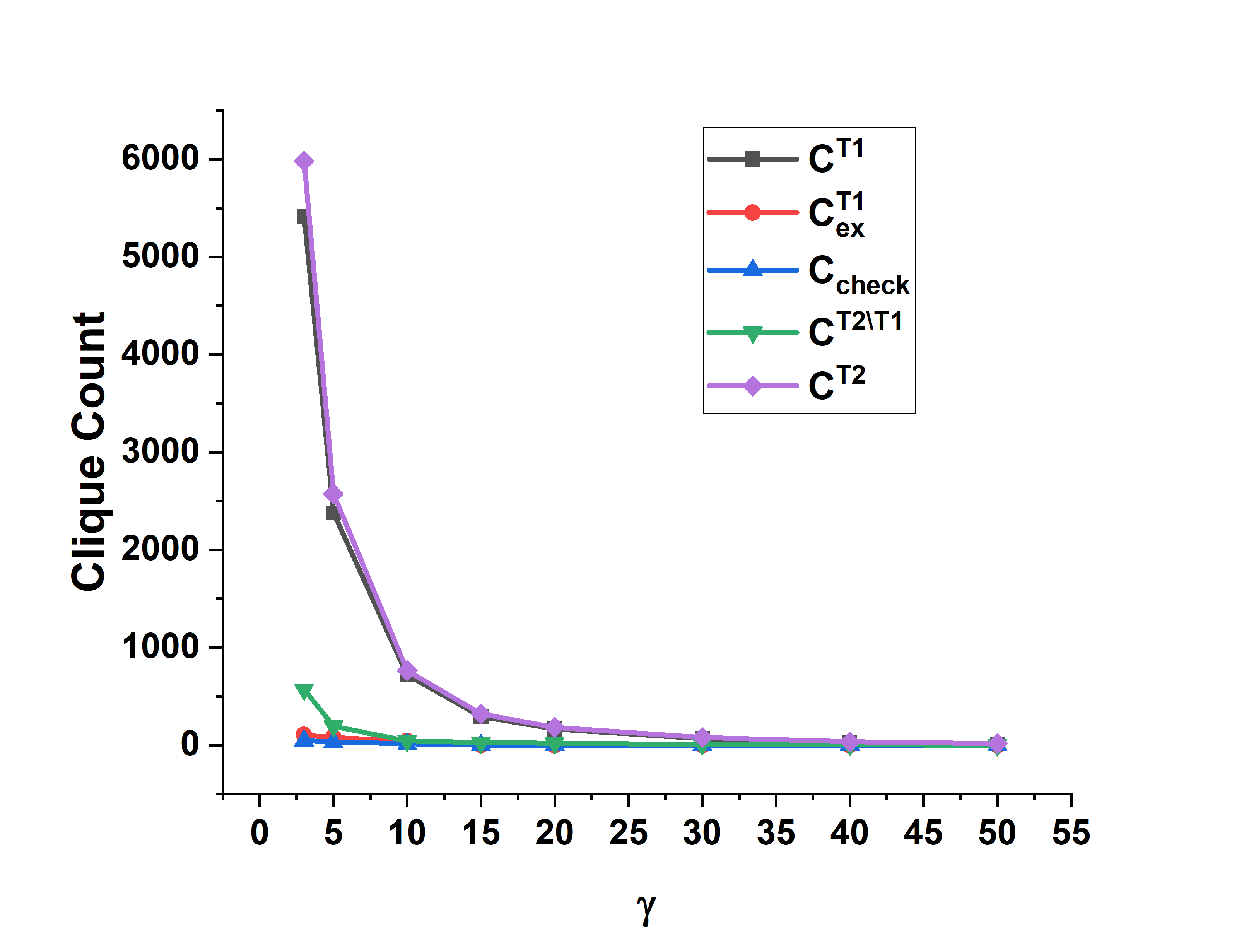}  & \includegraphics[scale=0.2]{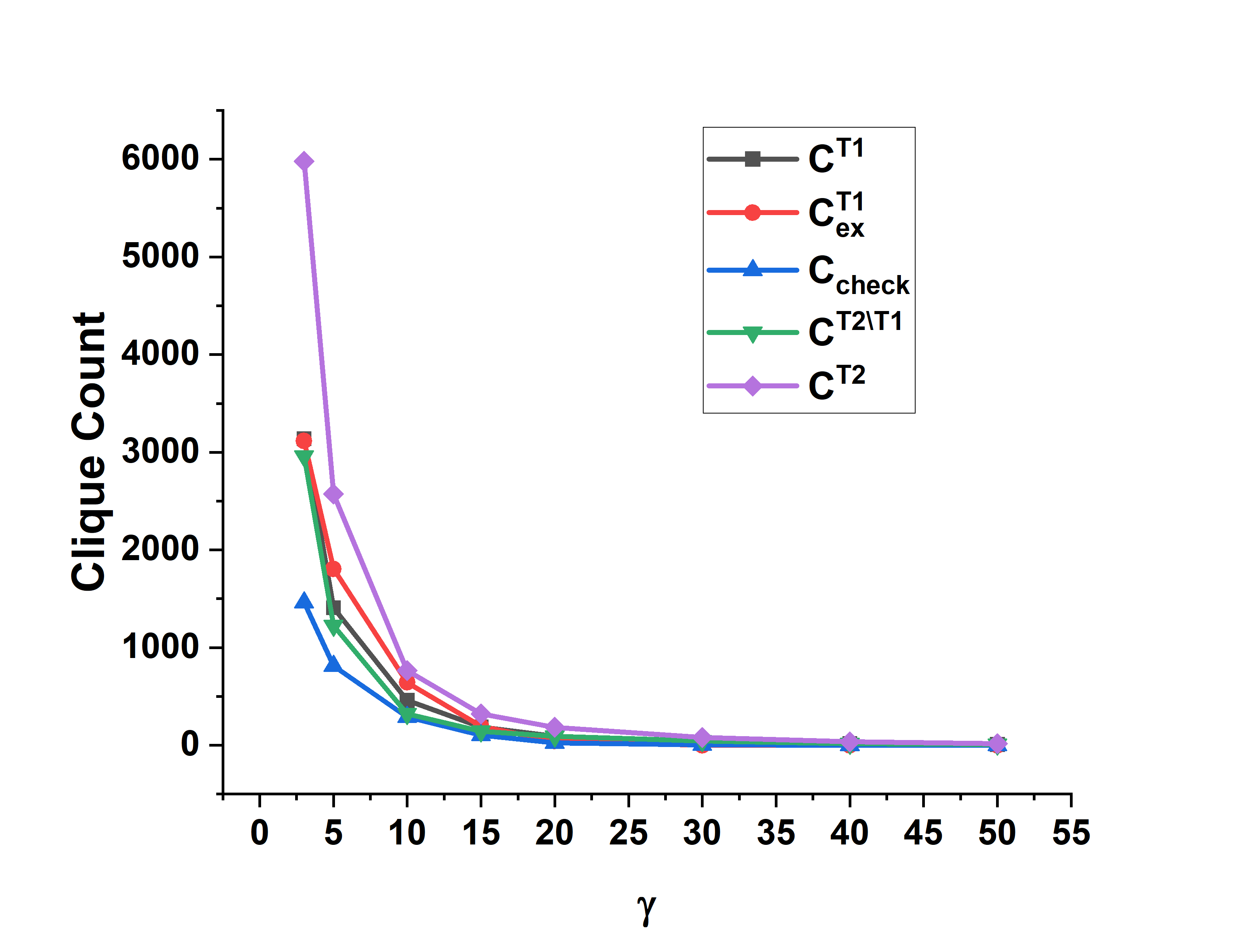}  \\
       (a) Uniform Interval-Based & (b) Uniform Link Count-Based \\
       \includegraphics[scale=0.2]{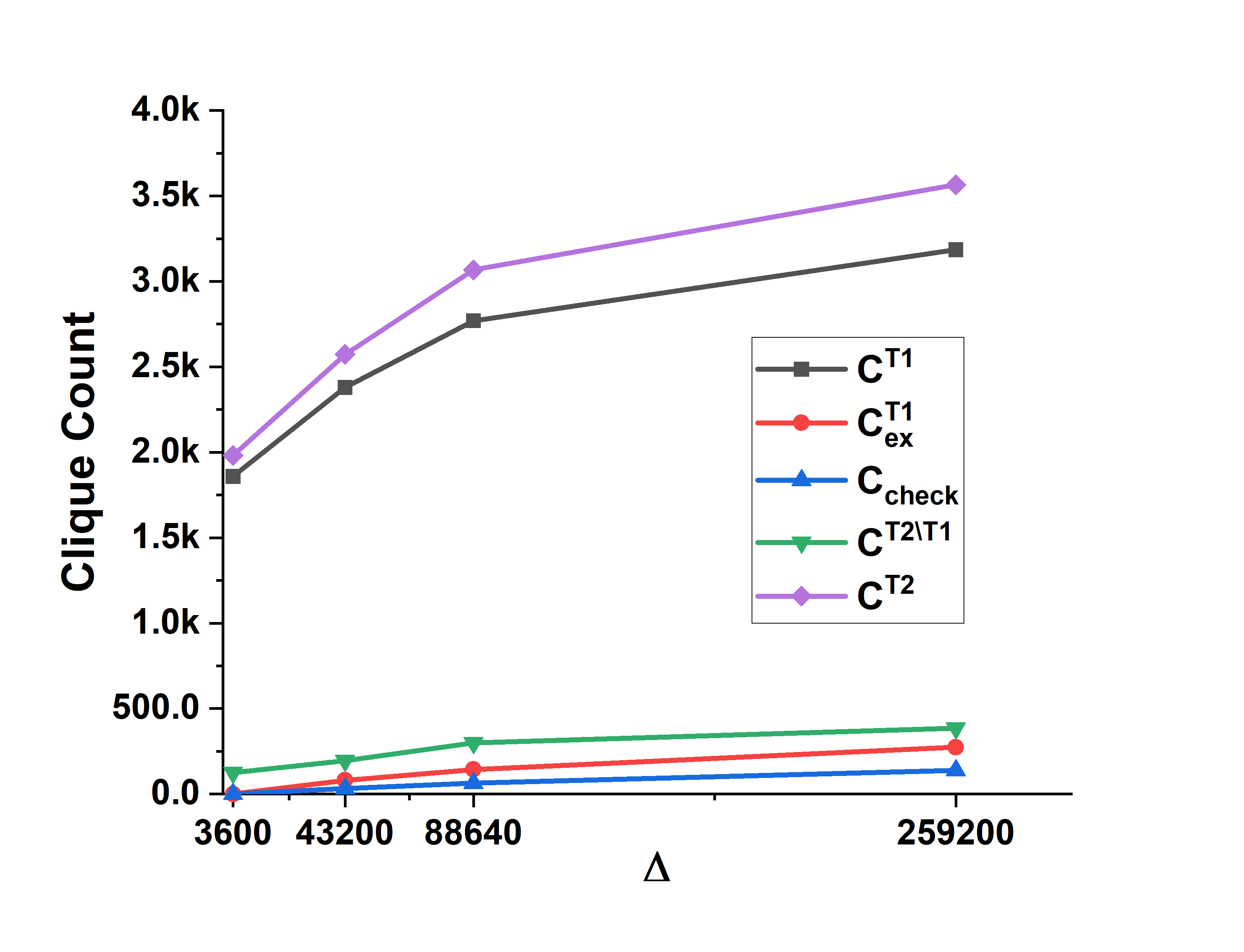}  & \includegraphics[scale=0.2]{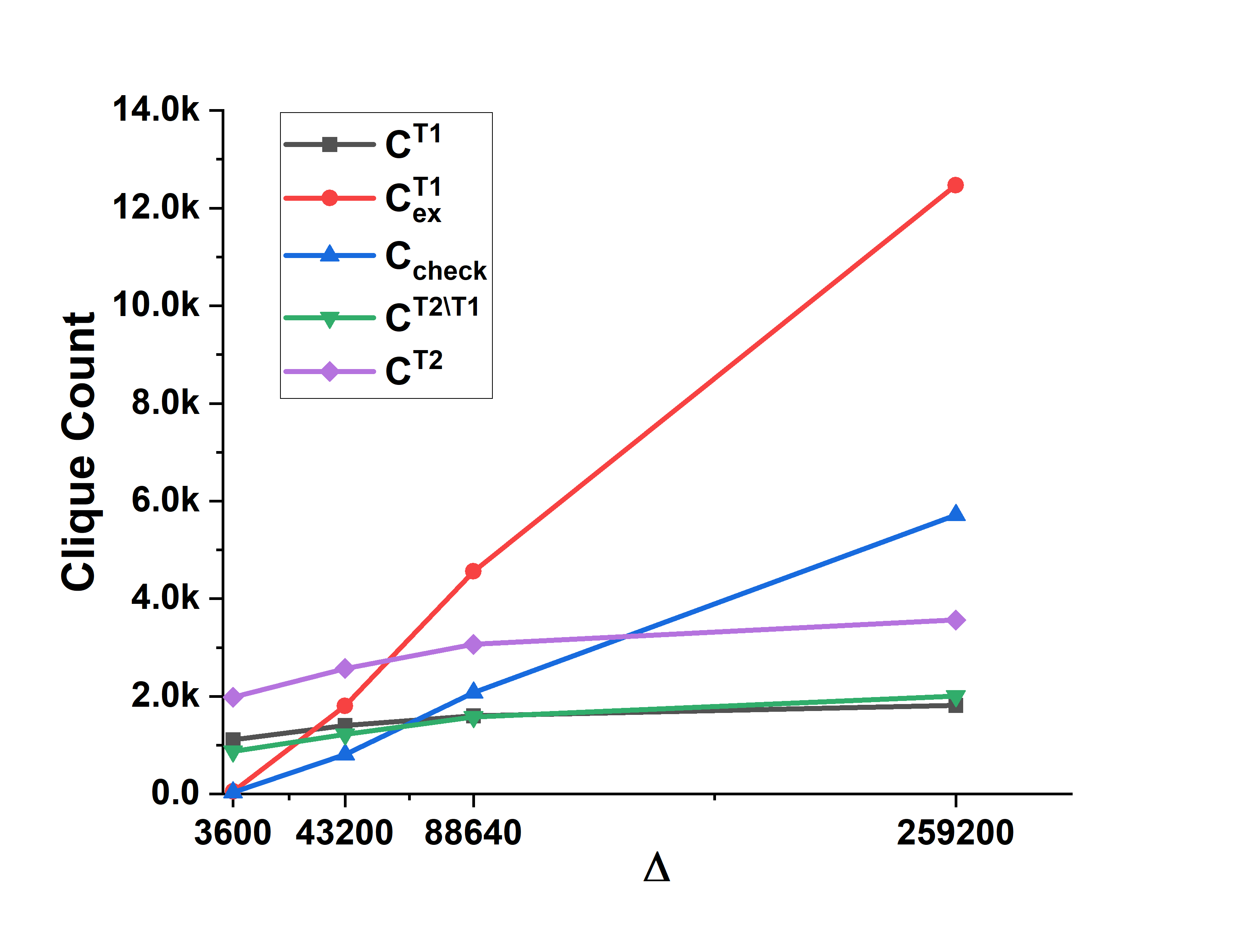} \\
       (c) Uniform Interval-Based & (d) Uniform Link Count-Based 
    \end{tabular}
    \caption{Result for the change of clique count w.r.t $\Delta$ and $\gamma$ for College Message Dataset; (a)-(b) fixed $\Delta = 43200$; (c)-(d) Collegemsg fixed $\gamma = 5$ }
    \label{fig:collegemsg}
\end{figure}

Figure \ref{fig:ht_cintact} a and b show the plots for the change in cardinality of $\mathcal{C}^{T_{1}}$, $\mathcal{C}_{ex}^{T_{1}}$, $\mathcal{C}^{T_{2} \setminus T_{1}}$, $\mathcal{C}^{T_{2}}$, $\mathcal{C}_{check}$ with the change in $\gamma$ for a fixed $\Delta$ for the Hypertext dataset. In this dataset also, we conduct our experiments with $\Delta=360$. For these two plots, our observations are same as the `Infectious' dataset, i.e., with the increment of $\gamma$, the cardinality of all these clique sets are gradually decreasing. As an example, for uniform time interval\mbox{-}based partitioning, when $\gamma=2$, $10$, and $18$ the value of $|\mathcal{C}^{T_{2} \setminus T_{1}}|$ are $1307$, $364$, and $268$, respectively. Figure \ref{fig:ht_cintact} c and d show the plots change in cardinality of all the lists with the increment of $\Delta$ for a fixed $\gamma$ value. Like Infectious dataset, in this dataset also we consider $\gamma=3$. In both the partitioning schemes, with the increment of $\Delta$, the value of $|\mathcal{C}^{T_{1}}|$, $|\mathcal{C}^{T_{2}}|$, and $|\mathcal{C}^{T_{2} \setminus T_{1}}|$ decreases, however, $|\mathcal{C}_{ex}^{T_{1}}|$ and $|\mathcal{C}_{check}|$ increases. As an example, when $\Delta=60$, the value of $|\mathcal{C}^{T_{1}}|$, $|\mathcal{C}^{T_{2}}|$, and $|\mathcal{C}^{T_{2} \setminus T_{1}}|$ are $1621$, $1765$, and $3378$, respectively. When the value of $\Delta$ has been increased to $600$ their values are $928$, $837$, and $1742$, respectively. However, when the value of $\Delta=60$ and  $\Delta=600$, the value of $|\mathcal{C}_{ex}^{T_{1}}|$ and $|\mathcal{C}_{check}|$ are $86$, $91$ and $2527$, $1099$, respectively. 

Similar to Infectious and Hypertext, Figure \ref{fig:collegemsg} shows the results for College Message dataset. Due to the large life-cycle of the dataset, we select higher value of $\Delta$ (= 43200 sec or 12 hours) compared to the earlier. Figure  \ref{fig:collegemsg} a and b show the plots for the change in cardinality of $\mathcal{C}^{T_{1}}$, $\mathcal{C}_{ex}^{T_{1}}$, $\mathcal{C}^{T_{2} \setminus T_{1}}$, $\mathcal{C}^{T_{2}}$, $\mathcal{C}_{check}$ with the change in $\gamma$ for a fixed $\Delta$. In both the partition schemes, the size of $\mathcal{C}_{ex}^{T_{1}}$, and $\mathcal{C}_{check}$ are less compared the number of maximal cliques in each partition. This helps to scale up the computation with number of partitions. Contradicting to Infectious and Hypertext, here, the clique counts increases with increasing $\Delta$ (Figure \ref{fig:collegemsg} c and d). In case of \emph{uniform link count partition}, the size of $\mathcal{C}_{ex}^{T_{1}}$ becomes 12470. It helps to exploit the partition mechanism and improve the computational efficacy of the algorithim. We will discuss this in detail in the next section.  

\begin{figure}
    \centering
    \begin{tabular}{cc}
      \includegraphics[scale=0.2]{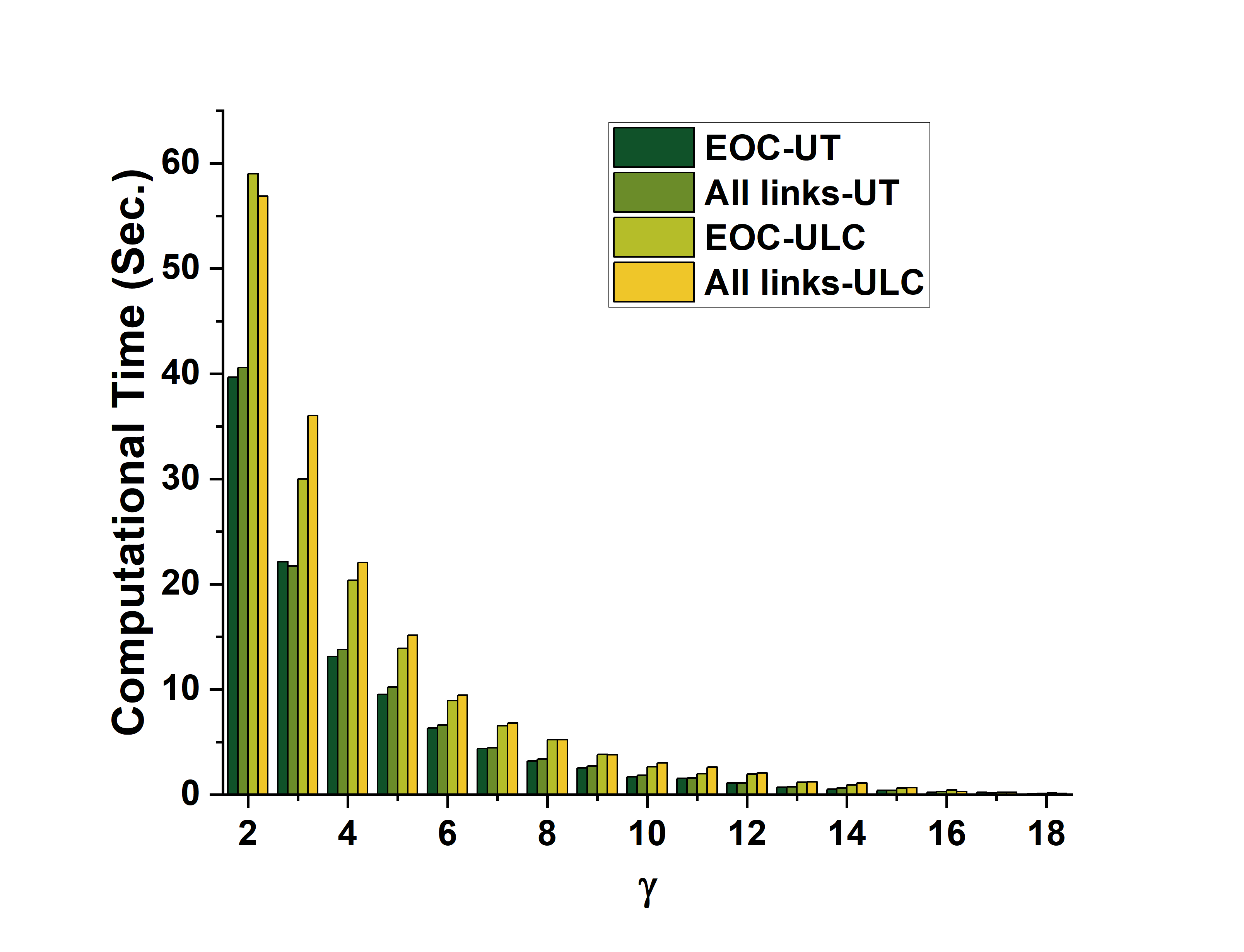}   &  \includegraphics[scale=0.2]{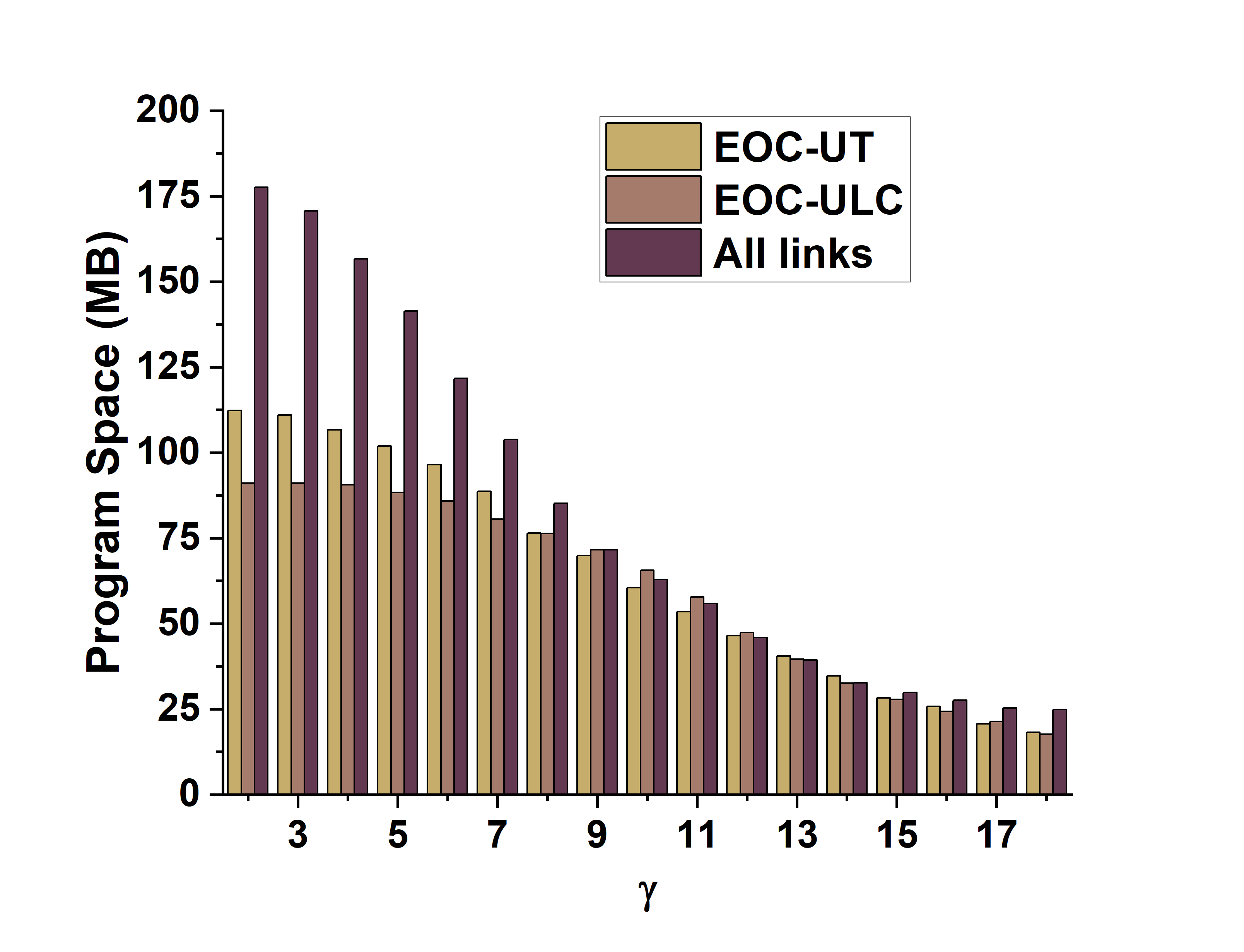}\\
        (a) Fixed - $\Delta=360$ time & (b) Fixed - $\Delta=360$ space  \\
    \includegraphics[scale=0.2]{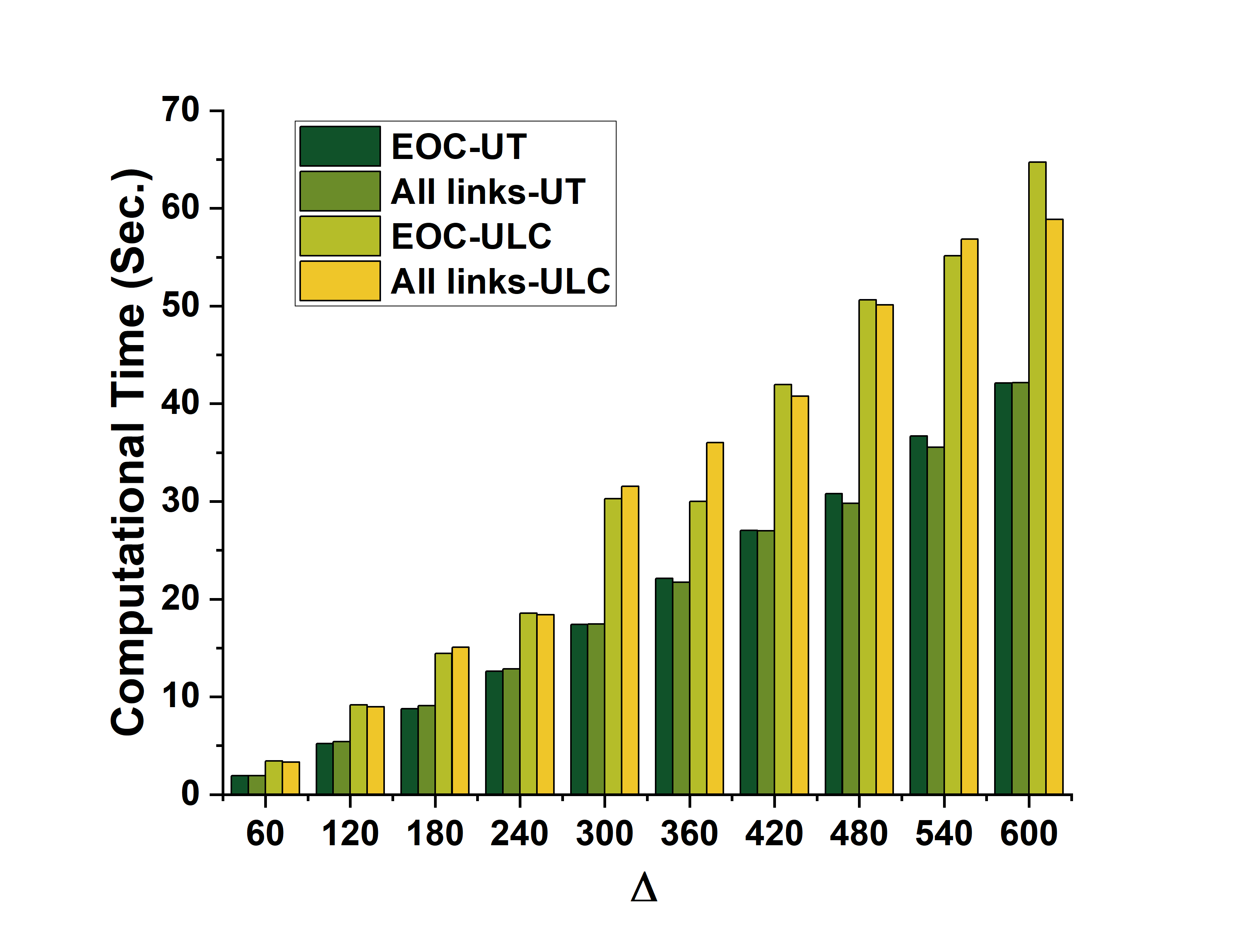}   &  \includegraphics[scale=0.2]{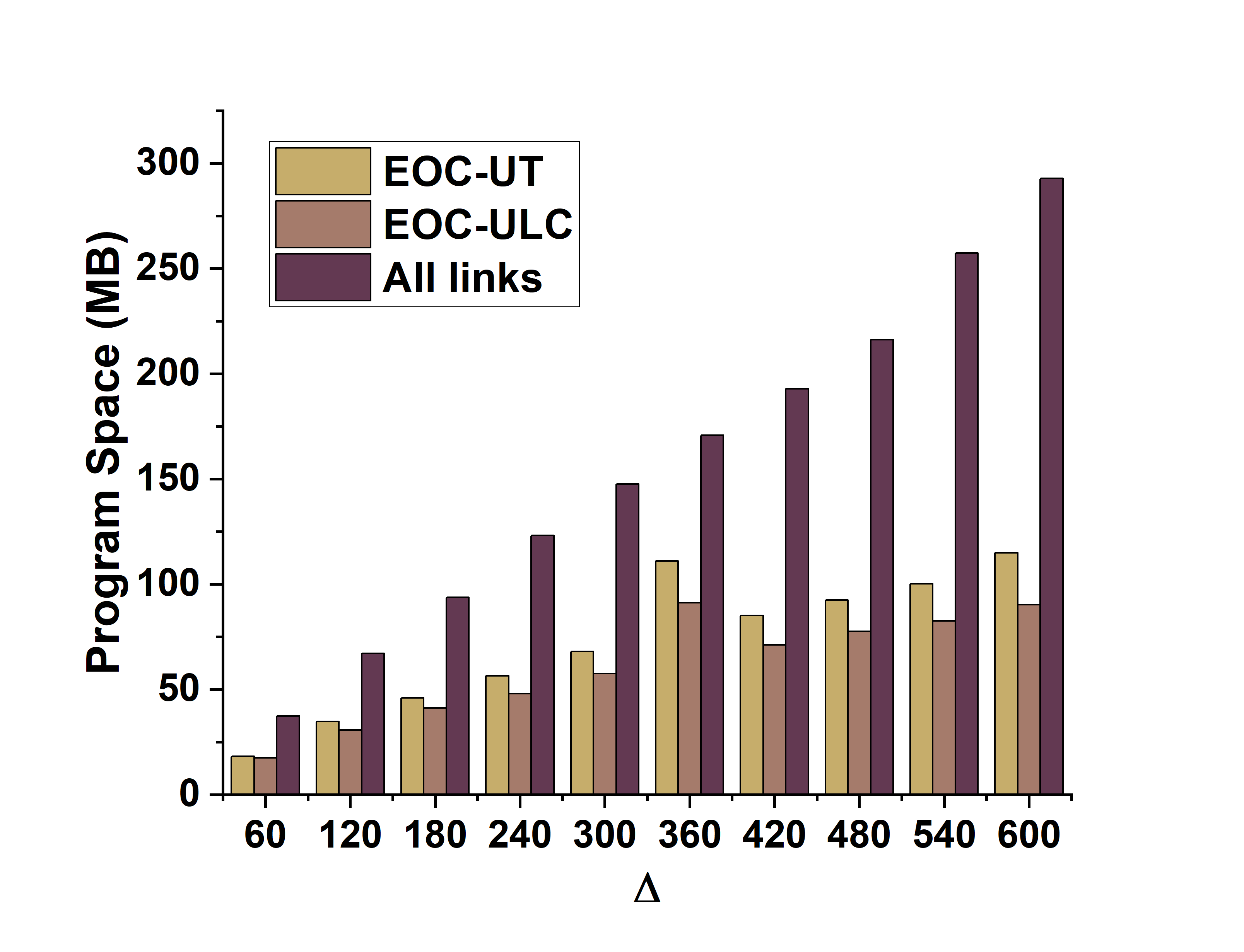}\\
        (c) Fixed - $\gamma=3$ time & (d) Fixed - $\gamma=3$ space 
    \end{tabular}
    \caption{Results for Computational Time and Space for Infectious dataset}
    \label{fig:Computational_Time_infectious}
\end{figure}

\begin{figure}
    \centering
    \begin{tabular}{cc}
      \includegraphics[scale=0.2]{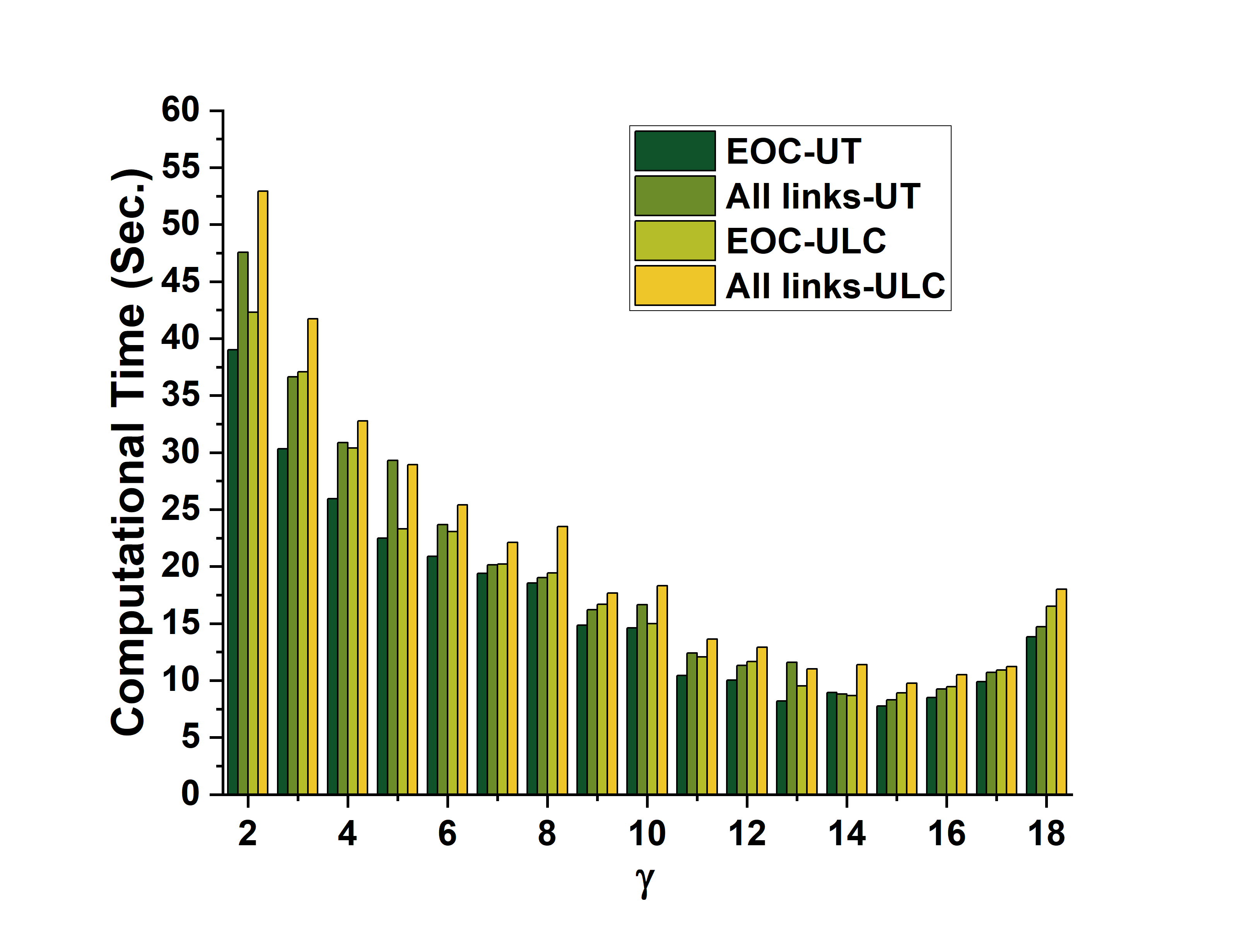}   &  \includegraphics[scale=0.2]{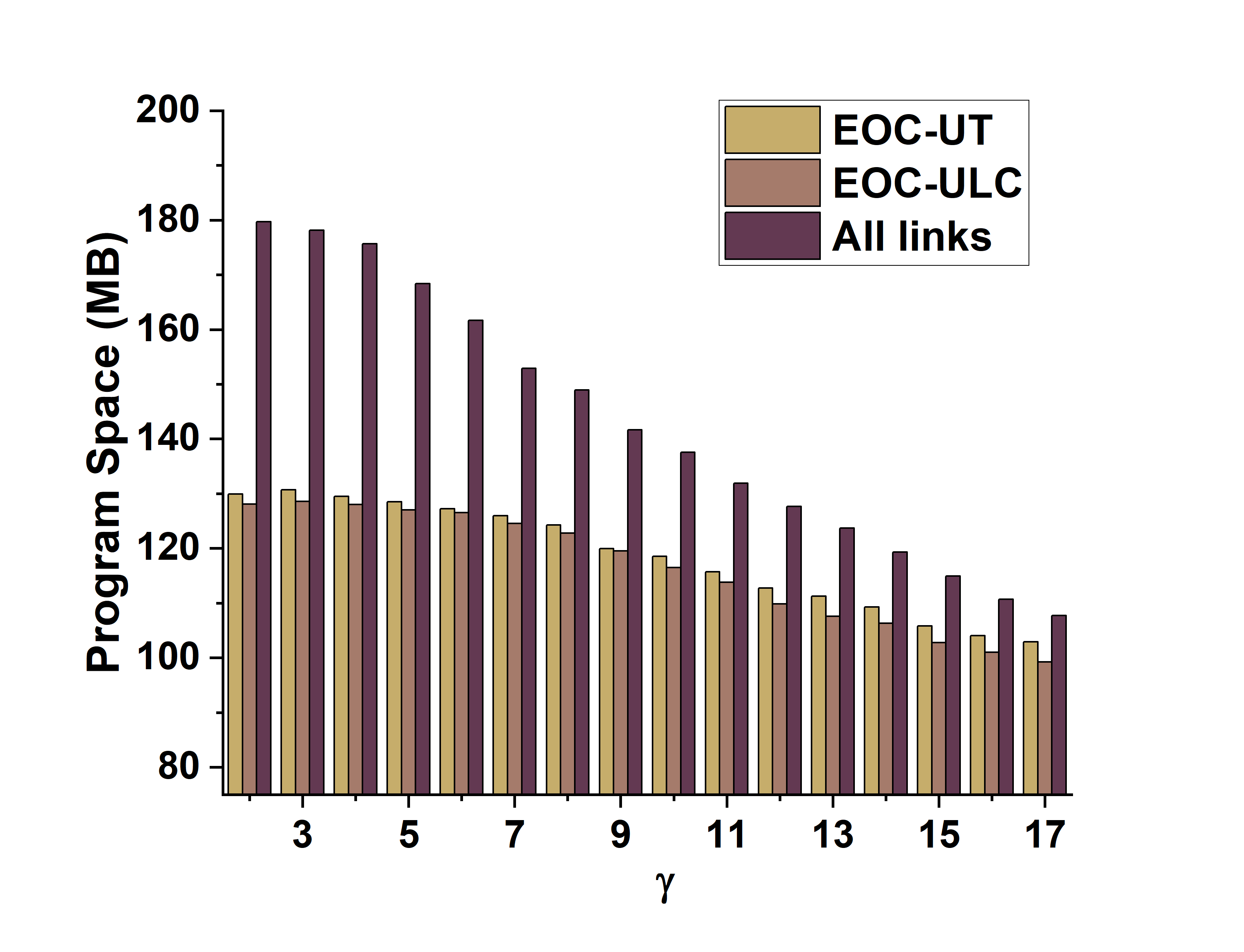} \\
        (a) Fixed - $\Delta=360$ time & (b) Fixed - $\Delta=360$ space  \\
    \includegraphics[scale=0.2]{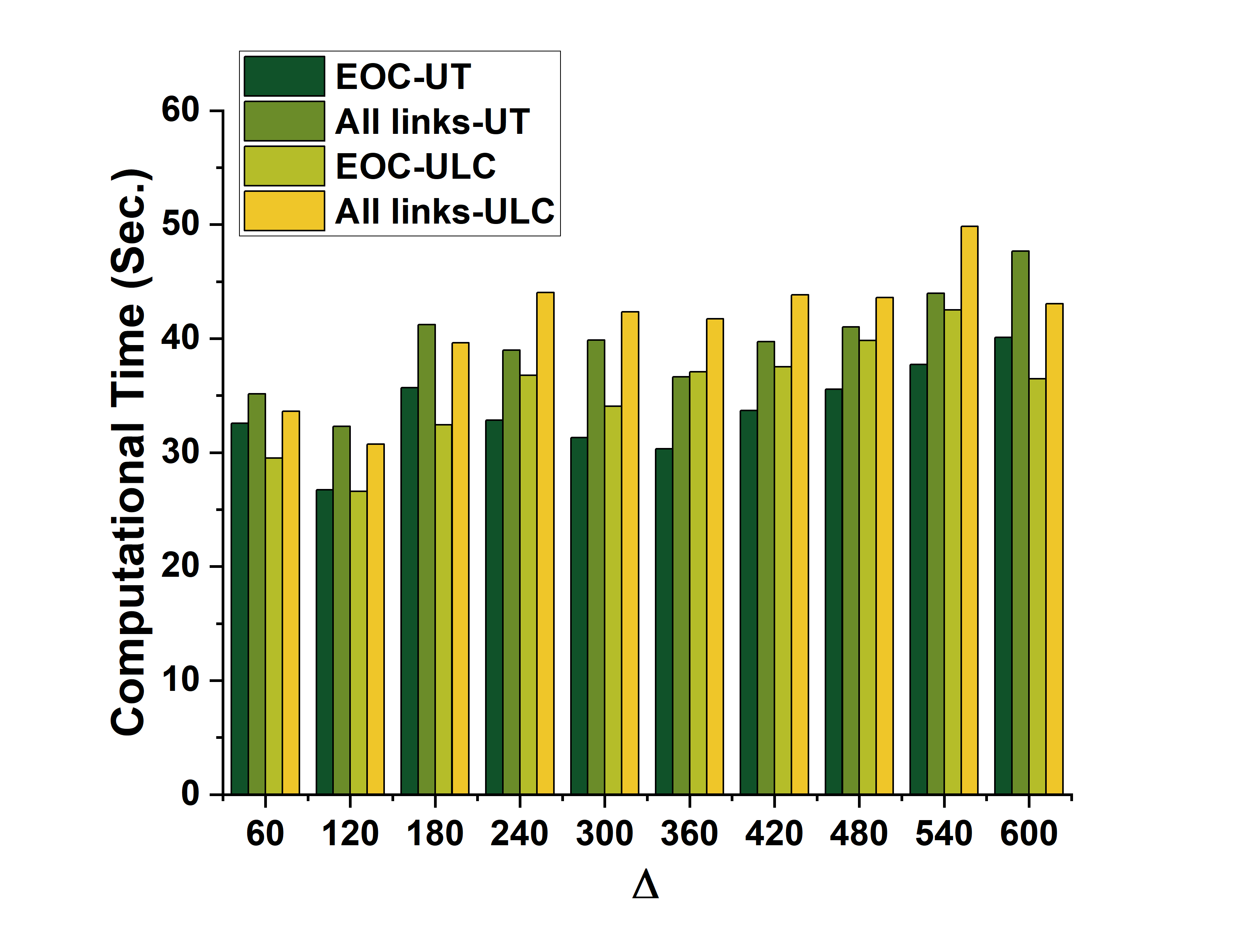}   &  \includegraphics[scale=0.2]{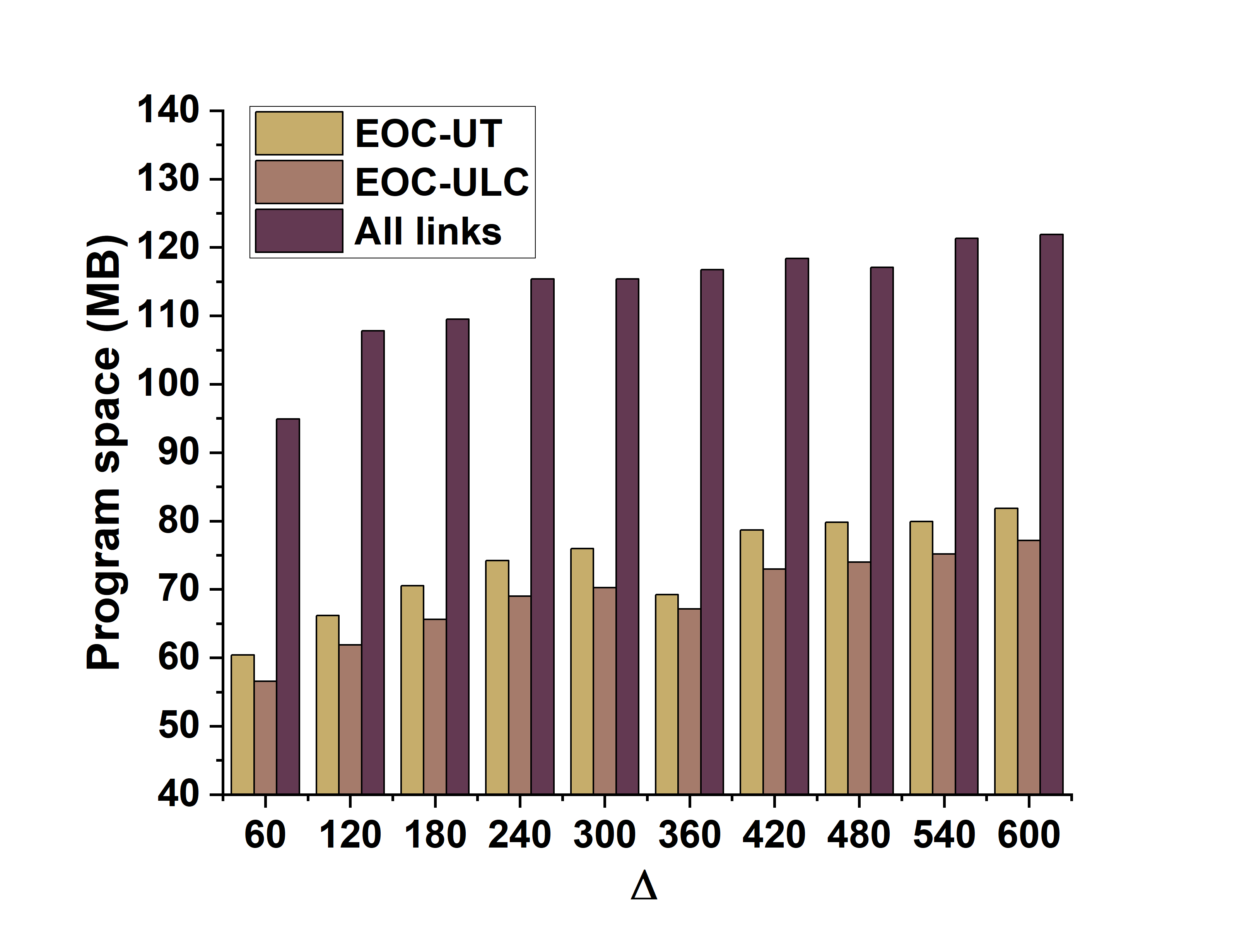}\\
           (c) Fixed - $\gamma=3$ time & (d) Fixed - $\gamma=3$ space 
    \end{tabular}
    \caption{Results for Computational Time and Space for Hypertext dataset}
    \label{fig:Computational_Time_ht09contact}
\end{figure}

\subsubsection{Computational Time and Space}
 Now, we turn our attention to discuss regarding time and space requirement of our proposed methodology. Note that, here we report the space taken by the process for execution. In case of partition, we report the total time to compute $\mathcal{C}^{T_2}$ by adding the time required to execute each partition. To compute the space, we report the maximum space required by among both the partitions. We continue our experiments in the setting as described (i.e., fixed $\Delta$ varying $\gamma$, and fixed $\gamma$ varying $\Delta$). Figures \ref{fig:Computational_Time_infectious}, \ref{fig:Computational_Time_ht09contact}, and  \ref{fig:Computational_Time_collegemsg} show the plots for the Infectious, Hypertext, College Message Datasets, respectively. For computational time, two comparing values, with or without partitions, are shown for each of the partition schemes (\emph{EOC-UT} and \emph{EOC-ULC}). For computational space, we consider three comparing values for \emph{EOC-UT}, \emph{EOC-ULC}, and with \emph{all links}. As we run all the experiments in high performance computing cluster with multiple programs running simultaneously, at the job submission, it allocates the resources (cores) according to it's availability and current workload. Hence, we run both the settings (with or without partition) in a same job to get the actual comparison. So, we report \emph{All links-UT} and \emph{All links-ULC} for computational time. 
\par Now, we discuss the computational time and space requirements for fixed $\Delta$ setting. From the Figure \ref{fig:Computational_Time_infectious} a, \ref{fig:Computational_Time_ht09contact} a, and  \ref{fig:Computational_Time_collegemsg} a, it can be observed that the computational time decreases exponentially with the increase of $\gamma$. The reason is with the increment of $\gamma$, the number of maximal cliques decreases (Refer \ref{fig:Infectious} a and b, \ref{fig:ht_cintact} a and b, \ref{fig:collegemsg} a and b). It is also observed that compared to the \emph{all links}, the partition\mbox{-}based schemes lead to an improvement in computational time and the improvement is more when the $\gamma$ value is less. This is because for lower values of $\gamma$, the size of $\mathcal{C}_{ex}^{T_{1}}$ is more, and as shown previously these cliques will only be expanded by extending its right time stamp. However, these cliques will be expanded by vertex addition, and both right and left time stamp expansion while processing all the links at a time. Hence, the larger size of $\mathcal{C}_{ex}^{T_{1}}$ leads to more improvement in computational time. For `Hypertext' dataset (Figure \ref{fig:Computational_Time_ht09contact} a ), the computational time increases for $\gamma \geq 16$ due to the growth in intermediate clique count, while building $\mathcal{C}^{T_{2} \setminus T_{1}}$. Now, we turn our attention for space requirement analysis in Figure \ref{fig:Computational_Time_infectious} b, \ref{fig:Computational_Time_ht09contact} b, and  \ref{fig:Computational_Time_collegemsg} b. For a fixed $\Delta$, when the value of $\gamma$ increases, the space requirement decreases. Also, the partition schemes lead to an improvement in terms of space compared to All-Links. As we have made two partitions, in majority of the cases the space requirement become approximately half for small value of $\gamma$. As the number of maximal cliques are comparatively more for small $\gamma$, the improvement is also significant. However, for large $\gamma$ improvement is negligible in Infectious and College Message (Figure \ref{fig:Computational_Time_infectious} b and \ref{fig:Computational_Time_collegemsg} b, respectively). Now, the space requirement by EOC-ULC is always less compared to EOC-UT. As explained earlier, the size of $\mathcal{C}^{T_1}_{ex}$ plays a major role to improve the computational time in partition-based scheme. Similarly, it effects the space requirement. Now, from the Figures \ref{fig:Infectious} a and b, \ref{fig:ht_cintact} a and b, \ref{fig:collegemsg} a and b, it can be noted that the size of $\mathcal{C}^{T_1}_{ex}$ is more in \emph{Uniform link count\mbox{-}based} compared to \emph{uniform time interval\mbox{-} based} partition scheme. It results to more improvement in space requirement in case of EOC-ULC.

\begin{figure}
    \centering
    \begin{tabular}{cc}
      \includegraphics[scale=0.2]{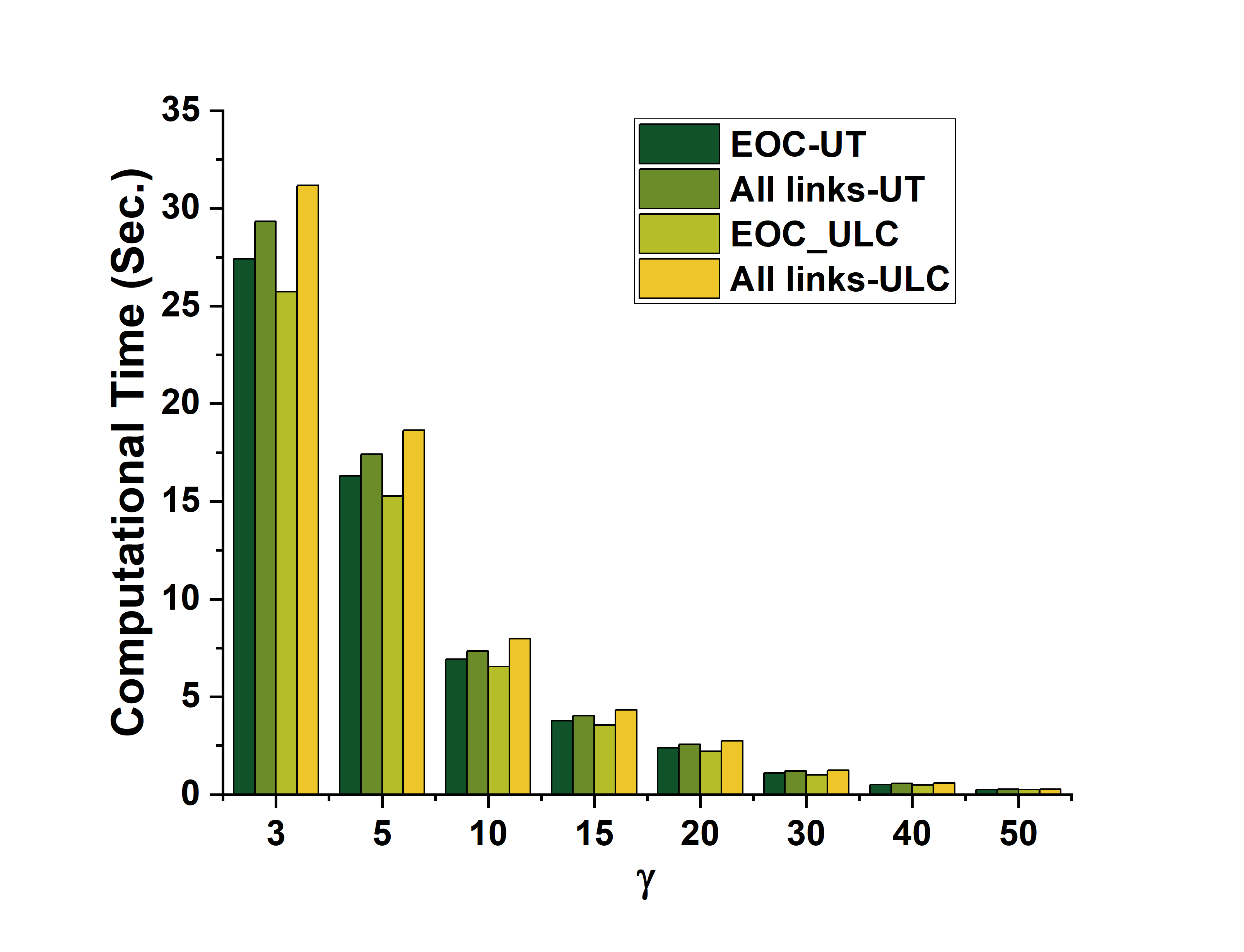}   &  \includegraphics[scale=0.2]{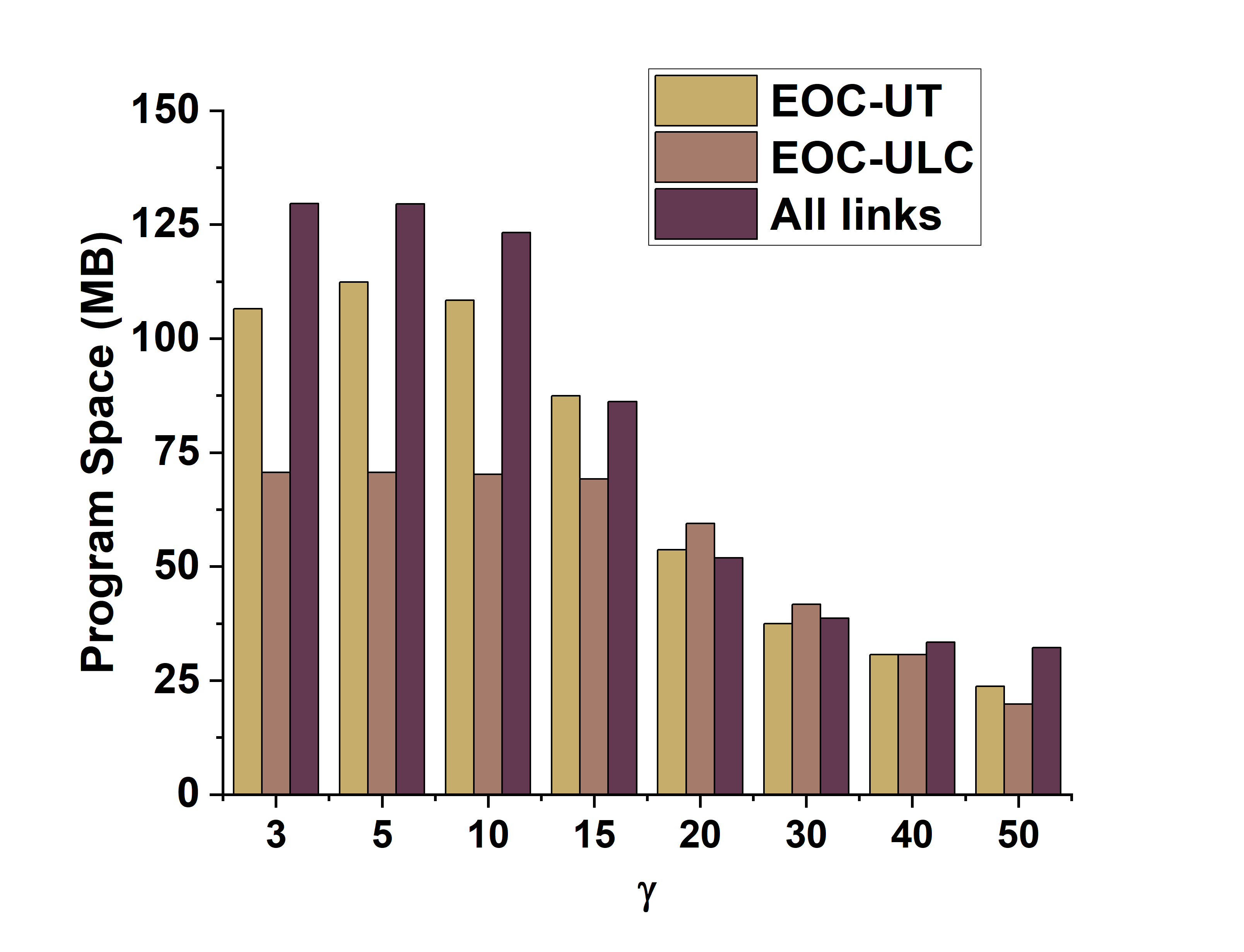} \\
        (a) Fixed - $\Delta=43200$ time & (b) Fixed - $\Delta=43200$ space   \\
    \includegraphics[scale=0.2]{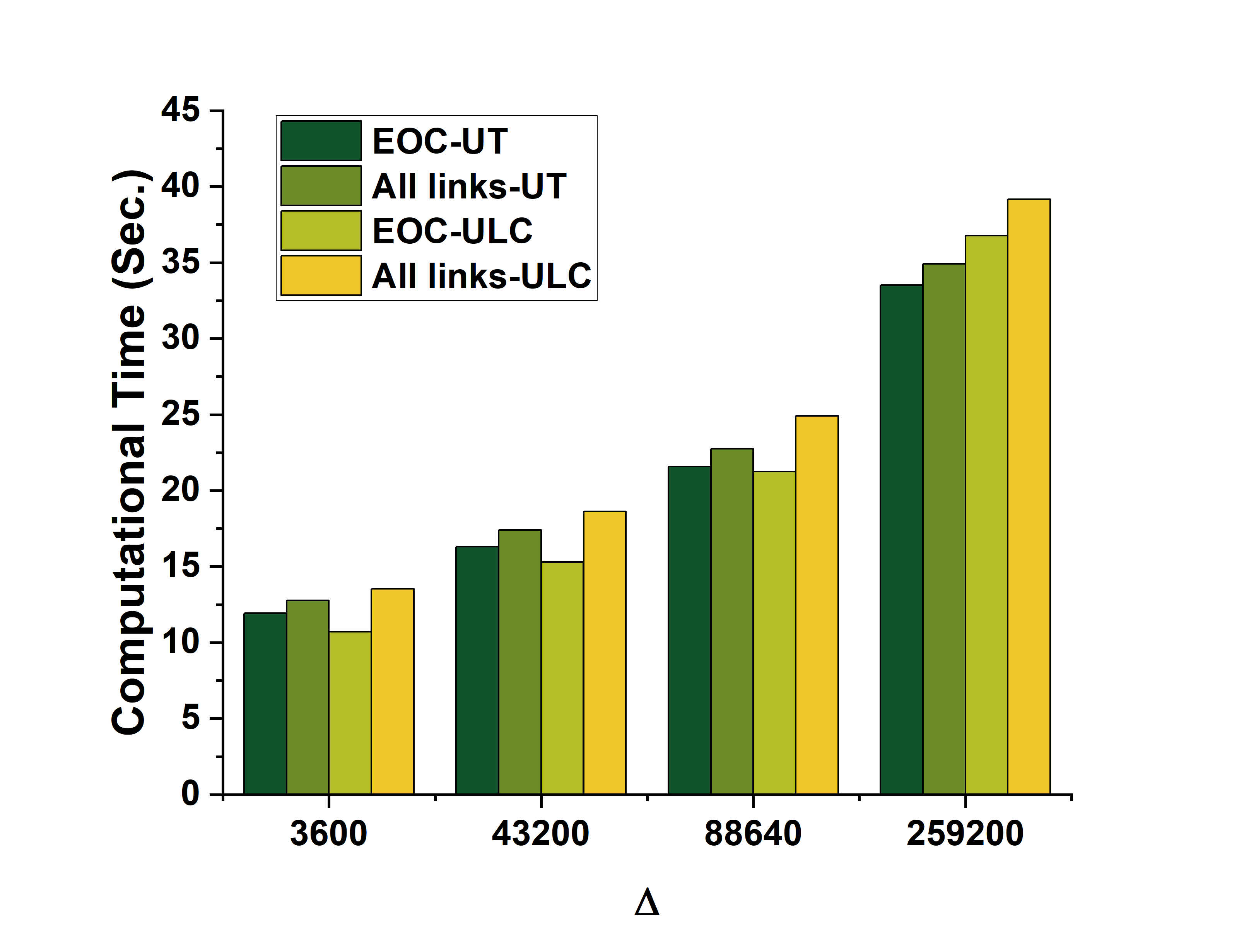}   &    \includegraphics[scale=0.2]{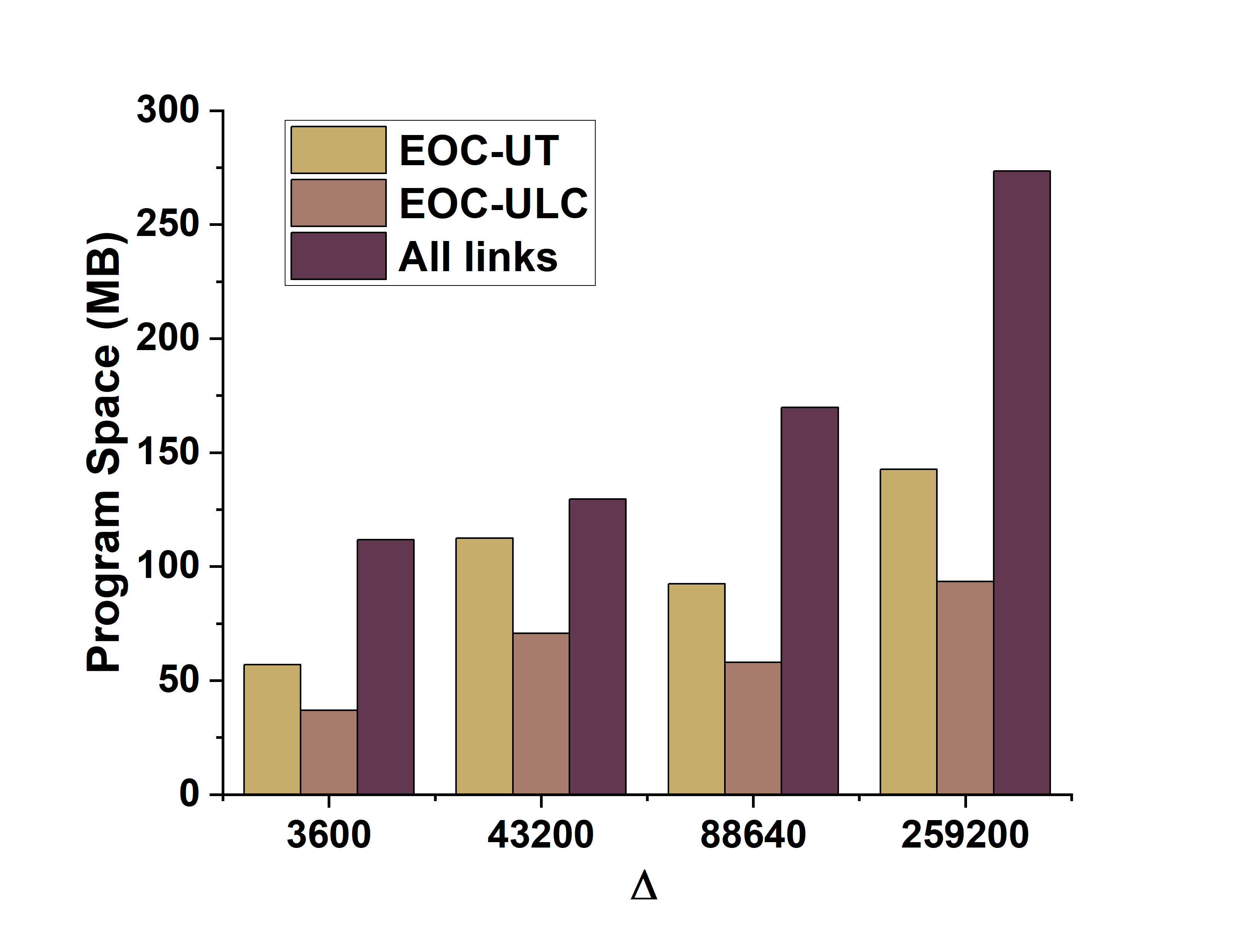}\\
          (c) Fixed - $\gamma=5$ time & (d) Fixed - $\gamma=5$ space  
    \end{tabular}
    \caption{ Results for Computational Time and Space for College Message dataset}
    \label{fig:Computational_Time_collegemsg}
\end{figure}

\par Figure \ref{fig:Computational_Time_infectious} c, \ref{fig:Computational_Time_ht09contact} c, and  \ref{fig:Computational_Time_collegemsg} c show the plots for computational time requirement for fixed $\gamma$ and varying $\Delta$. From these plots, it is observed that with the increase of $\Delta$, in most of the instances the computational time increases. It is also observed that the partition based schemes leads an improvement over the `all links'. In particular, between the EOC-UT and EOC-ULC schemes, in most of the cases, the computational time requirement by EOC-UT scheme is less. The effect is achieved due to the similar reason as described earlier for the varying $\gamma$ scenario. Similar to computational time, the space requirement grows with the increment of $\Delta$  and the improvement is more in EOC-ULC partition scheme compared to EOC-UT (Figure \ref{fig:Computational_Time_infectious} d, \ref{fig:Computational_Time_ht09contact} d, and  \ref{fig:Computational_Time_collegemsg} d). The rate of improvement also increases for large value of $\Delta$ wit a fixed $\gamma$.


\par In all the datasets, the computational time is in second. Hence, we experiment on a large dataset AS180 to observe improvement in the computational time and space in a huge scale. As the number of links in each time(day) is huge (Figure \ref{fig:dataset_description} d), we partition the entire dataset into three partitions with uniform link count based partition scheme. The result is reported in Table \ref{tab:as180_fixed_gamma}. As there data is collected on each day and 180 unique value of $t$ is present. We set the $\Delta$ as 5, 10, 15, 20. For all the cases, the improvement compared to with and without partition is significant for both time and space. 

\begin{table}
    \centering
    \begin{tabular}{|c|c||c|c||c|c|}
    \hline
         &  & \multicolumn{2}{c||}{computational time (hour)} & \multicolumn{2}{c|}{Program Space (GB)} \\
         \hline
delta & gamma & with partition & without partition & with partition & without partition \\ \hline
5 &	3 &	7.52344	& 20.42738 &	77.42066 &	133.41058 \\ \hline
10 &	3 &	8.24277	& 19.5564 &	55.2932	& 92.70591 \\ \hline
15 &	3 &	9.41666	& 17.89813 &	47.7809	& 79.45884 \\ \hline
20 &	3 &	9.52877	& 17.79293 &	44.84441 &	74.34061 \\ \hline
20 &   5 &	8.86788	& $>$ 96 & 50.54501 &	$>$ 256 GB \\ \hline

    \end{tabular}
    \caption{Results for Computational Time and Space for AS180 dataset}
    \label{tab:as180_fixed_gamma}
\end{table}

\begin{figure}
    \centering
    \begin{tabular}{cc}
      \includegraphics[scale=0.2]{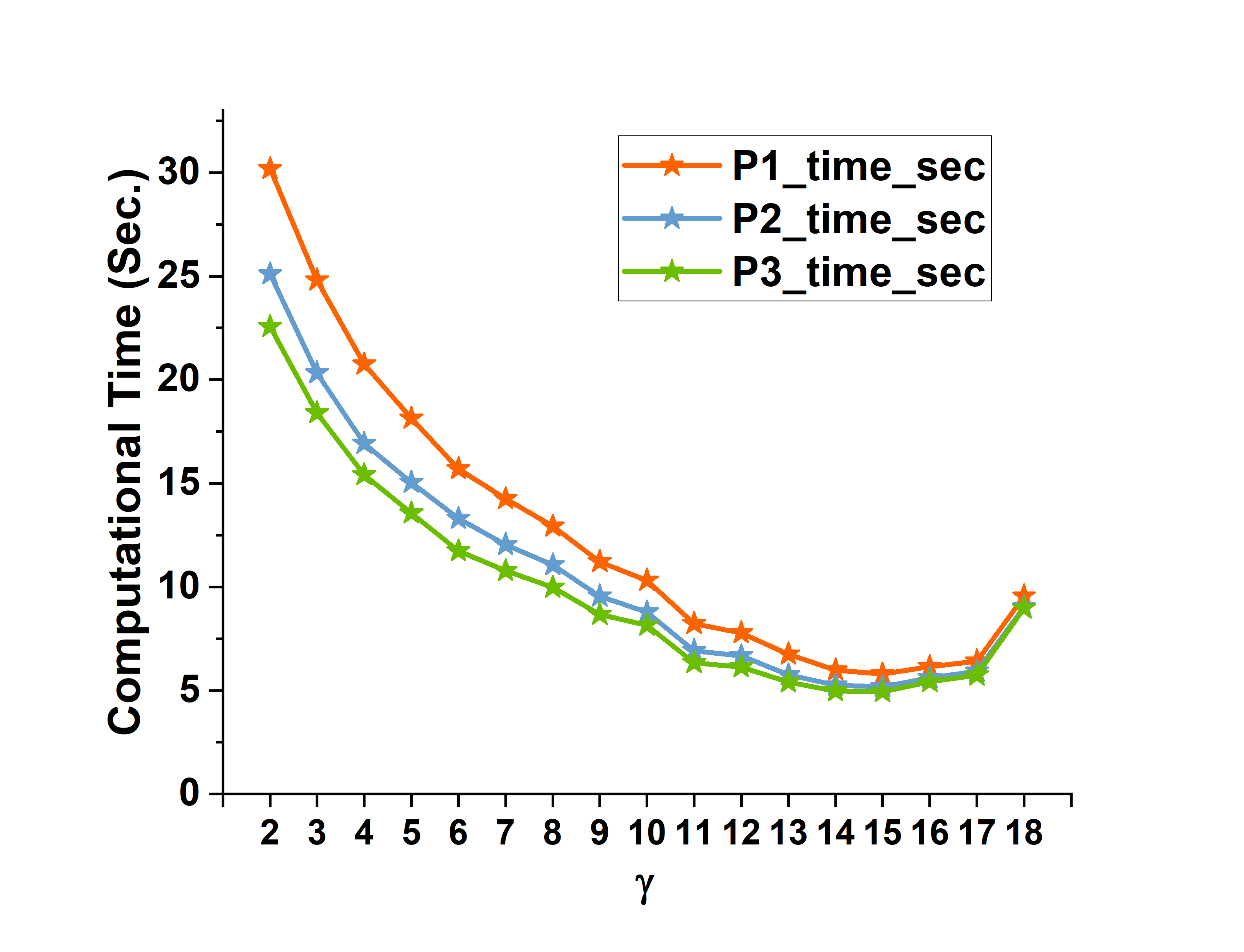}   &  \includegraphics[scale=0.2]{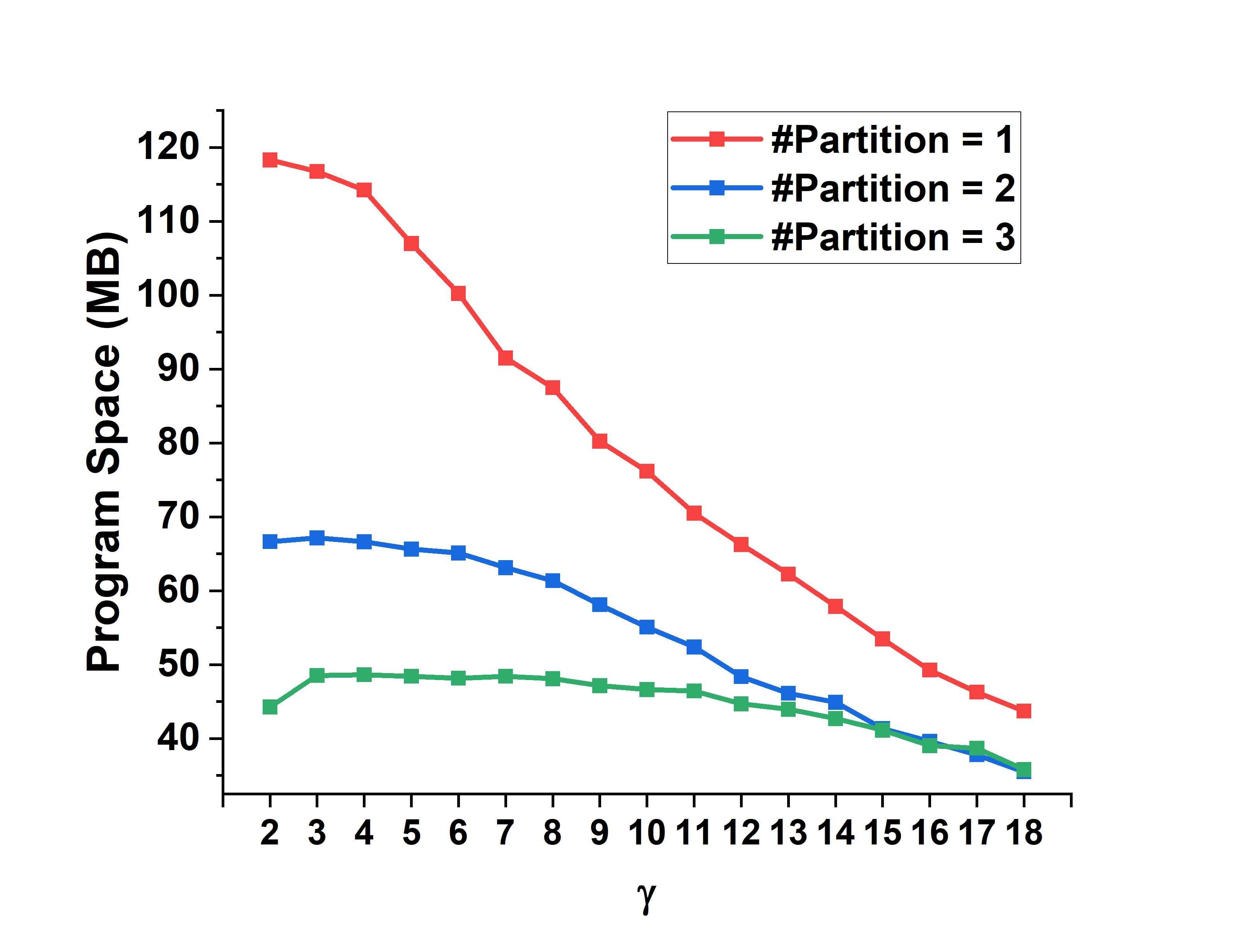} \\
        (a) Fixed - $\Delta=360$ time & (b) Fixed - $\Delta=360$ space  \\
    \includegraphics[scale=0.2]{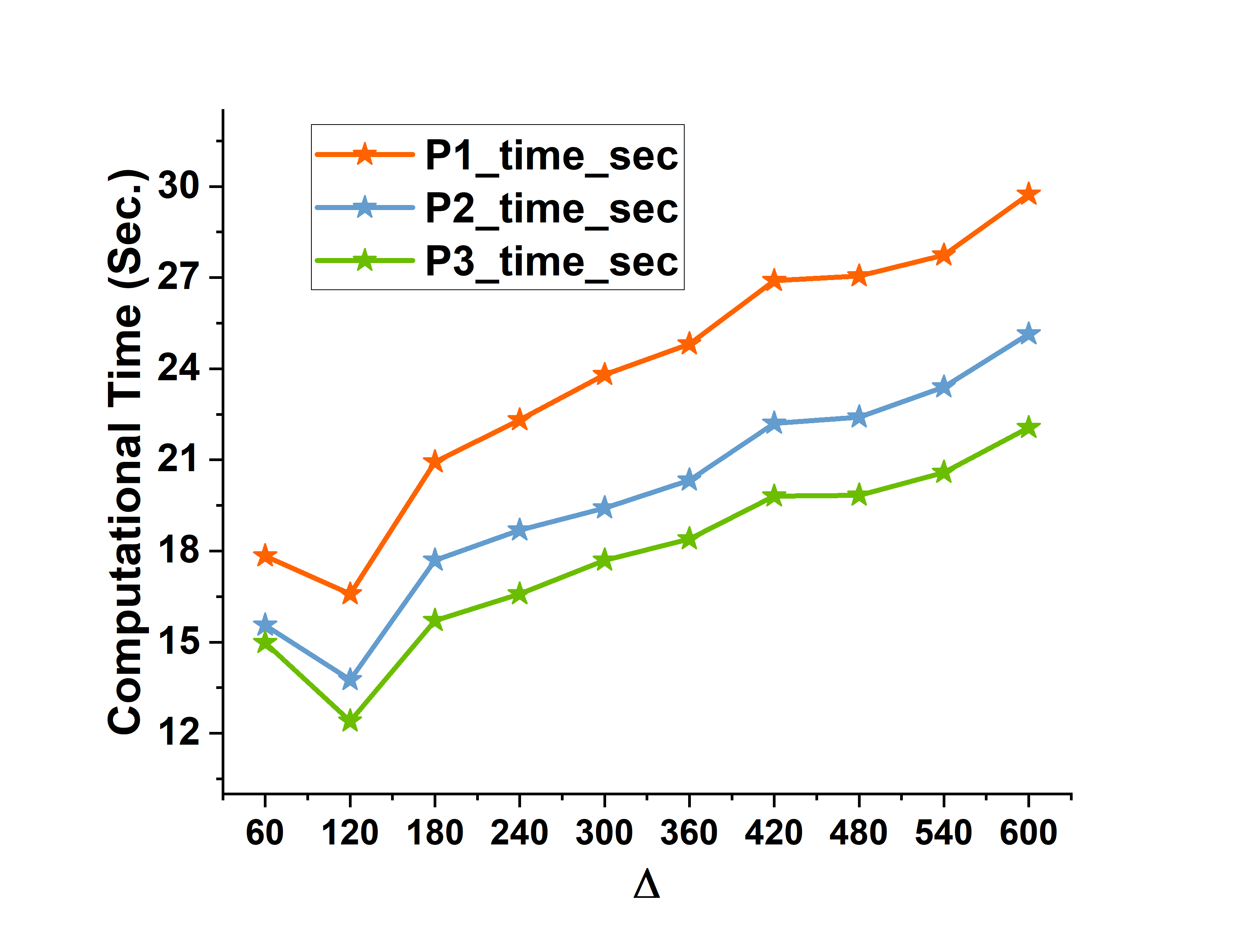}   &  \includegraphics[scale=0.2]{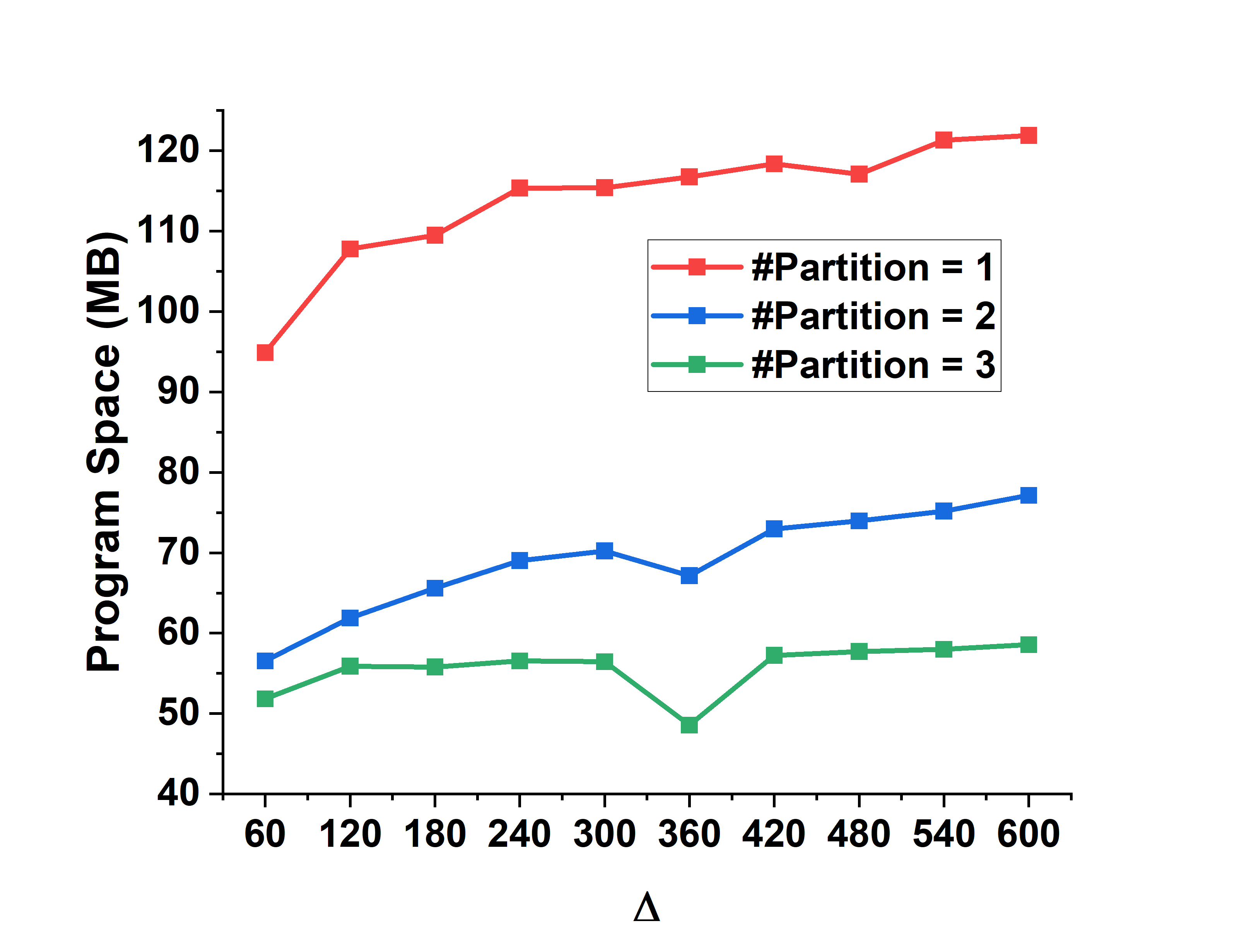}\\
           (c) Fixed - $\gamma=3$ time & (d) Fixed - $\gamma=3$ space 
    \end{tabular}
    \caption{Computational Time and Space for Hypertext dataset based on number of partitions}
    \label{fig:partition_ht09contact}
\end{figure}

\subsubsection{Effect on the Number of Partitions}
To understand the effect on the number of partitions, we select Hypertext dataset for the experiment. We can observe there exists three clear partitions in the dataset from Figure \ref{fig:dataset_description} b. We choose without partition as the number partition = 1, EOC-ULC with number of partition =2. For the number of partition = 3 case, we choose out $T_0, T_1, T_2, T_3$, we select $T_1 = 70800$, and $T_2 = 160000$. $T_0$ and $T_3$ are the start and end time of the temporal network, respectively. This follows the condition of lemma \ref{lemma:1}, for $\Delta = 360$ and $\gamma = 3$, which exploits the maximum effective improvement by partition. We report the result for computational time and space, with varying $\Delta$ and $\gamma$ in Figure \ref{fig:partition_ht09contact}. The effect of partiton clearly scales up the performance. However, for large $\gamma$, i.e. $\gamma \longrightarrow \Delta$ (the data is captured in each 20 second ). Hence, the improvement is not observed, for $\gamma \longrightarrow 18$. Whereas, the improvement is more significant in case of increasing $\Delta$.

\section{Conclusion and Future Research Directions} \label{Sec:CFD}
In this paper, we have introduced the Maximal $(\Delta,\gamma)$\mbox{-}Clique Updation Problem and proposed the `Edge on Clique' framework. We have established the correctness of the proposed methodology and analyzed it to obtain its time and space requirement. Also, we conduct an extensive set of experiments with four publicly available temporal network datasets. Experimental results show that the proposed methodology can be used to update the $(\Delta,\gamma)$\mbox{-}cliques efficiently. Now, one immediate future research direction is to consider the probabilistic nature of the links and modify the proposed methodology, so that it can be used in probabilistic setting as well.
\bibliographystyle{spbasic}      
\bibliography{paper}   

\begin{thebibliography}{56}
\providecommand{\natexlab}[1]{#1}
\providecommand{\url}[1]{{#1}}
\providecommand{\urlprefix}{URL }
\expandafter\ifx\csname urlstyle\endcsname\relax
  \providecommand{\doi}[1]{DOI~\discretionary{}{}{}#1}\else
  \providecommand{\doi}{DOI~\discretionary{}{}{}\begingroup
  \urlstyle{rm}\Url}\fi
\providecommand{\eprint}[2][]{\url{#2}}

\bibitem[{Akiba et~al.(2013)Akiba, Iwata, and Yoshida}]{akiba2013linear}
Akiba T, Iwata Y, Yoshida Y (2013) Linear-time enumeration of maximal
  k-edge-connected subgraphs in large networks by random contraction. In:
  Proceedings of the 22nd ACM international conference on Information \&
  Knowledge Management, ACM, pp 909--918

\bibitem[{Akkoyunlu(1973)}]{akkoyunlu1973enumeration}
Akkoyunlu EA (1973) The enumeration of maximal cliques of large graphs. SIAM
  Journal on Computing 2(1):1--6

\bibitem[{Almeida and Carvalho(2014)}]{almeida2014two}
Almeida MT, Carvalho FD (2014) Two-phase heuristics for the k-club problem.
  Computers \& Operations Research 52:94--104

\bibitem[{Banerjee and Pal(2019)}]{banerjee2019enumeration}
Banerjee S, Pal B (2019) On the enumeration of maximal ($\delta$,
  $\gamma$)-cliques of a temporal network. In: Proceedings of the ACM India
  Joint International Conference on Data Science and Management of Data, ACM,
  pp 112--120

\bibitem[{Banerjee and Pal(2020)}]{banerjee2020first}
Banerjee S, Pal B (2020) First stretch then shrink and bulk: A two phase
  approach for enumeration of maximal $(\delta,\gamma) $-cliques of a temporal
  network. arXiv preprint arXiv:200405935

\bibitem[{Barab{\'a}si et~al.(2016)}]{barabasi2016network}
Barab{\'a}si AL, et~al. (2016) Network science. Cambridge university press

\bibitem[{Basu et~al.(2014)Basu, Yu, Bar-Noy, and Rawitz}]{basu2014sample}
Basu P, Yu F, Bar-Noy A, Rawitz D (2014) To sample or to smash? estimating
  reachability in large time-varying graphs. In: Proceedings of the 2014 SIAM
  International Conference on Data Mining, SIAM, pp 983--991

\bibitem[{Bazzi et~al.(2016)Bazzi, Porter, Williams, McDonald, Fenn, and
  Howison}]{bazzi2016community}
Bazzi M, Porter MA, Williams S, McDonald M, Fenn DJ, Howison SD (2016)
  Community detection in temporal multilayer networks, with an application to
  correlation networks. Multiscale Modeling \& Simulation 14(1):1--41

\bibitem[{Bentert et~al.(2019)Bentert, Himmel, Molter, Morik, Niedermeier, and
  Saitenmacher}]{bentert2019listing}
Bentert M, Himmel AS, Molter H, Morik M, Niedermeier R, Saitenmacher R (2019)
  Listing all maximal k-plexes in temporal graphs. Journal of Experimental
  Algorithmics (JEA) 24(1):1--27

\bibitem[{Berlowitz et~al.(2015)Berlowitz, Cohen, and
  Kimelfeld}]{berlowitz2015efficient}
Berlowitz D, Cohen S, Kimelfeld B (2015) Efficient enumeration of maximal
  k-plexes. In: Proceedings of the 2015 ACM SIGMOD International Conference on
  Management of Data, pp 431--444

\bibitem[{Bron and Kerbosch(1973)}]{bron1973algorithm}
Bron C, Kerbosch J (1973) Algorithm 457: finding all cliques of an undirected
  graph. Communications of the ACM 16(9):575--577

\bibitem[{Casteigts et~al.(2012)Casteigts, Flocchini, Quattrociocchi, and
  Santoro}]{casteigts2012time}
Casteigts A, Flocchini P, Quattrociocchi W, Santoro N (2012) Time-varying
  graphs and dynamic networks. International Journal of Parallel, Emergent and
  Distributed Systems 27(5):387--408

\bibitem[{Casteigts et~al.(2015)Casteigts, Klasing, Neggaz, and
  Peters}]{casteigts2015efficiently}
Casteigts A, Klasing R, Neggaz YM, Peters JG (2015) Efficiently testing $ t$
  -interval connectivity in dynamic graphs. In: International Conference on
  Algorithms and Complexity, Springer, pp 89--100

\bibitem[{Chen et~al.(2016)Chen, Fang, Wang, Suo, Li, and
  Ives}]{chen2016parallelizing}
Chen Q, Fang C, Wang Z, Suo B, Li Z, Ives ZG (2016) Parallelizing maximal
  clique enumeration over graph data. In: International Conference on Database
  Systems for Advanced Applications, Springer, pp 249--264

\bibitem[{Cheng et~al.(2010)Cheng, Ke, Fu, Yu, and Zhu}]{cheng2010finding}
Cheng J, Ke Y, Fu AWC, Yu JX, Zhu L (2010) Finding maximal cliques in massive
  networks by h*-graph. In: Proceedings of the 2010 ACM SIGMOD International
  Conference on Management of data, ACM, pp 447--458

\bibitem[{Cheng et~al.(2011)Cheng, Ke, Fu, Yu, and Zhu}]{cheng2011finding}
Cheng J, Ke Y, Fu AWC, Yu JX, Zhu L (2011) Finding maximal cliques in massive
  networks. ACM Transactions on Database Systems (TODS) 36(4):21

\bibitem[{Cheng et~al.(2012)Cheng, Zhu, Ke, and Chu}]{cheng2012fast}
Cheng J, Zhu L, Ke Y, Chu S (2012) Fast algorithms for maximal clique
  enumeration with limited memory. In: Proceedings of the 18th ACM SIGKDD
  international conference on Knowledge discovery and data mining, ACM, pp
  1240--1248

\bibitem[{Eppstein and Strash(2011)}]{eppstein2011listing}
Eppstein D, Strash D (2011) Listing all maximal cliques in large sparse
  real-world graphs. In: International Symposium on Experimental Algorithms,
  Springer, pp 364--375

\bibitem[{Eppstein et~al.(2013)Eppstein, L{\"o}ffler, and
  Strash}]{eppstein2013listing}
Eppstein D, L{\"o}ffler M, Strash D (2013) Listing all maximal cliques in large
  sparse real-world graphs. Journal of Experimental Algorithmics (JEA) 18:3--1

\bibitem[{Fluschnik et~al.(2020)Fluschnik, Molter, Niedermeier, Renken, and
  Zschoche}]{fluschnik2020temporal}
Fluschnik T, Molter H, Niedermeier R, Renken M, Zschoche P (2020) Temporal
  graph classes: A view through temporal separators. Theoretical Computer
  Science 806:197--218

\bibitem[{Fournet and Barrat(2014)}]{fournet2014contact}
Fournet J, Barrat A (2014) Contact patterns among high school students. PloS
  one 9(9):e107878

\bibitem[{Galimberti et~al.(2019)Galimberti, Ciaperoni, Barrat, Bonchi,
  Cattuto, and Gullo}]{galimberti2019span}
Galimberti E, Ciaperoni M, Barrat A, Bonchi F, Cattuto C, Gullo F (2019)
  Span-core decomposition for temporal networks: Algorithms and applications.
  arXiv preprint arXiv:191003645

\bibitem[{He and Chen(2015)}]{he2015fast}
He J, Chen D (2015) A fast algorithm for community detection in temporal
  network. Physica A: Statistical Mechanics and its Applications 429:87--94

\bibitem[{Holme and Saram{\"a}ki(2012)}]{holme2012temporal}
Holme P, Saram{\"a}ki J (2012) Temporal networks. Physics reports
  519(3):97--125

\bibitem[{Hou et~al.(2016)Hou, Wang, Chen, Suo, Fang, Li, and
  Ives}]{hou2016efficient}
Hou B, Wang Z, Chen Q, Suo B, Fang C, Li Z, Ives ZG (2016) Efficient maximal
  clique enumeration over graph data. Data Science and Engineering
  1(4):219--230

\bibitem[{Huang et~al.(2015)Huang, Fu, and Liu}]{huang2015minimum}
Huang S, Fu AWC, Liu R (2015) Minimum spanning trees in temporal graphs. In:
  Proceedings of the 2015 ACM SIGMOD International Conference on Management of
  Data, pp 419--430

\bibitem[{Isella et~al.(2011)Isella, Stehl{\'e}, Barrat, Cattuto, Pinton, and
  Van~den Broeck}]{isella2011s}
Isella L, Stehl{\'e} J, Barrat A, Cattuto C, Pinton JF, Van~den Broeck W (2011)
  What's in a crowd? analysis of face-to-face behavioral networks. Journal of
  theoretical biology 271(1):166--180

\bibitem[{Khaouid et~al.(2015)Khaouid, Barsky, Srinivasan, and
  Thomo}]{khaouid2015k}
Khaouid W, Barsky M, Srinivasan V, Thomo A (2015) K-core decomposition of large
  networks on a single pc. Proceedings of the VLDB Endowment 9(1):13--23

\bibitem[{Khodaverdian et~al.(2016)Khodaverdian, Weitz, Wu, and
  Yosef}]{khodaverdian2016steiner}
Khodaverdian A, Weitz B, Wu J, Yosef N (2016) Steiner network problems on
  temporal graphs. arXiv preprint arXiv:160904918

\bibitem[{Kostakos(2009)}]{kostakos2009temporal}
Kostakos V (2009) Temporal graphs. Physica A: Statistical Mechanics and its
  Applications 388(6):1007--1023

\bibitem[{Leskovec et~al.(2005)Leskovec, Kleinberg, and
  Faloutsos}]{leskovec2005graphs}
Leskovec J, Kleinberg J, Faloutsos C (2005) Graphs over time: densification
  laws, shrinking diameters and possible explanations. In: Proceedings of the
  eleventh ACM SIGKDD international conference on Knowledge discovery in data
  mining, pp 177--187

\bibitem[{Michail and Spirakis(2016)}]{michail2016traveling}
Michail O, Spirakis PG (2016) Traveling salesman problems in temporal graphs.
  Theoretical Computer Science 634:1--23

\bibitem[{Molter et~al.(2019)Molter, Niedermeier, and
  Renken}]{molter2019enumerating}
Molter H, Niedermeier R, Renken M (2019) Enumerating isolated cliques in
  temporal networks. In: International Conference on Complex Networks and Their
  Applications, Springer, pp 519--531

\bibitem[{Mukherjee et~al.(2015)Mukherjee, Xu, and
  Tirthapura}]{mukherjee2015mining}
Mukherjee AP, Xu P, Tirthapura S (2015) Mining maximal cliques from an
  uncertain graph. In: Data Engineering (ICDE), 2015 IEEE 31st International
  Conference on, IEEE, pp 243--254

\bibitem[{Mukherjee et~al.(2017)Mukherjee, Xu, and
  Tirthapura}]{mukherjee2017enumeration}
Mukherjee AP, Xu P, Tirthapura S (2017) Enumeration of maximal cliques from an
  uncertain graph. IEEE Transactions on Knowledge and Data Engineering
  29(3):543--555

\bibitem[{Musial et~al.(2013)Musial, Budka, and
  Juszczyszyn}]{musial2013creation}
Musial K, Budka M, Juszczyszyn K (2013) Creation and growth of online social
  network. World Wide Web 16(4):421--447

\bibitem[{Panzarasa et~al.(2009)Panzarasa, Opsahl, and
  Carley}]{panzarasa2009patterns}
Panzarasa P, Opsahl T, Carley KM (2009) Patterns and dynamics of users'
  behavior and interaction: Network analysis of an online community. Journal of
  the American Society for Information Science and Technology 60(5):911--932

\bibitem[{Rossetti and Cazabet(2018)}]{rossetti2018community}
Rossetti G, Cazabet R (2018) Community discovery in dynamic networks: a survey.
  ACM Computing Surveys (CSUR) 51(2):1--37

\bibitem[{Rossi et~al.(2014)Rossi, Gleich, Gebremedhin, and
  Patwary}]{rossi2014fast}
Rossi RA, Gleich DF, Gebremedhin AH, Patwary MMA (2014) Fast maximum clique
  algorithms for large graphs. In: Proceedings of the 23rd International
  Conference on World Wide Web, ACM, pp 365--366

\bibitem[{Rossi et~al.(2015)Rossi, Gleich, and Gebremedhin}]{rossi2015parallel}
Rossi RA, Gleich DF, Gebremedhin AH (2015) Parallel maximum clique algorithms
  with applications to network analysis. SIAM Journal on Scientific Computing
  37(5):C589--C616

\bibitem[{Schmidt et~al.(2009)Schmidt, Samatova, Thomas, and
  Park}]{schmidt2009scalable}
Schmidt MC, Samatova NF, Thomas K, Park BH (2009) A scalable, parallel
  algorithm for maximal clique enumeration. Journal of Parallel and Distributed
  Computing 69(4):417--428

\bibitem[{Sun and Han(2013)}]{sun2013mining}
Sun Y, Han J (2013) Mining heterogeneous information networks: a structural
  analysis approach. Acm Sigkdd Explorations Newsletter 14(2):20--28

\bibitem[{Viard et~al.(2015)Viard, Latapy, and Magnien}]{viard2015revealing}
Viard J, Latapy M, Magnien C (2015) Revealing contact patterns among
  high-school students using maximal cliques in link streams. In: Advances in
  Social Networks Analysis and Mining (ASONAM), 2015 IEEE/ACM International
  Conference on, IEEE, pp 1517--1522

\bibitem[{Viard et~al.(2016)Viard, Latapy, and Magnien}]{viard2016computing}
Viard T, Latapy M, Magnien C (2016) Computing maximal cliques in link streams.
  Theoretical Computer Science 609:245--252

\bibitem[{Viard et~al.(2018)Viard, Magnien, and Latapy}]{viard2018enumerating}
Viard T, Magnien C, Latapy M (2018) Enumerating maximal cliques in link streams
  with durations. Information Processing Letters 133:44--48

\bibitem[{Whitbeck et~al.(2012)Whitbeck, Dias~de Amorim, Conan, and
  Guillaume}]{whitbeck2012temporal}
Whitbeck J, Dias~de Amorim M, Conan V, Guillaume JL (2012) Temporal
  reachability graphs. In: Proceedings of the 18th annual international
  conference on Mobile computing and networking, pp 377--388

\bibitem[{Wildemann et~al.(2015)Wildemann, Rudolf, and
  Paradies}]{wildemann2015time}
Wildemann M, Rudolf M, Paradies M (2015) The time has come: Traversal and
  reachability in time-varying graphs. In: Biomedical Data Management and Graph
  Online Querying, Springer, pp 169--183

\bibitem[{Wu et~al.(2014)Wu, Cheng, Huang, Ke, Lu, and Xu}]{wu2014path}
Wu H, Cheng J, Huang S, Ke Y, Lu Y, Xu Y (2014) Path problems in temporal
  graphs. Proceedings of the VLDB Endowment 7(9):721--732

\bibitem[{Wu et~al.(2015)Wu, Cheng, Lu, Ke, Huang, Yan, and Wu}]{wu2015core}
Wu H, Cheng J, Lu Y, Ke Y, Huang Y, Yan D, Wu H (2015) Core decomposition in
  large temporal graphs. In: 2015 IEEE International Conference on Big Data
  (Big Data), IEEE, pp 649--658

\bibitem[{Wu et~al.(2016)Wu, Huang, Cheng, Li, and Ke}]{wu2016reachability}
Wu H, Huang Y, Cheng J, Li J, Ke Y (2016) Reachability and time-based path
  queries in temporal graphs. In: 2016 IEEE 32nd International Conference on
  Data Engineering (ICDE), IEEE, pp 145--156

\bibitem[{Xiang et~al.(2013)Xiang, Guo, and Aboulnaga}]{xiang2013scalable}
Xiang J, Guo C, Aboulnaga A (2013) Scalable maximum clique computation using
  mapreduce. In: 2013 IEEE 29th International Conference on Data Engineering
  (ICDE), IEEE, pp 74--85

\bibitem[{Xu(2013)}]{xu2013topological}
Xu J (2013) Topological structure and analysis of interconnection networks,
  vol~7. Springer Science \& Business Media

\bibitem[{Xu and Xing(2015)}]{xu2015network}
Xu Z, Xing K (2015) Network reachability analysis on temporally varying
  interaction networks. In: International Conference on Wireless Algorithms,
  Systems, and Applications, Springer, pp 654--663

\bibitem[{Zhai et~al.(2016)Zhai, Haraguchi, Okubo, and Tomita}]{zhai2016fast}
Zhai H, Haraguchi M, Okubo Y, Tomita E (2016) A fast and complete algorithm for
  enumerating pseudo-cliques in large graphs. International Journal of Data
  Science and Analytics 2(3-4):145--158

\bibitem[{Zou et~al.(2010)Zou, Li, Gao, and Zhang}]{zou2010finding}
Zou Z, Li J, Gao H, Zhang S (2010) Finding top-k maximal cliques in an
  uncertain graph

\bibitem[{Zschoche et~al.(2020)Zschoche, Fluschnik, Molter, and
  Niedermeier}]{zschoche2020complexity}
Zschoche P, Fluschnik T, Molter H, Niedermeier R (2020) The complexity of
  finding small separators in temporal graphs. Journal of Computer and System
  Sciences 107:72--92

\end{thebibliography}


\end{document}